\documentclass{article}

\usepackage{arxiv}

\usepackage[utf8]{inputenc} 
\usepackage[T1]{fontenc}    
\usepackage{hyperref}       
\usepackage{url}            
\usepackage{booktabs}       
\usepackage{amsfonts}       
\usepackage{nicefrac}       
\usepackage{microtype}      
\usepackage{lipsum}		
\usepackage{graphicx}
\usepackage{natbib}
\usepackage{doi}
\usepackage[labelfont=bf]{caption}
\usepackage{bm}
\usepackage{graphicx}
\usepackage[space]{grffile}
\usepackage{latexsym}
\usepackage{textcomp}
\usepackage{longtable}
\usepackage{tabulary}
\usepackage{booktabs,array,multirow}
\usepackage{amsfonts,amsmath,amssymb}
\usepackage{natbib}
\usepackage{url}
\usepackage{hyperref}
\hypersetup{colorlinks=false,pdfborder={0 0 0}}
\usepackage{etoolbox}
\usepackage{mathtools}
\usepackage{dsfont}

\usepackage{siunitx}

\usepackage{caption}
\usepackage{subcaption}

\DeclareMathOperator*{\dif}{\mathrm{d} \!}
\DeclareMathOperator*{\vech}{vech}
\DeclareMathOperator*{\vect}{vec}
\DeclareMathOperator*{\svec}{svec}
\DeclareMathOperator*{\diag}{diag}
\DeclareMathOperator*{\tr}{Tr}

\DeclareMathOperator*{\argmin}{arg\,min}

\DeclareMathOperator{\arcsinh}{arcsinh}

\newtheorem{theorem}{Theorem}[section]
\newtheorem{lemma}[theorem]{Lemma}
\newtheorem{proposition}[theorem]{Proposition}

\newtheorem{definition}[theorem]{Definition}
\newtheorem{remark}{Remark}

\allowdisplaybreaks

\newenvironment{proof}[1][Proof]{\begin{trivlist}
\item[\hskip \labelsep {\bfseries #1}]}{\end{trivlist}}

\newcommand{\qed}{\nobreak \ifvmode \relax \else
      \ifdim\lastskip<1.5em \hskip-\lastskip
      \hskip1.5em plus0em minus0.5em \fi \nobreak
      \vrule height0.75em width0.5em depth0.25em\fi}

\makeatletter
\newcommand{\myitem}[1]{%
\item[#1]\protected@edef\@currentlabel{#1}%
} 
\makeatother

\makeatletter
\DeclareRobustCommand\widecheck[1]{{\mathpalette\@widecheck{#1}}}
\def\@widecheck#1#2{%
    \setbox\z@\hbox{\m@th$#1#2$}%
    \setbox\tw@\hbox{\m@th$#1%
       \widehat{%
          \vrule\@width\z@\@height\ht\z@
          \vrule\@height\z@\@width\wd\z@}$}%
    \dp\tw@-\ht\z@
    \@tempdima\ht\z@ \advance\@tempdima2\ht\tw@ \divide\@tempdima\thr@@
    \setbox\tw@\hbox{%
       \raise\@tempdima\hbox{\scalebox{1}[-1]{\lower\@tempdima\box
\tw@}}}%
    {\ooalign{\box\tw@ \cr \box\z@}}}
\makeatother

\title{Strang splitting estimator for nonlinear multivariate stochastic differential equations with Pearson-type multiplicative noise}

\author{ \href{https://orcid.org/0000-0002-8890-421X}{\includegraphics[scale=0.06]{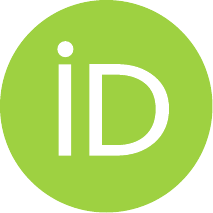}\hspace{1mm} Predrag ~Pilipović}\\
	Department of Mathematical Sciences\\
	University of Copenhagen\\
	2100 Copenhagen, Denmark \\
    \texttt{predrag@math.ku.dk} \\
    Bielefeld Graduate School of Economics and Management\\
    University of Bielefeld\\
    33501 Bielefeld, Germany\\
	\texttt{predrag.pilipovic@uni-bielefeld.de} \\
	\And
	Adeline ~Samson \\
	Univ. Grenoble Alpes\\
	CNRS, Grenoble INP, LJK\\
	38000 Grenoble, France\\
	\texttt{adeline.leclercq-samson@univ-grenoble-alpes.fr} \\
	\And
	\href{https://orcid.org/0000-0002-1998-2783}{\includegraphics[scale=0.06]{orcid.pdf}\hspace{1mm} Susanne ~Ditlevsen} \\
    Department of Mathematical Sciences\\
	University of Copenhagen\\
	2100 Copenhagen, Denmark \\
	\texttt{susanne@math.ku.dk}\\
}

\date{}

\renewcommand{\shorttitle}{Strang estimator for nonlinear multivariate SDEs with Pearson-type noise}

\hypersetup{
pdftitle={},
pdfsubject={math.ST},
pdfauthor={Predrag ~Pilipović, Adeline ~Samson, Susanne ~Ditlevsen},
pdfkeywords={First keyword, Second keyword, More},
}

\begin{document}
\maketitle

\begin{abstract}
Multivariate Pearson diffusions are characterized by a linear drift and a diffusion matrix that is quadratic in the state variables. We derive closed-form expressions for the mean and covariance matrix of this class using matrix exponential integrals, and extend this framework to a broader class of nonlinear diffusions with Pearson-type multiplicative noise. The main contribution is a new parameter estimator for these nonlinear multiplicative models based on Strang splitting, which decomposes the stochastic system into a deterministic nonlinear ordinary differential equation and a multivariate Pearson diffusion. The estimator is constructed by composing their respective flows and applying a Gaussian transition approximation with exact moments from the Pearson component. We prove that the estimator is consistent and asymptotically efficient. We also introduce a new model within this class, the Student Kramers oscillator, and prove existence and uniqueness of the strong solution and of an invariant measure. We evaluate the estimator through simulation studies on this oscillator and on the multivariate Wright-Fisher diffusion from population genetics, where it outperforms the Euler-Maruyama, Gaussian approximation, and local linearization estimators. We conclude with an application to Greenland ice core data using the Student Kramers oscillator.
\end{abstract}


\section{Introduction}

Models based on multivariate nonlinear stochastic differential equations (SDEs) are often formulated with additive noise, primarily due to the resulting simplifications in statistical inference. However, in many applications, the underlying mechanisms naturally give rise to state-dependent (multiplicative) noise. This is particularly evident in continuous diffusion approximations of large discrete systems derived from Poisson processes. Representative examples include infectious disease models, predator–prey dynamics, and genetic population models, where the diffusion term inherently depends on the system state.

State-dependent noise also arises in models designed to capture heavy-tailed behavior or in processes constrained to a subset of the state space. In such settings, additive noise formulations may fail to adequately represent key distributional features. The motivation for the present work stems from the inadequate fit of an additive-noise diffusion model applied to ice core data in \cite{pilipovic2024SecondOrder}. In particular, the observed data exhibit pronounced heavy-tailed behavior that cannot be captured within an additive noise framework, thereby motivating the development of models with state-dependent noise.

The Pearson diffusions are a versatile class of tractable one-dimensional diffusion models applied across various disciplines due to their rich statistical properties and ease of parameter estimation \citep{FormanSorensen2008}. These models are particularly advantageous because moments and conditional moments are explicitly available, making them an attractive choice for statistical inference. Pearson diffusions encompass a wide range of stationary distributions, namely the family of Pearson distributions, including light- and heavy-tailed distributions, and with different state spaces. The Ornstein-Uhlenbeck process (OU), the Cox-Ingersoll-Ross process (CIR), the t-diffusion and the Jacobi diffusion are special cases of Pearson diffusions, with Gaussian, non-central chi-squared, skew $t$ and Beta invariant distributions, respectively. 

A Pearson diffusion is the solution to a one-dimensional SDE of the form \citep{FormanSorensen2008}
\begin{equation}
\dif X_t = -\lambda (X_t - m)\dif t + \sqrt{\alpha X_t^2 + \beta X_t + \gamma} \dif W_t,
\label{eq:Pearson1D}
\end{equation}
where $\lambda > 0$ determines the speed of mean reversion, $m$ is the mean of the invariant distribution and is in the interior of the state space, and $\alpha$, $\beta$, and $\gamma$ are such that the square root is well-defined when $X_t$ is in the state space. The parameters of \eqref{eq:Pearson1D} are $\bm{\theta} = \{\lambda, m, \alpha, \beta, \gamma\}$, where $\alpha$, $\beta$, and $\gamma$ shape the state space and the invariant distribution.

Parameter estimation in Pearson diffusions is facilitated by optimal martingale estimating functions based on exact conditional moments derived from eigenfunctions of the infinitesimal generator \citep{BibbySorensen1995, KesslerSørensen1999, FormanSorensen2008, MS12}. This approach simplifies statistical inference and yields consistent estimators even under low-frequency sampling, often encountered in empirical work. Other estimation methods based on conditional moments, such as the generalized method of moments, quasi-likelihood, and non-linear weighted least squares, are also straightforward to implement for Pearson diffusions, further highlighting their flexibility and statistical tractability.

Although Pearson diffusions are tractable and flexible, progress in extending these models to higher dimensions has been limited. A probabilistic generalization was developed by \cite{filipovic2016polynomial}, who proved existence and uniqueness of polynomial diffusions, a class of models that we refer to as multivariate Pearson diffusions. However, statistical methods for the multivariate case are still underdeveloped. \cite{Leonenko2012} used a high-order spectral approximation of the Fokker–Planck equation for Pearson diffusions and proposed to extend this to the multivariate setting by a quadratic form without linear or constant terms for the squared diffusion term. A significant gap remains in the statistical literature regarding the generalization of these models to multivariate settings, particularly to more general forms with nonlinear drift.

As we will show, extending to multivariate Pearson diffusions retains many of the analytical and statistical advantages of the univariate model, particularly the availability of closed‑form expressions for first, second, and conditional moments. A specific Pearson-type invariant distribution can be prescribed for each coordinate of the multi-dimensional process. For example, a process can be defined on the $d$-dimensional hypercube $[0,1]^d$, with each coordinate having a Beta invariant distribution. Alternatively, combinations of Pearson-type diffusion terms can enter in different coordinates of the quadratic diffusion matrix. Moreover, the quadratic diffusion matrix may be hypoelliptic, meaning it can be singular while the diffusion still admits a smooth density. Finally, allowing for a nonlinear drift opens new directions for modeling complex systems in fields such as finance, biology, and physics, capturing dependencies and interactions among multiple variables and providing a more realistic representation of the underlying processes.

The goal of this study is to develop efficient and accurate methods for parameter estimation in nonlinear SDEs with Pearson-type noise. Traditional estimation techniques often fall short in these scenarios due to the complexity of the noise structure and the nonlinearity in the drift function. Our goal is to create a robust framework for parameter estimation that can be applied to a wide range of models, including hypoelliptic diffusions. We have recently shown the potential of pseudo-likelihood approaches using splitting techniques, which approximate the likelihood while maintaining high accuracy \citep{pilipovic2024,pilipovic2024SecondOrder, DitlevsenDitlevsen2023}, however, only for additive noise. By building on these techniques, we aim to provide a more effective and versatile method for parameter estimation in complex systems with state-dependent noise structures, allowing for non-Gaussian and possibly heavy-tailed or bounded distributions.

Parameter estimation for SDEs with state-dependent noise has been studied before. The state-dependence of the diffusion coefficient complicates inference and has motivated a range of specialized methods. Most methods transform the original SDE into an SDE with additive noise at the cost of a more complicated drift function, then applying a method that requires additive noise. These transformation-based methods can be applied when the Lamperti transform is available \citep{AitSahalia2008}. For a review of likelihood-based parameter estimation for SDEs with additive noise, see \citet{pilipovic2024,pilipovic2024SecondOrder} and the references therein. 

Martingale estimating functions have emerged as a powerful tool for parameter estimation in SDEs, in particular, for one-dimensional SDEs with state-dependent noise. They exploit the underlying process' martingale property to construct consistent and asymptotically normal estimators. \cite{KesslerSørensen1999} introduced optimal martingale estimating functions based on eigenfunctions of the generator. It mimics the score function when the likelihood is not available or simplifies estimation when the likelihood is complex (as in the CIR model) and ensures high efficiency, particularly in high-frequency sampling scenarios often encountered in financial data. \cite{Sørensen2007} further demonstrated that these optimal estimating functions are comparable to maximum likelihood estimators under high-frequency asymptotics, offering a simpler alternative.

Quasi-likelihood methods, which approximate the likelihood using a tractable form, have also been applied to SDEs with non-constant noise. \citet{Kessler1997} proposed an estimator that approximates the unknown transition density of a diffusion process by a Gaussian density. This Gaussian approximation (GA) is constructed using the true conditional mean and covariance when available, or suitable approximations derived from the infinitesimal generator of the diffusion process.

Building on Kessler's foundational work, \citet{UchidaYoshida2012} extended the GA estimator to multivariate elliptic diffusions. They developed an adaptive-type contrast estimator that achieves a central limit theorem under the same design condition proposed by Kessler, namely $N h^p \to 0$, for $p \geq 2$.

More recently, \citet{iguchi2023.3} extended GA-type estimation to hypoelliptic SDEs, which pose additional challenges due to the degeneracy of the diffusion coefficient. They proposed a modified contrast estimator that achieves consistency and asymptotic normality  under the design condition $nh^p \to 0$ for $p \geq 2$. This matches the general condition proved in the elliptic setting, but not previously attained in the hypoelliptic case.

\citet{Gloter2006} developed a contrast function based on the Euler–Maruyama discretization for integrated diffusion processes, focusing on asymptotic properties as the sampling interval tends to zero and the sample size tends to infinity. To address the ill-conditioned nature of the contrast arising from the EM discretization, \citet{Gloter2006} proposed to use rough equations of the SDE and to recover the unobserved components through finite-difference approximations. However, this approach introduces bias and requires correction factors, which in turn affect the estimator’s variance.

Like \citet{Kessler1997}, \citet{Hurn2013} developed a quasi-maximum likelihood procedure for parameter estimation in multi-dimensional diffusions, where the transition density is approximated by a multivariate Gaussian distribution. Unlike \citet{Kessler1997}, \citet{Hurn2013} focused specifically on systems with affine drift and diffusion functions, where both the drift and the quadratic diffusion matrix are linear functions of the state variable. For such affine models, they derive closed-form expressions for the first and second moments of the true transition density, which are exact due to the affine structure. This explicit calculation of moments enables the construction of a Gaussian approximation that yields consistent parameter estimates even when the true transition density is misspecified, thereby providing robustness in parameter estimation. For non-affine models, \citet{Hurn2013} showed that numerical methods can still accurately compute the required moments of the transition distribution, so that Gaussian approximations achieve high computational precision in integral evaluations.

A key technical ingredient in our approach is the computation of conditional means and covariances for linear SDEs and for the linear moment systems obtained from polynomial (multivariate Pearson) diffusions. For linear SDEs with drift and diffusion coefficients affine in the state, explicit formulas for the mean and covariance in terms of matrix exponentials can be obtained via the block–matrix method of \citet{VanLoan1978}. Earlier work by \citet{JimenezOzaki2002,JimenezOzaki2003} and \citet{CarbonellJimenez2008} expressed these moments as linear combinations of several higher-dimensional matrix exponentials, while \citet{Jimenez2013} later derived simplified formulas in which the mean and covariance can be computed from a single, smaller block matrix exponential. These results provide the canonical reference for covariance computation in affine (linear) diffusion models and correspond to the linear special case of the multivariate Pearson framework considered in this paper.

Recent contributions have continued to build on these foundational methods. For instance, \cite{gloter2020,gloter2021} introduced adaptive and non-adaptive methods for hypoelliptic diffusion models, demonstrating asymptotic normality in complete observation regimes. They used higher-order It\^o–Taylor expansions to introduce additional Gaussian noise into the smooth coordinates, accompanied by higher-order mean approximations for the rough coordinates. These contributions have refined the conditions for asymptotic normality of the resulting estimators, extending the applicability of such methods to a broader class of models.

In this paper, we extend the one-dimensional Pearson diffusion to a multivariate setting, where the quadratic diffusion matrix has entries that are quadratic functions of the state variable. This generalization builds on the concept of affine diffusions as used in \citet{Hurn2013}, where first and second moments can be explicitly computed despite the unknown transition density. While \citet{Hurn2013} focused on the affine (linear) case, encompassing certain types of Pearson diffusions such as the OU and CIR processes, our approach fully generalizes to all Pearson diffusions. We further extend this framework from linear multivariate Pearson diffusions to models with nonlinear drift. We show that the Wright–Fisher diffusion, widely used in genetics, fits naturally into our framework, and we define the Student Kramers oscillator as a nonlinear hypoelliptic multivariate Pearson diffusion, proving that its solution exists, is unique, and possesses an invariant measure. 

We propose a novel method for parameter estimation in nonlinear multivariate Pearson diffusions that combines Strang splitting with a Gaussian transition density approximation. Specifically, we split the nonlinear drift into a multivariate Pearson diffusion and a nonlinear ODE, then approximate the transition density of the multivariate Pearson component by a Gaussian distribution with exact mean and covariance, computed via matrix-exponential formulas for linear SDEs. We solve the nonlinear ODE and compose the solutions of the split subsystems to obtain the Strang splitting approximation, upon which our Strang splitting (SS) estimator is based. We prove that the estimator is consistent and asymptotically normal, and we illustrate its performance in detailed simulation studies for the Wright–Fisher diffusion and the Student Kramers oscillator.

The main contributions of this paper are:
\begin{enumerate}
\item We define the class of multivariate Pearson diffusions, also known as polynomial diffusions with linear drift and quadratic diffusion matrix \citep{filipovic2016polynomial}. We derive closed-form expressions for the mean and covariance matrix of their transition distribution using established matrix-exponential methods for linear SDEs in the spirit of \citet{VanLoan1978, CarbonellJimenez2008}.
\item We introduce the Student Kramers oscillator as a specific example of a nonlinear hypoelliptic multivariate Pearson diffusion, and we show that its solution exists, is unique, and admits an invariant density. This example is also used in our simulation study.
\item We develop a new parameter estimation method for nonlinear SDEs with a multivariate Pearson-type diffusion matrix based on the Strang splitting scheme. Specifically, we split the nonlinear drift into a multivariate Pearson diffusion (the linear component) and a nonlinear ODE. The transition density of the multivariate Pearson component is approximated by a Gaussian distribution with exact first two moments, as derived in our moment formulas.
\item We prove that the proposed Strang splitting estimator is consistent and asymptotically normal.
\item We conduct a simulation study comparing our proposed estimator with the EM, GA, and LL methods and find that our estimator achieves substantially smaller estimation error.
\item We illustrate the applicability of the Strang splitting estimator and the Student Kramers oscillator by applying them to high-resolution Greenland ice-core data.
\end{enumerate}

The structure of the paper is as follows. Section \ref{sec:ProblemSetup} outlines the problem setup and introduces the multivariate Wright–Fisher diffusion and the Student Kramers oscillator as motivating examples. Section \ref{sec:MainResults} derives the first two moments of the multivariate Pearson diffusion models, presents the Strang splitting scheme and the resulting estimator, and states the asymptotic results. In Section \ref{sec:Simulations}, we report a simulation study comparing the performance of the proposed estimator with the EM, GA, and LL methods. Section \ref{sec:Greenland} applies the methodology and the Student Kramers oscillator model to Greenland ice-core data. Section \ref{sec:TechnicalDetails} contains further technical details and a description of the competing estimators used in the study, while all proofs of the main results are collected in the Supplementary Material.

\textbf{\emph{Notation}.} Capital bold letters denote random vectors, vector-valued functions, and matrices, lowercase bold letters denote deterministic vectors. $\|\cdot\|$ denotes the $L^2$ vector norm in $\mathbb{R}^d$. Superscript $(i)$ on a vector denotes the $i$-th component. The double subscript $ij$ on a matrix denotes the component in the $i$-th row and $j$-th column. $\tr(\cdot)$ returns the trace, $\det(\cdot)$ the determinant, $\top$ the transpose, $\vect$ vectorization, and $\diag(\cdot)$ either the diagonal matrix formed from a vector or the vector of diagonal entries of a matrix. The Kronecker product and the sum of two matrices are $\otimes$ and $\oplus$, respectively. $\mathbf{I}_d$ denotes the $d$-dimensional identity matrix, $\bm{0}_{d\times d}$ a $d\times d$-dimensional zero matrix, $\mathbf{1}$ a vector of ones. We denote by $[a_i]_{i=1}^d$ a vector with coordinates $a_i$, and by $[b_{ij}]_{i,j=1}^d$ a matrix with coordinates $b_{ij}$. For a function $\bm{g}$, $\partial_{x^{(i)}} \bm{g}(\mathbf{x})$ denotes the partial derivative with respect to $x^{(i)}$ and $\partial_{x^{(i)}x^{(j)}}^2 \bm{g}(\mathbf{x})$ with respect to $x^{(i)}$ and $x^{(j)}$. The nabla operator $\nabla_{\mathbf{x}}$ denotes the gradient vector of real-valued $g$ with respect to $\mathbf{x}$; $\nabla_{\mathbf{x}} g(\mathbf{x}) = [\partial_{x^{(i)}} g(\mathbf{x})]_{i=1}^d$. For function $\mathbf{F}: \mathbb{R}^d \to \mathbb{R}^d$, $D$ denotes the Jacobian matrix $D\mathbf{F}(\mathbf{x}) = [\partial_{x^{(i)}} F^{(j)}(\mathbf{x})]_{i,j=1}^d$. $\mathbf{R}$ represents a vector (or matrix) valued function defined on $(0,1) \times \mathbb{R}^d$ (or $(0,1) \times \mathbb{R}^{d\times d}$), such that, for some constant $C$, $\|\mathbf{R}(a,\mathbf{x})\| < aC(1+\|\mathbf{x}\|)^C$ for all $a, \mathbf{x}$; $R$ refers to a scalar function. $\overline{A}$ denotes the closure of an open set $A$. $\delta_{ij}$ denotes the Kronecker delta. We write $\mathbf{X} \sim \mathcal{N}_d(\bm{\mu}, \bm{\Omega})$ for  $\mathbf{X}$ a $d$-dimensional Gaussian with mean $\bm{\mu}$ and covariance $\bm{\Omega}$; in general, $\sim$ denotes distribution. $\propto$ denotes proportionality up to a constant, and $\approx$ denotes a numerical approximation. We write $\xrightarrow{\mathbb{P}}, \xrightarrow{d}, \xrightarrow{a.s.}$ for convergence in probability, in distribution, and almost surely, respectively. 

\section{Problem setup} \label{sec:ProblemSetup}

Let $\overline{\Theta} = \overline{\Theta}_{\theta_1} \times \overline{\Theta}_{\theta_2}$, where $\Theta_{\theta_1} \subset \mathbb{R}^r$ and $\Theta_{\theta_2} \subset \mathbb{R}^s$ are open, convex, and bounded. We consider the $d$-dimensional SDE parameterized by $\bm{\theta} = (\bm{\theta}^{(1)}, \bm{\theta}^{(2)}) \in \Theta$
\begin{equation}
    \dif \mathbf{X}_t = \mathbf{F}(\mathbf{X}_t; \bm{\theta}^{(1)}) \dif t + \bm{\Sigma}(\mathbf{X}_t; \bm{\theta}^{(2)}) \dif \mathbf{W}_t, 
    \qquad \mathbf{X}_0 = \mathbf{x}_0 \in \mathcal{X} \subset \mathbb{R}^d.
    \label{eq:SDE}
\end{equation}
For each $\bm{\theta} \in \Theta$, let $\mathbf{X} = (\mathbf{X}_t)_{t \geq 0}$ be a strong solution of \eqref{eq:SDE} (assumed to exist) on a complete filtered probability space $(\Omega, \mathcal{F}, (\mathcal{F}_t)_{t \geq 0}, \mathbb{P}_{\bm{\theta}})$, where $\mathbf{W} = (\mathbf{W}_t)_{t \geq 0}$ is a $d$-dimensional Wiener process adapted to $(\mathcal{F}_t)_{t \geq 0}$ under $\mathbb{P}_{\bm{\theta}}$.

We assume that
\begin{equation}
    [\bm{\Sigma}\bm{\Sigma}^\top(\mathbf{x}; \bm{\theta}^{(2)})]_{ij} \coloneqq[\bm{\Sigma}(\mathbf{x}; \bm{\theta}^{(2)})\bm{\Sigma}(\mathbf{x}; \bm{\theta}^{(2)})^\top]_{ij}
    = \mathbf{x}^\top \bm{\alpha}^{ij} \mathbf{x} + \mathbf{x}^\top \bm{\beta}^{ij} + \gamma^{ij},
    \qquad i,j = 1,\dots,d,
    \label{eq:SigmaSigmaT}
\end{equation}
where $\bm{\alpha}^{ij} \in \mathbb{R}^{d \times d}$, $\bm{\beta}^{ij} \in \mathbb{R}^d$, and $\gamma^{ij} \in \mathbb{R}$ satisfy $\bm{\alpha}^{ij} = \bm{\alpha}^{ji}$, $\bm{\beta}^{ij} = \bm{\beta}^{ji}$, and $\gamma^{ij} = \gamma^{ji}$$.$ Without loss of generality, matrices $\bm{\alpha}^{ij}$ can be chosen to be symmetric. Most entries of these matrices, vectors, and constants are fixed by the model, for example $0$, $1$, or $-1$, so the number of free diffusion parameters is typically much smaller than the total number of coefficients (see Remark~\ref{rmrk:DimensionalitySigma} below).

Since any square root of $\bm{\Sigma}\bm{\Sigma}^\top(\mathbf{x}; \bm{\theta}^{(2)})$ induces the same distribution, $\bm{\Sigma}$ is identifiable only up to equivalence classes, and we therefore work directly with $\bm{\Sigma}\bm{\Sigma}^\top(\mathbf{x}; \bm{\theta}^{(2)})$. For $d=1$, Eq.~\eqref{eq:SigmaSigmaT} reduces to Eq.~\eqref{eq:Pearson1D}. 
The diffusion parameter $\bm{\theta}^{(2)} \in \mathbb{R}^s$ consists of the free coefficients in \eqref{eq:SigmaSigmaT}. 

We write the drift as a sum of a linear and a nonlinear part 
and rewrite \eqref{eq:SDE} as
\begin{align}
     \dif \mathbf{X}_t
     = \mathbf{A}(\bm{\theta}^{(1)})(\mathbf{X}_t - \mathbf{b}(\bm{\theta}^{(1)})) \dif t
     + \mathbf{N}(\mathbf{X}_t; \bm{\theta}^{(1)}) \dif t
     + \bm{\Sigma}(\mathbf{X}_t; \bm{\theta}^{(2)}) \dif \mathbf{W}_t.
     \label{eq:SDEsplitted}
\end{align}

We denote the true parameter by $\bm{\theta}_0 = (\bm{\theta}_0^{(1)}, \bm{\theta}_0^{(2)}) \in \Theta$. 
We write $\mathbf{A}_0$, $\mathbf{b}_0$, $\mathbf{N}_0(\mathbf{x})$, and $\bm{\Sigma}\bm{\Sigma}_0^\top(\mathbf{x})$ instead of $\mathbf{A}(\bm{\theta}_0^{(1)})$, $\mathbf{b}(\bm{\theta}_0^{(1)})$, $\mathbf{N}(\mathbf{x}; \bm{\theta}_0^{(1)})$, and $\bm{\Sigma}\bm{\Sigma}^\top(\mathbf{x}; \bm{\theta}_0^{(2)})$. For a general parameter $\bm{\theta}$, we often simply write $\mathbf{A}$, $\mathbf{b}$, $\mathbf{N}(\mathbf{x})$, and $\bm{\Sigma}\bm{\Sigma}^\top(\mathbf{x})$. We also suppress the parameter in the expectation and write $\mathbb{E}$ instead of $\mathbb{E}_{\bm{\theta}}$ when there is no ambiguity.

\subsection{Examples}

We now present two examples of SDEs that fit the above framework. 
The first is a multivariate Wright--Fisher diffusion with Pearson-type multiplicative noise corresponding to Jacobi diffusions, restricted to a hypercube. 
The associated multivariate Pearson diffusion has an invariant generalized multivariate beta distribution. 
The second example is an extension of the Kramers oscillator, which we refer to as the Student Kramers oscillator (the Duffing oscillator with a skew $t$-distribution type noise), a hypoelliptic model with two quasi-stationary states.

\subsubsection{Multivariate Wright--Fisher diffusion} \label{sec:WF}

\begin{figure}
    \centering
    \includegraphics[width=\textwidth]{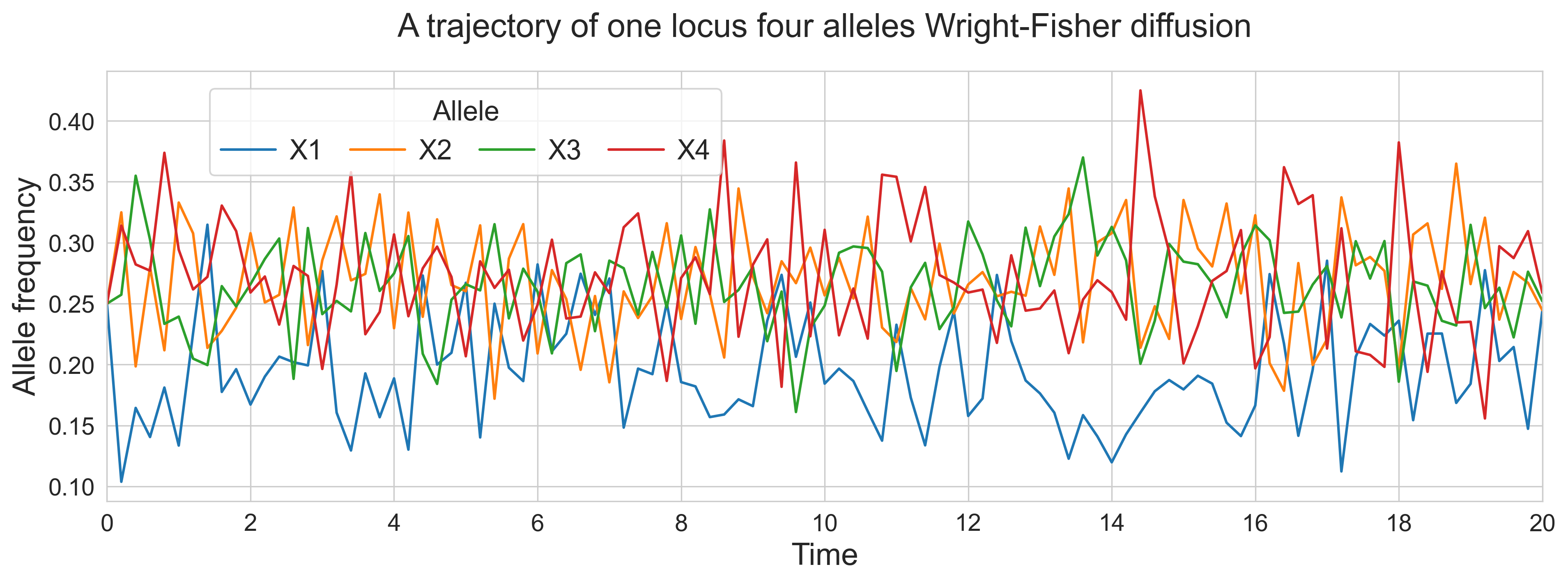}
    \caption{{\bf A trajectory of the Wright-Fisher diffusion.} Simulated from \eqref{eq:WF_final_vec} using the Euler-Maruyama scheme with $h_{\text{sim}} = 0.0001$ and $N_{\text{sim}} = 200000$ and then subsampled such that $h = 0.2$ and $N = 100$. Starting point is $\mathbf{X}_0 = (1/4, 1/4, 1/4, 1/4)$. Parameters are set to $\tau^{0} = 10$, $p_{11}^0 = 0.2$, $p_{12}^0 = 0.3$, $p_{13}^0 = 0.15$, $p_{21}^0 = 0.2$, $p_{22}^0 = 0.05$, $p_{23}^0 = 0.35$, $p_{31}^0 = 0.25$, $p_{32}^0 = 0.6$, $p_{33}^0 = 0.1$, $p_{41}^0 = 0.15$, $p_{42}^0 = 0.1$, $p_{43}^0 = 0.1$, $q_1^0 = 25$, $q_2^0 = 40$, $q_3^0 = 30$ and $q_4^0 = 10$.}
    \label{fig:WF_traj_plot}
\end{figure}

The multivariate Wright--Fisher diffusion describes the evolution of allele frequencies across $L$ loci \citep{Aurell2019}. Let $M_l$ be the number of alleles at locus $l$ and $M = \sum_{l=1}^L M_l$. The state vector $\mathbf{X}_t = (\bm{X}_t^{(1)\top}, \ldots, \bm{X}_t^{(L)\top})^\top \in \mathcal{X} \subset \mathbb{R}^M$ collects the allele frequencies at all loci, where $\mathcal{X}$ is the product of simplices defined by $X_t^{(li)} \geq 0$ and $\sum_{j=1}^{M_l} X_t^{(lj)} = 1$ for each $l$. The dynamics are governed by
\begin{equation}
    \dif \mathbf{X}_t = \bm{\mu}(\mathbf{X}_t) \dif t + \bm{\Sigma}\bm{\Sigma}^\top(\mathbf{X}_t) \nabla V(\mathbf{X}_t) \dif t + \bm{\Sigma}(\mathbf{X}_t) \dif \mathbf{W}_t, 
    \label{eq:SDEWF}
\end{equation}
where $\bm{\mu}(\mathbf{X})$ is the mutation drift, $V(\mathbf{X})$ is a quadratic selection potential so that $\bm{\Sigma}\bm{\Sigma}^\top(\mathbf{X}) \nabla V(\mathbf{X})$ represents selection. The diffusion matrix is block-diagonal
\begin{equation}
    \bm{\Sigma}\bm{\Sigma}^\top_l(\bm{X}^{(l)}) = \diag(\bm{X}^{(l)}) - \bm{X}^{(l)}\bm{X}^{(l)\top},
    \label{eq:SigmaSigmaT_final_WF}
\end{equation}
with one block per locus, so each locus evolves on its simplex with Wright--Fisher-type noise. We verify that the squared diffusion matrix  \eqref{eq:SigmaSigmaT_final_WF} is of the form \eqref{eq:SigmaSigmaT}. The component form of \eqref{eq:SigmaSigmaT_final_WF} for locus $l$ is
\begin{equation}\label{eq:SigmaSigmaTl}
    [\bm{\Sigma}\bm{\Sigma}^\top_l(\bm{X}^{(l)})]_{ij} = \begin{cases} X^{(li)}(1 - X^{(li)}) & \text{if } i = j \\ -X^{(li)}X^{(lj)} & \text{if } i \neq j \end{cases}, \qquad i,j = 1,\ldots,M_l.
\end{equation}
Matching \eqref{eq:SigmaSigmaTl} against \eqref{eq:SigmaSigmaT}, we identify the coefficients $\bm{\alpha}^{(l)ij} \in \mathbb{R}^{M_l \times M_l}$ from
\begin{equation*}
    \sum_{m = 1}^{M_l}\sum_{n=1}^{M_l} \alpha^{(l)ij}_{mn} x^{(lm)} x^{(ln)} = -x^{(li)}x^{(lj)}, \qquad \forall \mathbf{x} \in \mathcal{X}.
\end{equation*}
This gives $\alpha^{(l)ij}_{ij} + \alpha^{(l)ij}_{ji} = -1$ for $i \neq j$ and $\alpha^{(l)ii}_{ii} = -1$ for $i = j$. A natural symmetric choice is
\begin{equation} \label{eq:alpha_l_ij}
    \bm{\alpha}^{(l)ij} = \left[-\frac{1}{2}(\delta_{mi}\delta_{nj} + \delta_{ni}\delta_{mj})\right]_{m,n=1}^{M_l}.
\end{equation}
The linear coefficient $\bm{\beta}^{(l)ij}$ is zero everywhere except at $m = i = j$, where it equals $1$,
\begin{equation}\label{eq:beta_l_ij}
    \bm{\beta}^{(l)ij} = [\delta_{mi}\delta_{mj}]_{m=1}^{M_l} = \delta_{ij}\mathbf{1},
\end{equation}
and $\gamma^{(l)ij} = 0$. 

The drift of \eqref{eq:SDEWF} is a sum of two functions, $\bm{\mu}(\mathbf{X})$ and $\bm{\Sigma}\bm{\Sigma}^\top (\mathbf{X})\nabla V(\mathbf{X})$. The function $\bm{\mu}$ represents the mutation dynamics. It is assumed that mutations occur independently at each locus. Specifically, for the $l$-th locus,
\begin{equation*}
    \bm{\mu}^{(l)}(\bm{X}^{(l)}) 
    = \frac{\tau_l}{2}\bigl(\bm{P}^{(l)\top} - \mathbf{I}\bigr)\bm{X}^{(l)},
\end{equation*}
where $\tau_l \geq 0$ is the mutation rate. $\bm{P}^{(l)}$ is the mutation  probability matrix at locus $l$ given by
\begin{equation}\label{eq:P}
    \bm{P}^{(l)} = [p^{(l)}_{ij}]_{i,j=1}^{M_l}, \qquad p^{(l)}_{ij} \geq 0, \qquad \sum_{j=1}^{M_l} p^{(l)}_{ij} = 1, \qquad i = 1,\ldots,M_l.
\end{equation}

The function $\bm{\Sigma}\bm{\Sigma}^\top (\mathbf{X})\nabla V(\mathbf{X})$ represents selection for the current allele type at the current locus. The fitness potential $V(\mathbf{X}_t)$ is explicitly constructed such that the coupled Wright-Fisher diffusion \eqref{eq:SDEWF} treats the effects of selection and mutation as independent mechanisms, ignoring their cross-effects, implying that $\bm{\Sigma}\bm{\Sigma}^\top (\mathbf{X})\nabla V(\mathbf{X})$ includes at most pairwise interactions between different loci and their allele types. Thus, $V(\mathbf{X})$ is given by
\begin{equation*}
    V(\mathbf{X}) = \mathbf{X}^\top \mathbf{q} + \frac{1}{2} \mathbf{X}^\top \mathbf{S} \mathbf{X}
\end{equation*} 
where $\mathbf{q} \in \mathbb{R}_+^M$ is the locus selection parameter, and $\mathbf{S} \in \mathbb{R}_+^{M \times M}$ is the interaction matrix, a symmetric block matrix with non-negative entries such that the blocks on the diagonal equal to zero. The gradient is then $\nabla V(\mathbf{X}) = \mathbf{q} + \mathbf{S} \mathbf{X}$.

To test our method in Section \ref{sec:WFsimulation}, we chose a simple one-locus model with four alleles, i.e., $L = 1$ and $M = 4$. Since $L = 1$, then $\mathbf{S} = \mathbf{0}$, so the final model is 
\begin{equation} \label{eq:WF_final_vec}
    \begin{aligned}
    \dif \mathbf{X}_t &= \frac{\tau}{2} (\bm{P}^{\top} - \mathbf{I}) \mathbf{X}_t\dif t + (\diag(\mathbf{X}_t) - \mathbf{X}_t\mathbf{X}_t^\top) \mathbf{q}\dif t + \bm{\Sigma}(\mathbf{X}_t) \dif \mathbf{W}_t,
\end{aligned}
\end{equation}
where $\mathbf{q} = (q_1, q_2, q_3, q_4)^\top \in \mathbb{R}_+^4$, $\bm{P}\in \mathbb{R}^{4\times 4}_+$. Thus, the unknown parameter is 
\begin{equation*}
    \bm{\theta} = (\tau, \mathbf{q}, \mathbf{P}).
\end{equation*}

Figure~\ref{fig:WF_traj_plot} shows a trajectory of this Wright--Fisher diffusion. The trajectory is simulated using the Euler--Maruyama scheme with $h_{\text{sim}} = 0.0001$ and $N_{\text{sim}} = 200000$, and then subsampled such that $h = 0.2$ and $N = 100$.

\subsubsection{The Student Kramers oscillator} \label{sec:Kramers}

\begin{figure}
    \centering
    \includegraphics[width=\textwidth]{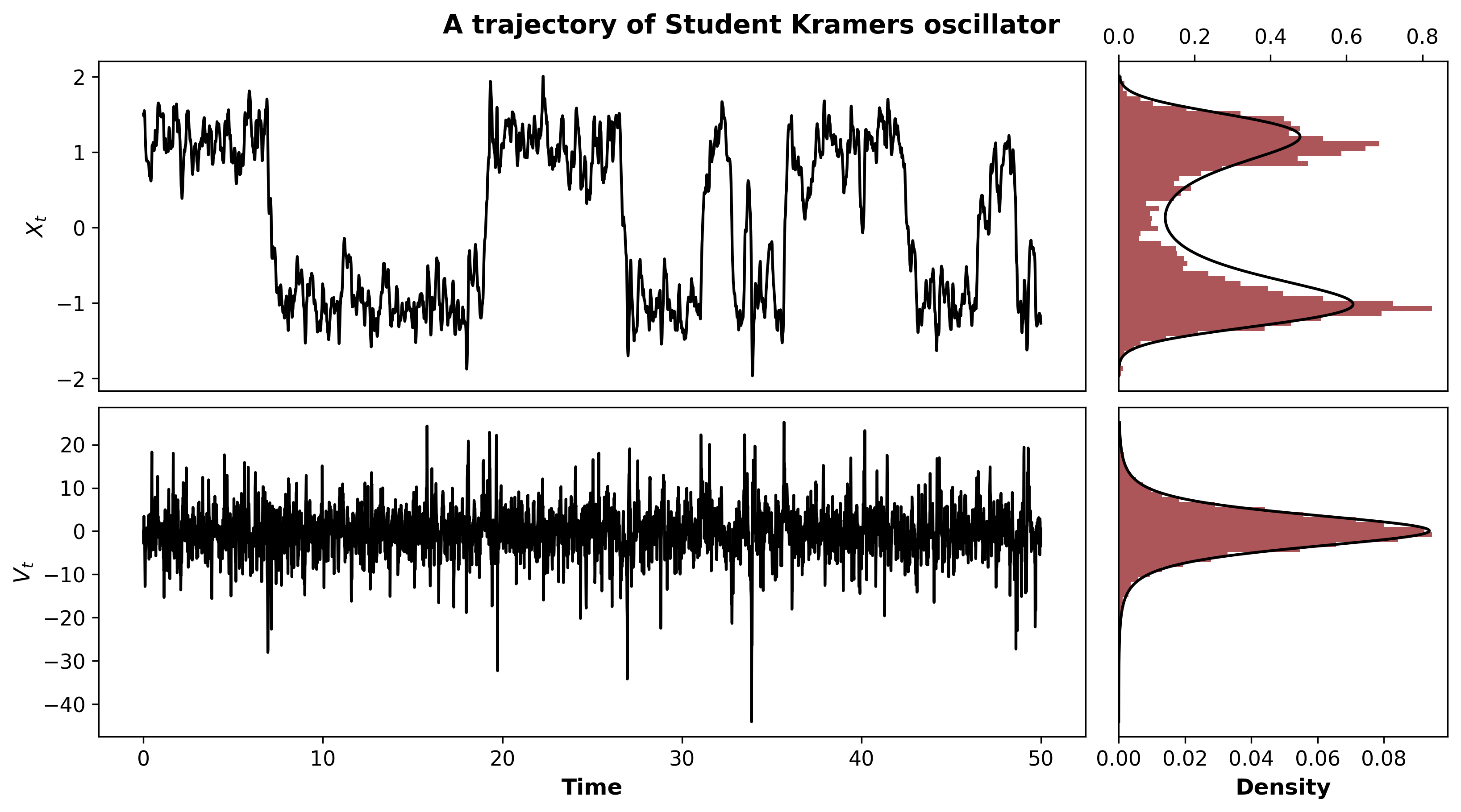}
    \caption{{\bf A trajectory of Student Kramers oscillator.} Simulated from \eqref{eq:StudentKramersSDE} using the Milstein scheme with $h_{\text{sim}} = 0.0001$ and $N_{\text{sim}} = 500000$ and then subsampled such that $h = 0.001$ and $N = 50000$. Parameters are $\eta_0 = 30$, $a_0 = -125$, $b_0 = 40$, $c_0 = 150$, $d_0 = -20$, $\alpha_0 = 20$, $\beta_0 = -8$, and $\gamma_0 = 1280.8$. \textbf{Left:} Trajectories of the individual components $X_t$ and $V_t$. \textbf{Right:} Histograms of empirical densities of $X_t$ and $V_t$ shown in red, overlined by the approximated theoretical invariant densities \eqref{eq:pi_X_fin}-\eqref{eq:pi_V_fin}. }
    \label{fig:traj_plot}
\end{figure}

The Kramers oscillator with additive noise is given by \citep{Kramers1940BrownianTransitionState}
\begin{equation} \label{eq:KramersSDE}
    \begin{aligned}
        \dif X_t &= V_t \dif t, \\
        \dif V_t &= \left(-\eta V_t + a X_t^3 + c X_t \right) \dif t + \sigma \dif W_t,
    \end{aligned}
\end{equation}
where $a<0$ and $c, \sigma > 0$, and $\eta \geq 0$. This is a stochastic damping Hamiltonian system. It characterizes the stochastic movement of a particle within a bistable potential
\begin{equation}
    U(x) = -a \frac{x^4}{4} - c \frac{x^2}{2}. \label{eq:KramersPotential}
\end{equation}
The parameter $\eta$ in \eqref{eq:KramersSDE} indicates the damping level, $c$ regulates the linear stiffness, and $a$ determines the nonlinear component of the restoring force. When $a = 0$, the equation simplifies to a damped harmonic oscillator.

In \cite{pilipovic2024SecondOrder}, we fitted the Kramers oscillator \eqref{eq:KramersSDE} to Greenland Ice Core data \citep{Rasmussen2014} to understand the abrupt temperature changes during the ice ages, known as the Dansgaard–Oeschger (DO) events. We found that this model only partially fits the data. In particular, the velocity variable $V_t$ did not adequately capture the spread in the approximated velocity. This discrepancy is likely due to the data being heavy-tailed, which cannot be captured by additive Gaussian noise. Here, we generalize model \eqref{eq:KramersSDE} to allow for more heavy-tailed intrinsic noise.

We propose a diffusion function that yields the stationary distribution to be a skew $t$-distribution if the drift function were linear. Specifically, the invariant distribution of the Pearson diffusion \eqref{eq:Pearson1D} for $\alpha > 0$, $\beta^2 - 4 \alpha \gamma < 0$, $\alpha < 2 \lambda $ is a skew $t$-distribution (also known as the Pearson Type IV distribution) with density 
\begin{equation}
\label{eq:skew_student}
\begin{aligned}
    f(x) &= C \left(1 + \frac{(x - \mu)^2}{\nu \sigma^2}\right)^{-\frac{\nu+1}{2}} \exp\left( \omega \arctan\left(\frac{x - \mu}{\sqrt{\nu\sigma^2}}\right) \right), \qquad x \in \mathbb{R},
\end{aligned}
\end{equation}
where $C$ is the normalization constant, and parameters are
\begin{equation*}
       \nu = \frac{2 \lambda}{\alpha} + 1, \quad \mu = - \frac{\beta}{2 \alpha}, \quad \nu\sigma^2 = \frac{4\alpha\gamma - \beta^2}{4\alpha^2}, \quad \omega = \frac{2\beta \lambda}{\alpha \sqrt{4\alpha\gamma - \beta^2}}.
\end{equation*}
The skew $t$-distribution is peaked around the location parameter $\mu$. The peak width is determined by the scale parameter $\sigma > 0$. The shape parameter, or degrees of freedom, $\nu$, influences the heaviness of the tails, while $\omega$ controls the skewness of the distribution. For symmetric noise ($\beta = 0$), the skewness parameter vanishes ($\omega = 0$), and the distribution reduces to a standard location-scale Student's $t$-distribution.

We also generalize the potential function $U(x)$ \eqref{eq:KramersPotential} to allow for asymmetries between the two modes of the distribution. The generalized potential function is 
\begin{equation}
    U(x) = -a \frac{x^4}{4} - b \frac{x^3}{3} - c \frac{x^2}{2} - d x. \label{eq:StudentKramersPotential}
\end{equation}
Finally, we define the Student Kramers oscillator (see Figure \ref{fig:traj_plot} for a simulated trajectory) as a solution to 
\begin{equation} \label{eq:StudentKramersSDE}
    \begin{aligned}
        \dif X_t &= V_t \dif t, \\
        \dif V_t &= \left(-\eta V_t +a X_t^3 + b X_t^2 + c X_t + d \right) \dif t + \sqrt{\alpha V_t^2 + \beta V_t + \gamma} \dif W_t,
    \end{aligned}
\end{equation}
where $\eta \geq 0$, $a < 0$, $\beta^2 < 4 \alpha \gamma$, $0<\alpha < 2 \eta$, and $X_t, V_t$ are defined on $\mathbb{R}$. These conditions are necessary to ensure a non-exploding solution. We refer to an SDE as having Student-type noise if its corresponding version with linear drift belongs to the class of Student Pearson diffusions \citep{FormanSorensen2008, Leonenko2012}. The condition $\alpha < 2 \eta$ is sufficient for an invariant density to exist. Thus, the unknown parameters are 
\begin{equation*}
    \bm{\theta}^{(1)} = (\eta, a, b, c, d), \qquad \bm{\theta}^{(2)} = (\alpha, \beta, \gamma).
\end{equation*}
The squared diffusion function of \eqref{eq:StudentKramersSDE} is 
\begin{equation*}
    \bm{\Sigma}\bm{\Sigma}^\top(x,v) = \begin{bmatrix}
        0 & 0 \\
        0 & \alpha v^2 + \beta v + \gamma
    \end{bmatrix},
\end{equation*}
so it aligns with the structure given in \eqref{eq:SigmaSigmaT}, where  $\bm{\alpha}^{11},\bm{\alpha}^{12}, \bm{\alpha}^{21}$ are zero matrices in $\mathbb{R}^{2\times 2}$,  $\bm{\beta}^{11},\bm{\beta}^{12}, \bm{\beta}^{21}$ are zero vectors in $\mathbb{R}^2$, $\gamma^{11} = \gamma^{12} = \gamma^{21} =0$ and 
\begin{equation*}
    \bm{\alpha}^{22} = \begin{bmatrix}
        0 & 0\\
        0 & \alpha
    \end{bmatrix}, \qquad \bm{\beta}^{22} = \begin{bmatrix}
        0\\
        \beta
    \end{bmatrix}, \qquad \gamma^{22} = \gamma.
\end{equation*}
The following theorem is proved in Supplementary Material \ref{sec:Proofs}. 

\begin{theorem} \label{thm:StudentKramers}
A unique, strong solution exists to the Student Kramers SDE \eqref{eq:StudentKramersSDE} for $\eta \geq 0$, $a < 0$, $\alpha > 0$, $\beta^2 - 4 \alpha \gamma < 0$. Furthermore, an invariant probability measure exists for $\alpha < 2 \eta$.
\end{theorem}

The invariant joint density $\pi(x,v)$ solves the stationary Fokker-Planck equation
\begin{equation} \label{eq:joint_fpe}
    -v \frac{\partial \pi}{\partial x} - \frac{\partial}{\partial v} \Big[ \big(-\eta v - U'(x)\big) \pi \Big] + \frac{1}{2} \frac{\partial^2}{\partial v^2} \Big[ \sigma^2(v) \pi \Big] = 0.
\end{equation}
Because $\sigma^2(v) = \alpha v^2 + \beta v + \gamma$ depends explicitly on $v$, the exact analytical solution for $\pi(x,v)$ is intractable. However, it can be approximated by decoupling the coordinates as $ \pi(x,v)\approx \pi_X(x)\pi_V(v)$, where
\begin{align} 
     \pi_X(x) &\propto  \exp\left( -\frac{2\eta - \alpha}{\gamma} U(x) \right), \label{eq:pi_X_fin}\\
    \pi_V(v) &\propto (\alpha v^2 + \beta v + \gamma)^{-\left(\frac{\eta}{\alpha} + 1\right)} \exp\left( \frac{2\beta \eta}{\alpha \sqrt{4\alpha\gamma - \beta^2}} \arctan\left( \frac{2\alpha v + \beta}{\sqrt{4\alpha\gamma - \beta^2}} \right) \right).\label{eq:pi_V_fin}
\end{align}
The justification for the above approximation is discussed in Supplementary Material \ref{sec:Proofs}.

\subsection{Assumptions} \label{sec:Assumptions}

We split the assumptions into three groups: model and well-posedness \ref{as:Diffusion}-\ref{as:F_polynomial_growth}, statistical regularity \ref{as:Ergodic}, and identifiability \ref{as:Identifiability}.

\begin{itemize}
    \myitem{(A1)} \label{as:Diffusion}
    For every $\bm{\theta} \in \overline{\Theta}$, $\bm{\Sigma}\bm{\Sigma}^\top(\mathbf{x}; \bm{\theta}^{(2)})$ is positive semidefinite on $\mathcal{X}$. If it is positive definite, the model is elliptic and all subsequent results hold directly. If it is only positive semidefinite, we additionally assume that \eqref{eq:SDE} is hypoelliptic, i.e., the solution admits a smooth density on $\mathcal{X}$.

    \myitem{(A2)} \label{as:monoton_condition}
    For all $\bm{\theta} \in \overline{\Theta}$ and sufficiently large $p \geq 1$, there exists a constant $C_{\bm{\theta}} > 0$ such that
    \begin{equation*}
        (\mathbf{x} - \mathbf{y})^\top \left(\mathbf{F}(\mathbf{x}; \bm{\theta}^{(1)}) 
        - \mathbf{F}(\mathbf{y}; \bm{\theta}^{(1)})\right)
        + \frac{2p-1}{2} \sum_{i=1}^d \| \bm{\Sigma}_{i\cdot}(\mathbf{x}; \bm{\theta}^{(2)}) 
        - \bm{\Sigma}_{i\cdot}(\mathbf{y}; \bm{\theta}^{(2)}) \|^2
        \leq C_{\bm{\theta}} \|\mathbf{x} - \mathbf{y}\|^2
    \end{equation*}
    for all $\mathbf{x}, \mathbf{y} \in \mathcal{X}$.

    \myitem{(A3)} \label{as:F_polynomial_growth}
    Function $\mathbf{F}$ is three times continuously differentiable with respect to $\mathbf{x}$ and $\bm{\theta}^{(1)}$ on $\mathcal{X} \times \overline{\Theta}_{\theta_1}$. All derivatives of $\mathbf{F}$ up to third order are of polynomial growth in $\mathbf{x}$, uniformly in $\bm{\theta}^{(1)} \in \overline{\Theta}_{\theta_1}$. In particular, there exist 
    constants $C_{\bm{\theta}^{(1)}} > 0$ and $p \geq 1$ such that
    \begin{equation*}
        \| \mathbf{F}(\mathbf{x}; \bm{\theta}^{(1)}) - \mathbf{F}(\mathbf{y}; 
        \bm{\theta}^{(1)}) \|^2
        \leq C_{\bm{\theta}^{(1)}} \bigl(1 + \|\mathbf{x}\|^{2p-2} + 
        \|\mathbf{y}\|^{2p-2}\bigr)\|\mathbf{x} - \mathbf{y}\|^2,
        \qquad \forall \mathbf{x}, \mathbf{y} \in \mathcal{X}.
    \end{equation*}
\end{itemize}

Assumptions \ref{as:monoton_condition} and \ref{as:F_polynomial_growth} together are sufficient, but not necessary, for existence and uniqueness of a strong solution to \eqref{eq:SDE} on $\mathcal{X}$; see \cite{TretyakovAndZhang}. Assumption \ref{as:monoton_condition} is stated in terms of the rows $\bm{\Sigma}_{i\cdot}$, whereas our model is parameterized through $\bm{\Sigma}\bm{\Sigma}^\top$. This makes the assumption restrictive, but it is a sufficient condition adopted for theoretical completeness. Polynomial growth of $\bm{\Sigma}\bm{\Sigma}^\top$ and its derivatives follows automatically from \eqref{eq:SigmaSigmaT}.

\begin{itemize}
    \myitem{(A4)} \label{as:Ergodic}
    Under $\bm{\theta}_0$, the solution $\mathbf{X}$ of \eqref{eq:SDE} has a unique invariant probability measure $\nu_0(\dif \mathbf{x})$ and satisfies the ergodic theorem: for any continuous $\mathbf{g}$ with polynomial growth
    \begin{equation*}
        \frac{1}{T} \int_0^T \mathbf{g}(\mathbf{X}_s; \bm{\theta}) \dif s
        \xrightarrow[T \to \infty]{a.s.}
        \int_{\mathcal{X}} \mathbf{g}(\mathbf{x}; \bm{\theta}) \dif \nu_0(\mathbf{x}).
    \end{equation*}
\end{itemize}

\begin{itemize}
    \myitem{(A5)} \label{as:Identifiability}
    If $\mathbf{F}(\mathbf{x}; \bm{\theta}^{(1)}) = \mathbf{F}(\mathbf{x}; \bm{\theta}_\star^{(1)})$ and $\bm{\Sigma}\bm{\Sigma}^\top(\mathbf{x}; \bm{\theta}^{(2)}) = \bm{\Sigma}\bm{\Sigma}^\top(\mathbf{x}; \bm{\theta}_\star^{(2)})$ for $\dif\nu_0( \mathbf{x})$-almost every $\mathbf{x} \in \mathcal{X}$, then $\bm{\theta}^{(1)} = \bm{\theta}_\star^{(1)}$ and $\bm{\theta}^{(2)} = \bm{\theta}_\star^{(2)}$.
\end{itemize}

We observe $(\mathbf{X}_{t_k})_{k=0}^N \equiv \mathbf{X}_{0:t_N}$ at equidistant times $0 = t_0 < t_1 < \cdots < t_N = T$ with step size $h = t_k - t_{k-1}$.

\section{Main results} \label{sec:MainResults}

We first introduce the notation needed for the computation of the first two moments of multivariate Pearson diffusions. Then, we define the Strang splitting scheme for a generalized class of models with nonlinear drift. Finally, we introduce the Strang splitting estimator and obtain asymptotic results.

In general, it is not possible to represent $\bm{\Sigma}\bm{\Sigma}^\top(\mathbf{x})$ in \eqref{eq:SigmaSigmaT} in matrix form without using tensors, so we vectorize $\bm{\Sigma}\bm{\Sigma}^\top(\mathbf{x})$. We use standard properties of the $\vect$ operator and the Kronecker product $\otimes$ (for details see \cite{Magnus2019}). Equation \eqref{eq:SigmaSigmaT} reads 
\begin{align*}
    [\bm{\Sigma}\bm{\Sigma}^\top(\mathbf{x})]_{ij} &= \vect(\bm{\alpha}^{ij})^\top \vect(\mathbf{x}\mathbf{x}^\top) + \bm{\beta}^{ij \top} \mathbf{x} + \gamma_{ij}.
\end{align*}
We express the vectorization of $\bm{\Sigma}\bm{\Sigma}^\top(\mathbf{x})$ as
\begin{align} \label{eq:SigmaSigmaT_vec}
    \vect(\bm{\Sigma}\bm{\Sigma}^\top(\mathbf{x})) &= \check{\bm{\alpha}} \vect(\mathbf{x}\mathbf{x}^\top) + \check{\bm{\beta}} \mathbf{x} + \check{\bm{\gamma}},
\end{align}
where we defined
\begin{equation} \label{eq:check_abc}
    \begin{aligned}
    \check{\bm{\alpha}} &\coloneqq [ \vect(\bm{\alpha}^{11})^\top,
        \dots,
        \vect(\bm{\alpha}^{1d})^\top,
        \dots,
        \vect(\bm{\alpha}^{dd})^\top]^\top \in \mathbb{R}^{d^2 \times d^2},\\
    \check{\bm{\beta}} &\coloneqq [\bm{\beta}^{11\top},
        \dots,
        \bm{\beta}^{1d\top},
        \dots,
        \bm{\beta}^{dd\top}]^\top \in \mathbb{R}^{d^2 \times d},\\
    \check{\bm{\gamma}} &\coloneqq (\gamma^{11\top},
        \dots,
        \gamma^{1d\top},
        \dots,
        \gamma^{dd\top})^\top = \vect(\bm{\gamma})^\top \in \mathbb{R}^{d^2}.
\end{aligned}
\end{equation}

\begin{remark} \label{rmrk:sym}
     Sometimes, the structure of $\check{\bm{\alpha}}$ in \eqref{eq:check_abc} is more complicated than if $\bm{\alpha}^{ij}$ is allowed to be asymmetric. For example, in the Wright-Fisher diffusion given in Section \ref{sec:WF}, we can choose an asymmetric $\bm{\alpha}^{ij}$ that corresponds to a diagonal $\check{\bm{\alpha}}$ (see more details in Supplementary Material \ref{sec:Implementation}).
\end{remark}

\begin{remark} \label{rmrk:DimensionalitySigma}
When symmetric, each matrix $\bm{\alpha}^{ij}$ has $d(d+1)/2$ coefficients. Therefore, $[\bm{\Sigma}\bm{\Sigma}^\top(\mathbf{x}; \bm{\theta}^{(2)})]_{ij}$ has $(d+1)(d+2)/2$ coefficients. Since $\bm{\Sigma}\bm{\Sigma}^\top(\mathbf{x}; \bm{\theta}^{(2)})$ is also symmetric, it has $d(d+1)^2(d+2)/4$ coefficients, which is of order $d^4$. However, many of the coefficients will be structural in dependence on the particular model and therefore known. For example, for Gaussian models, $\bm{\alpha}^{ij}$ and $\bm{\beta}^{ij}$ will be zero. Moreover, in most applications, the quadratic diffusion matrix will be diagonal or nearly diagonal, as will $\bm{\alpha}^{ij}$, where most quadratic interactions between the coefficients of the state $\mathbf{X}_t$ are zero. Consequently, $\bm{\alpha}^{ij}$, $\bm{\beta}^{ij}$, and $\gamma^{ij}$ will generally be sparse with few non-zero elements and with only a few unknown parameters to be estimated.
\end{remark}

\subsection{Moments of multivariate Pearson diffusions}

A multivariate Pearson diffusion is defined by \eqref{eq:SDEsplitted} with $\mathbf{N} = \mathbf{0}$, that is
\begin{equation}\label{eq:MultivariatePearson}
\begin{aligned}
    \mathbf{F}(\mathbf{x}) &= \mathbf{A}(\mathbf{x} - \mathbf{b}), \\
    [\bm{\Sigma}\bm{\Sigma}^\top(\mathbf{x})]_{ij} &= \mathbf{x}^\top \bm{\alpha}^{ij} \mathbf{x} + \mathbf{x}^\top \bm{\beta}^{ij} + \gamma^{ij}, \qquad i,j=1,2,...,d.
\end{aligned}
\end{equation}

The following theorem provides an analytic formula for the mean of a multivariate Pearson diffusion and a computational tool for explicitly deriving the covariance matrix. The proof can be found in Section \ref{sec:TechnicalDetails}.

\begin{theorem}[Mean and covariance matrix of multivariate Pearson diffusions] \label{thm:Covariance_matrix}
Let $\mathbf{X}_t$ be the strong solution of the SDE defined by the drift and diffusion coefficient in \eqref{eq:MultivariatePearson}. The mean vector $\mathbf{m}(t) = \mathbb{E}[\mathbf{X}_t]$ with initial conditions $\mathbf{m}(0)$ is given by
\begin{equation}
    \mathbf{m}(t) = \exp(\mathbf{A} t)(\mathbf{m}(0) - \mathbf{b}) + \mathbf{b}. \label{eq:MeanODESolution}
\end{equation}
The vectorized covariance matrix $\mathbf{C}(t) = \mathbb{E}[(\mathbf{X}_t - \mathbf{m}(t))(\mathbf{X}_t - \mathbf{m}(t))^\top]$, with initial conditions $\mathbf{m}(0), \mathbf{C}(0)$ is given by
\begin{equation}
\label{eq:vecCfinal}
\begin{aligned}
    \vect(\mathbf{C}(t)) &= \exp(( \mathbf{A} \oplus \mathbf{A} + \check{\bm{\alpha}})t) \vect(\mathbf{C}(0))\\
    &+ \mathbf{I}_1(t, \mathbf{A}, \check{\bm{\alpha}}) \vect((\mathbf{m}(0) - \mathbf{b})(\mathbf{m}(0) - \mathbf{b})^\top) +  
    \mathbf{I}_2(t, \mathbf{A}, \check{\bm{\alpha}}) \vect((\mathbf{m}(0) - \mathbf{b})\mathbf{b}^\top)\\
    &+  \mathbf{I}_3(t, \mathbf{A}, \check{\bm{\alpha}})  \vect(\mathbf{b}(\mathbf{m}(0) - \mathbf{b})^\top)
    +  \mathbf{I}_4(t, \mathbf{A}, \check{\bm{\alpha}}, \check{\bm{\beta}})  (\mathbf{m}(0) - \mathbf{b})
    +  \mathbf{I}_5(t, \mathbf{A}, \check{\bm{\alpha}}) \vect(\bm{\Sigma}\bm{\Sigma}^\top(\mathbf{b})),
\end{aligned}
\end{equation}
where
\begin{align}
    \mathbf{I}_1(t, \mathbf{A}, \check{\bm{\alpha}}) &\coloneqq \int_0^t \exp((\mathbf{A} \oplus \mathbf{A} + \check{\bm{\alpha}})(t-s)) \check{\bm{\alpha}} \exp((\mathbf{A} \oplus \mathbf{A}) s ) \dif s \in \mathbb{R}^{d^2 \times d^2},\label{eq:cov_int_1}\\
    \mathbf{I}_2(t, \mathbf{A}, \check{\bm{\alpha}}) &\coloneqq  \int_0^t \exp((\mathbf{A} \oplus \mathbf{A} + \check{\bm{\alpha}})(t-s)) \check{\bm{\alpha}} \exp((\mathbf{I} \oplus \mathbf{A}) s ) \dif s \in \mathbb{R}^{d^2 \times d^2},\label{eq:cov_int_2}\\
    \mathbf{I}_3(t, \mathbf{A}, \check{\bm{\alpha}})&\coloneqq \int_0^t \exp((\mathbf{A} \oplus \mathbf{A} + \check{\bm{\alpha}})(t-s)) \check{\bm{\alpha}} \exp((\mathbf{A} \oplus \mathbf{I}) s ) \dif s \in \mathbb{R}^{d^2 \times d^2},\label{eq:cov_int_3}\\
    \mathbf{I}_4(t, \mathbf{A}, \check{\bm{\alpha}}, \check{\bm{\beta}})&\coloneqq \int_0^t \exp((\mathbf{A} \oplus \mathbf{A} + \check{\bm{\alpha}})(t-s)) \check{\bm{\beta}} \exp(\mathbf{A} s ) \dif s \in \mathbb{R}^{d^2 \times d},\label{eq:cov_int_4}\\
    \mathbf{I}_5(t, \mathbf{A}, \check{\bm{\alpha}}) &\coloneqq \int_0^t \exp((\mathbf{A} \oplus \mathbf{A} + \check{\bm{\alpha}})(t-s)) \dif s \in \mathbb{R}^{d^2 \times d^2}.\label{eq:cov_int_5}
\end{align}
\end{theorem}

Each of the integrals \eqref{eq:cov_int_1}--\eqref{eq:cov_int_5} can be evaluated using Theorem 1 in \citep{VanLoan1978}. For instance, if we define the following block matrix
\begin{equation*}
    \mathbf{M}_1(\mathbf{A}, \check{\bm{\alpha}}) = \begin{bmatrix}
        \mathbf{A} \oplus \mathbf{A} + \check{\bm{\alpha}} & \check{\bm{\alpha}}\\
        \mathbf{0} &  \mathbf{A} \oplus \mathbf{A}
    \end{bmatrix},
\end{equation*}
then
\begin{equation*}
    \exp(\mathbf{M}_1(\mathbf{A}, \check{\bm{\alpha}}) t) = \begin{bmatrix}
        \star  & \mathbf{I}_1(t, \mathbf{A}, \check{\bm{\alpha}})\\
        \mathbf{0} &  \star
    \end{bmatrix},
\end{equation*}
where $\star$ denotes a matrix of no particular interest. Similarly, we can compute the other integrals $\mathbf{I}_2$-$\mathbf{I}_5$.

From $\vect(\mathbf{C}(t))$ we can easily derive $\mathbf{C}(t)$ by reshaping. 

\begin{remark}
    While this algorithm explicitly calculates the covariance matrix of a multivariate Pearson diffusion, it becomes computationally intensive for large dimension $d$ due to the need to compute matrices of size $d^2 \times d^2$. This can be improved to some extent by employing symmetric vectorization $\svec$ (cf. \citep{Klerk2002}) or half-vectorization $\vech$ (cf. \citep{Magnus2019}) instead of the $\vect$ operator.
\end{remark}

\subsection{Strang splitting scheme}

Consider the following splitting of SDE \eqref{eq:SDEsplitted}
\begin{align}
    \dif \mathbf{X}^{[1]}_t &=  \mathbf{A}(\mathbf{X}^{[1]}_t - \mathbf{b}) \dif t + \bm{\Sigma}(\mathbf{X}_t^{[1]}; \bm{\theta}^{(2)}) \dif \mathbf{W}_t, & \mathbf{X}^{[1]}_0 = \mathbf{x}_0, \label{eq:SplittingEq1}\\    
    \dif \mathbf{X}^{[2]}_t &= \mathbf{N}(\mathbf{X}^{[2]}_t; \bm{\theta}^{(1)}) \dif t, & \mathbf{X}^{[2]}_0 = \mathbf{x}_0. \label{eq:SplittingEq2}
\end{align}
The splitting is not unique, and the choice of $\mathbf{A}$ and $\mathbf{b}$ will influence the quality of the approximation. In \citep{pilipovic2024}, we proposed to choose $\mathbf{b}$ as a fixed point of the drift, and the linear part as the linearization around $\mathbf{b}$. Equation \eqref{eq:SplittingEq1} is a multivariate Pearson diffusion whose solution cannot be explicitly obtained in general. However, we can approximate the transition density by a  Gaussian, as suggested by \cite{Kessler1997}, who proposed an approximating Gaussian distribution, matching the mean and variance of the transition density. In general, moments cannot be obtained without knowing the transition density, but they can be approximated using the infinitesimal generator. However, for the multivariate Pearson diffusions, mean and covariance can be explicitly calculated by Theorem \ref{thm:Covariance_matrix}, enabling us to approximate the solution to \eqref{eq:SplittingEq1} as
\begin{equation}
     \mathbf{X}_{t_k}^{[1]} = \Psi_h^{[1]}(\mathbf{X}_{t_{k-1}}^{[1]}) = \bm{\mu}_h(\mathbf{X}_{t_{k-1}}^{[1]}; \bm{\theta}^{(1)}) + \bm{\xi}_{h}(\mathbf{X}_{t_{k-1}}^{[1]}; \bm{\theta}), \label{eq:OU}
\end{equation}
where $\bm{\xi}_{h}(\mathbf{X}_{t_{k-1}}; \bm{\theta}) \mid \mathbf{X}_{t_{k-1}} = \mathbf{x}_{k-1} \sim \mathcal{N}_d (\bm{0}, \bm{\Omega}_h(\mathbf{x}_{k-1}; \bm{\theta}))$ are independent for $k=1, \ldots , N$, i.e., the transition density is Gaussian, 
\begin{equation}
    p^{[1]}(\mathbf{x}_{t_k} \mid \mathbf{x}_{t_{k-1}}) = \mathcal{N}(\mathbf{x}_{t_k}; \bm{\mu}_h(\mathbf{x}_{t_{k-1}};\bm{\theta}^{(1)}), \bm{\Omega}_h(\mathbf{x}_{t_{k-1}}; \bm{\theta})). \label{eq:density_split1}
\end{equation}
Let $\mathbb{E}^{[1]}$ be the expectation with respect to the probability density \eqref{eq:density_split1}. In  Eq.~\eqref{eq:difExp} in Section \ref{sec:TechnicalDetails}, we derive  $\bm{\mu}_h(\mathbf{x}; \bm{\theta}^{(1)})$ and $\bm{\Omega}_h(\mathbf{x}; \bm{\theta})$ conditioning on $\mathbf{X}_{t_{k-1}}^{[1]} = \mathbf{x}$, obtaining
\begin{equation}
\bm{\mu}_h(\mathbf{x}; \bm{\theta}^{(1)}) = \mathbb{E}^{[1]}[\mathbf{X}_{t_k}^{[1]} \mid \mathbf{X}_{t_{k-1}}^{[1]} = \mathbf{x}] = \exp(\mathbf{A} h)(\mathbf{x} - \mathbf{b}) + \mathbf{b}. \label{eq:mu_h}
\end{equation}
Conditioning on $\mathbf{X}_{t_{k-1}}^{[1]} = \mathbf{x}$ in the covariance is equivalent to setting $\mathbf{C}(0) = \mathbf{0}$ in \eqref{eq:Sylvester_sol}. This yields 
\begin{equation} \label{eq:Omegah_int}
    \begin{aligned}
    \bm{\Omega}_h(\mathbf{x}; \bm{\theta}) &= \mathbb{E}^{[1]}[(\mathbf{X}_{t_k}^{[1]}  - \bm{\mu}_h(\mathbf{x}; \bm{\theta}^{(1)}))(\mathbf{X}_{t_k}^{[1]}  - \bm{\mu}_h(\mathbf{x}; \bm{\theta}^{(1)}))^\top \mid \mathbf{X}_{t_{k-1}}^{[1]} = \mathbf{x} ]  \\
    &= \int_{t_{k-1}}^{t_k} \exp(\mathbf{A}(t_k-s))\mathbb{E}^{[1]}[\bm{\Sigma}\bm{\Sigma}^\top(\mathbf{X}_s^{[1]}) \mid \mathbf{X}_{t_{k-1}}^{[1]} = \mathbf{x} ] \exp(\mathbf{A}^\top(t_k-s)) \dif s \\
    &=  \int_0^h \exp(\mathbf{A}(h-u))\mathbb{E}^{[1]}[\bm{\Sigma}\bm{\Sigma}^\top(\mathbf{X}_u^{[1]}) \mid \mathbf{X}_0^{[1]} = \mathbf{x} ] \exp(\mathbf{A}^\top(h-u)) \dif u, 
\end{aligned}
\end{equation}
\begin{remark}\label{remark:Omegah_no_dependence_t}
Due to the linear drift, the conditional law of $\mathbf{X}_s^{[1]} \mid \mathbf{X}_{t_{k-1}}^{[1]} = \mathbf{x}$ depends on $s - t_{k-1}$ and not on $s$ and $t$ individually. Thus, $\bm{\Omega}_h$ is time-homogeneous, as shown via change of variables in \eqref{eq:Omegah_int}.
\end{remark}

Since the integral in \eqref{eq:Omegah_int} cannot be computed in the general case, we again vectorize the covariance matrix and follow the same logic to get 
\begin{align}
    \vect(\bm{\Omega}_h(\mathbf{x}; \bm{\theta})) &= \mathbf{I}_1(h, \mathbf{A}, \check{\bm{\alpha}}) \vect((\mathbf{x} - \mathbf{b})(\mathbf{x} - \mathbf{b})^\top) +  
    \mathbf{I}_2(h, \mathbf{A}, \check{\bm{\alpha}}) \vect((\mathbf{x} - \mathbf{b})\mathbf{b}^\top)\notag\\
    &+  \mathbf{I}_3(h, \mathbf{A}, \check{\bm{\alpha}})  \vect(\mathbf{b}(\mathbf{x} - \mathbf{b})^\top)
    +  \mathbf{I}_4(h, \mathbf{A}, \check{\bm{\alpha}}, \check{\bm{\beta}})  (\mathbf{x} - \mathbf{b})
    +  \mathbf{I}_5(h, \mathbf{A}, \check{\bm{\alpha}}) \vect(\bm{\Sigma}\bm{\Sigma}^\top(\mathbf{b})). \label{eq:Omega_vec}
\end{align}

Assumptions \ref{as:monoton_condition} and \ref{as:F_polynomial_growth} ensure the existence and uniqueness of the solution of \eqref{eq:SplittingEq2} (Theorem 1.2.17 in \citet{Humphries2002}). Thus, there exists a unique function $\bm{f}_h : \mathbb{R}^d \times \overline{\Theta}_{\theta^{(1)}} \to \mathbb{R}^d$, for $h \geq 0$, such that
\begin{equation}
     \mathbf{X}_{t_k}^{[2]} = \Phi_h^{[2]}(\mathbf{X}_{t_{k-1}}^{[2]}) = \bm{f}_h(\mathbf{X}_{t_{k-1}}^{[2]}; \bm{\theta}^{(1)}). \label{eq:fhflow}
\end{equation}
Ideally, we would like the function $\bm{f}_h$ in \eqref{eq:fhflow} to be explicit. However, if this is not the case, we can approximate it using standard numerical tools, such as Runge-Kutta algorithms. Here, we assume $\bm{f}_h$ is readily available, but all the following results also hold if we approximate $\bm{f}_h$ up to order $h^2$ (see Proposition 2.2 in \cite{pilipovic2024}).

For all $\bm{\theta}^{(1)} \in \overline{\Theta}_{\theta^{(1)}}$, the time flow $\bm{f}_h$ fulfills the following semi-group properties
\begin{align}
    \bm{f}_0 (\mathbf{x}; \bm{\theta}^{(1)}) &= \mathbf{x}, \qquad \bm{f}_{t+s}(\mathbf{x}; \bm{\theta}^{(1)}) = \bm{f}_t(\bm{f}_s(\mathbf{x}; \bm{\theta}^{(1)}); \bm{\theta}^{(1)}), \ \ t, s \geq 0. \label{eq:fhassociativity}
\end{align}
Now, we are ready to define the Strang splitting approximation of \eqref{eq:SDEsplitted}.
\begin{definition} \label{def:splitting}
Let Assumptions \ref{as:Diffusion}-\ref{as:F_polynomial_growth} hold. The Strang splitting approximation of the solution of \eqref{eq:SDEsplitted} is given by
\begin{align}
    &{\mathbf{X}}^\mathrm{[SS]}_{t_k} \coloneqq \Phi_h^\mathrm{[SS]}({\mathbf{X}}^\mathrm{[SS]}_{t_{k-1}}) = (\Phi_{h/2}^{[2]} \circ \Psi_h^{[1]} \circ \Phi_{h/2}^{[2]} )({\mathbf{X}}^\mathrm{[SS]}_{t_{k-1}}) = \bm{f}_{h/2} (\bm{\mu}_h(\bm{f}_{h/2}({\mathbf{X}}^\mathrm{[SS]}_{t_{k-1}}) ) +  \bm{\xi}_{h,k}(\bm{f}_{h/2}({\mathbf{X}}^\mathrm{[SS]}_{t_{k-1}}))). \label{eq:StrangSplitting}
\end{align}
\end{definition}

\subsection{Strang splitting estimator}

To define the estimator, we also need the backward flow $\bm{f}_{-h}=\bm{f}_h^{-1}$, which might not be defined for all $h \geq 0, \mathbf{x} \in \mathcal{X}, \bm{\theta}^{(1)} \in \overline{\Theta}_{\theta^{(1)}}$. We, therefore, introduce the following and last assumption.
\begin{itemize}
    \myitem{(A6)} \label{as:fhInv} 
    There exists $h_0 > 0$ such that $\bm{f}_h^{-1}(\cdot;\bm{\theta}^{(1)})$ is well-defined on $\mathcal{X} \times \overline{\Theta}_{\theta^{(1)}}$ for all $h \in [0,h_0)$.
\end{itemize}

The SS splitting \eqref{eq:StrangSplitting} is a nonlinear transformation of the Gaussian random variable
\begin{equation*}
    \bm{\mu}_h(\bm{f}_{h/2}({\mathbf{X}}_{t_{k-1}}; \bm{\theta}^{(1)}); \bm{\theta}^{(1)}) + \bm{\xi}_{h,k}(\bm{f}_{h/2}({\mathbf{X}}_{t_{k-1}}; \bm{\theta}^{(1)}); \bm{\theta}).
\end{equation*}
We define
\begin{align}
    \mathbf{Z}_{t_k}(\bm{\theta}^{(1)}) &\coloneqq \bm{f}_{h/2}^{-1}(\mathbf{X}_{t_k}; \bm{\theta}^{(1)}) - \bm{\mu}_h(\bm{f}_{h/2}(\mathbf{X}_{t_{k-1}}; \bm{\theta}^{(1)}); \bm{\theta}^{(1)}),  \label{eq:Ztk}\\
    \bm{\Omega}_h^\mathrm{[SS]}(\mathbf{X}_{t_{k-1}}; \bm{\theta}) &\coloneqq \bm{\Omega}_h(\bm{f}_{h/2}({\mathbf{X}}_{t_{k-1}}; \bm{\theta}^{(1)});  \bm{\theta}) \label{eq:OmegahSS}
\end{align}
and apply a change of variables to derive the following objective function
\begin{equation} \label{eq:S_obj} 
\begin{aligned}
     \mathcal{L}^\mathrm{[SS]}(\mathbf{X}_{0:t_N}; \bm{\theta}) &= \sum_{k=1}^N \left( \log\det  \bm{\Omega}_h^\mathrm{[SS]}(\mathbf{X}_{t_{k-1}}; \bm{\theta}) + \mathbf{Z}_{t_k}(\bm{\theta}^{(1)})^\top  \bm{\Omega}_h^\mathrm{[SS]}(\mathbf{X}_{t_{k-1}}; \bm{\theta})^{-1} \mathbf{Z}_{t_k}(\bm{\theta}^{(1)}) \right) \\
     &- 2\sum_{k=1}^N \log |\det D_\mathbf{x} \bm{f}_{-h/2}(\mathbf{X}_{t_k}; \bm{\theta}^{(1)})|. 
\end{aligned}
\end{equation}
The SS estimator is then defined as
\begin{equation}
    \widehat{\bm{\theta}}_N^\mathrm{[SS]} \coloneqq \argmin_{\bm{\theta}} \mathcal{L}^\mathrm{[SS]}\left(\mathbf{X}_{0:t_N}; \bm{\theta}\right). \label{eq:S_est}
\end{equation}

\subsection{Asymptotic results}

Here, we state the consistency and asymptotic normality of the Strang splitting estimator. Proofs can be found in Supplementary Material \ref{sec:Proofs}. Let $\mathcal{L}^\mathrm{[SS]}$ be the Strang objective function \eqref{eq:S_obj} and $\widehat{\bm{\theta}}_N^\mathrm{[SS]}$ be the corresponding estimator \eqref{eq:S_est}. 

\begin{theorem}[Consistency of the estimator] \label{thm:Consistency}
Assume \ref{as:Diffusion}-\ref{as:fhInv}, $h \to 0$, and $Nh \to \infty$. Then 
\begin{align*}
    \widehat{\bm{\theta}}_N^\mathrm{[SS]} \xrightarrow[]{\mathbb{P}_{\bm{\theta}_0}} \bm{\theta}_0.
\end{align*}
\end{theorem}

To state the asymptotic normality of the estimator, we introduce some preliminaries. Let $\rho >0$ and define a closed ball around $\bm{\theta}_0$ by $\mathcal{B}_\rho\left(\bm{\theta}_0\right) = \{\bm{\theta} \in \Theta \mid \|\bm{\theta}-\bm{\theta}_0\| \leq \rho\}$. Since $\bm{\theta}_0 \in \Theta$, for a sufficiently small $\rho >0$, we have $\mathcal{B}_\rho(\bm{\theta}_0) \subset \Theta$. For $\hat{\bm{\theta}}_N^\mathrm{[SS]} \in \mathcal{B}_\rho\left(\bm{\theta}_0\right)$, the mean value theorem yields
\begin{equation}
    \left(\int_0^1 \mathbb{H}_{\mathcal{L}^\mathrm{[SS]}}(\bm{\theta}_0 + t (\hat{\bm{\theta}}_N^\mathrm{[SS]} - \bm{\theta}_0))\dif t\right) (\hat{\bm{\theta}}_N^\mathrm{[SS]} - \bm{\theta}_0) = - \nabla_{\bm{\theta}} \mathcal{L}^\mathrm{[SS]}\left(\bm{\theta}_0\right). \label{eq:AssymptoticNormalityDecomp}
\end{equation}
Define the block-scaled Hessian matrix
\begin{align}
    \mathbf{C}_N^\mathrm{[SS]}(\bm{\theta}) \coloneqq 
    \begin{bmatrix} \vspace{1ex}
    \left[ \frac{1}{Nh}\partial^2_{\theta^{(1)}_{i_1} \theta^{(1)}_{i_2}} \mathcal{L}^\mathrm{[SS]}(\bm{\theta})\right]_{i_1,i_2 =1}^r & \left[\frac{1}{N\sqrt{h}}\partial^2_{\theta^{(1)}_i \theta^{(2)}_j} \mathcal{L}^\mathrm{[SS]}(\bm{\theta})\right]_{i=1,j=1}^{r,s}\\ 
    \left[\frac{1}{N\sqrt{h}}\partial^2_{\theta^{(2)}_j \theta^{(1)}_i} \mathcal{L}^\mathrm{[SS]}(\bm{\theta})\right]_{i=1,j=1}^{r,s} & \left[\frac{1}{N} \partial^2_{\theta^{(2)}_{j_1} \theta^{(2)}_{j_2}} \mathcal{L}^\mathrm{[SS]} (\bm{\theta})\right]_{j_1, j_2 = 1}^s
    \end{bmatrix}, \label{eq:C_mat}
\end{align}
along with the scaled error and score vectors
\begin{align}
    \mathbf{s}_N^\mathrm{[SS]} \coloneqq \begin{bmatrix}
    \sqrt{Nh} (\hat{\bm{\theta}}_N^\mathrm{(1)[SS]} - \bm{\theta}_0^{(1)}) \vspace{1ex}\\
    \sqrt{N} (\hat{\bm{\theta}}_N^\mathrm{(2)[SS]} - \bm{\theta}_0^{(2)})
    \end{bmatrix},  \qquad
    \bm{\lambda}_N^\mathrm{[SS]} \coloneqq \begin{bmatrix}
    -\dfrac{1}{\sqrt{Nh}} \nabla_{\bm{\theta}^{(1)}} \mathcal{L}^\mathrm{[SS]}(\bm{\theta}_0)\vspace{1ex}\\
    -\dfrac{1}{\sqrt{N}}  \nabla_{\bm{\theta}^{(2)}} \mathcal{L}^\mathrm{[SS]}(\bm{\theta}_0)
    \end{bmatrix}, \label{eq:lambda}
\end{align}
and $\mathbf{D}_N^\mathrm{[SS]} \coloneqq \int_0^1 \mathbf{C}_N^\mathrm{[SS]}(\bm{\theta}_0 + t (\hat{\bm{\theta}}_N^\mathrm{[SS]} - \bm{\theta}_0)) \dif t$. Consequently, \eqref{eq:AssymptoticNormalityDecomp} can be rewritten equivalently as $\mathbf{D}_N^\mathrm{[SS]} \mathbf{s}_N^\mathrm{[SS]} = \bm{\lambda}_N^\mathrm{[SS]}$. 

Let
\begin{align*}
    &[\mathbf{C}_{\bm{\theta}^{(1)}}(\bm{\theta}_0)]_{i_1,i_2} \coloneqq \int (\partial_{\theta^{(1)}_{i_1}} \mathbf{F}_0(\mathbf{x}))^\top \bm{\Sigma}\bm{\Sigma}_0^\top(\mathbf{x})^{-1}(\partial_{\theta^{(1)}_{i_2}} \mathbf{F}_0(\mathbf{x}))\dif \nu_0(\mathbf{x}), \quad 1\leq i_1,i_2 \leq r, \\
    &[\mathbf{C}_{\bm{\theta}^{(2)}}(\bm{\theta}_0)]_{j_1,j_2} \coloneqq \frac{1}{2}\int \tr\Big((\partial_{\theta^{(2)}_{j_1}} \bm{\Sigma}\bm{\Sigma}_0^\top(\mathbf{x}))\bm{\Sigma}\bm{\Sigma}_0^\top(\mathbf{x})^{-1}(\partial_{\theta^{(2)}_{j_2}} \bm{\Sigma}\bm{\Sigma}_0^\top(\mathbf{x}))\bm{\Sigma}\bm{\Sigma}_0^\top(\mathbf{x})^{-1}\Big)\dif \nu_0(\mathbf{x}), \quad 1\leq j_1,j_2 \leq s, 
\end{align*}
and define the block-diagonal asymptotic information matrix as
\begin{align*}
        &\mathbf{C}^\mathrm{[SS]}(\bm{\theta}_0) \coloneqq \begin{bmatrix}
        \mathbf{C}_{\bm{\theta}^{(1)}}(\bm{\theta}_0) & \bm{0}_{r\times s}\\
        \bm{0}_{s\times r} & \mathbf{C}_{\bm{\theta}^{(2)}}(\bm{\theta}_0)
        \end{bmatrix}.
\end{align*}

\begin{theorem}[Asymptotic normality of the estimator]\label{thm:AsymtoticNormality}
    Assume \ref{as:Diffusion}-\ref{as:fhInv}. If $h \to 0$, $Nh \to \infty$, and $Nh^2 \to 0$, then 
    \begin{align*}
    \begin{bmatrix}
        \sqrt{Nh} (\hat{\bm{\theta}}_N^\mathrm{(1)[SS]} - \bm{\theta}_0^{(1)})\vspace{1ex} \\
        \sqrt{N} (\hat{\bm{\theta}}_N^\mathrm{(2)[SS]} - \bm{\theta}_0^{(2)})
    \end{bmatrix} & \xrightarrow[\substack{h \to 0 \\ Nh \to \infty\\ Nh^2 \to 0}]{d} \mathcal{N}(\bm{0}, \mathbf{C}^\mathrm{[SS]}(\bm{\theta}_0)^{-1}), 
\end{align*}
under $\mathbb{P}_{\bm{\theta}_0}$.
\end{theorem}

The proof of Theorem \ref{thm:AsymtoticNormality} follows the same structure as Theorem 1 in \citep{Kessler1997}, and it suffices to prove Lemmas \ref{lemma:AsymptoticNormality1} and \ref{lemma:LnConvergence} below. 

\begin{lemma} \label{lemma:AsymptoticNormality1}
Let $\mathbf{C}_N^\mathrm{[SS]}(\bm{\theta}_0)$ be defined as in \eqref{eq:C_mat}. As $h \to 0$ and $Nh \to \infty$, it holds that
\begin{align*}
        &\mathbf{C}_N^\mathrm{[SS]}(\bm{\theta}_0) \xrightarrow[\substack{h \to 0 \\ Nh \to \infty}]{\mathbb{P}_{\bm{\theta}_0}} 2\mathbf{C}^\mathrm{[SS]}(\bm{\theta}_0).
\end{align*}
Moreover, let $\rho_N$ be a sequence such that $\rho_N \to 0$. Then,
\begin{align*}
    \sup_{\|\bm{\theta}\| \leq \rho_N} \|\mathbf{C}_N^\mathrm{[SS]}(\bm{\theta}_0 + \bm{\theta}) -\mathbf{C}_N^\mathrm{[SS]}(\bm{\theta}_0) &\|\xrightarrow[\substack{h \to 0 \\ Nh \to \infty}]{\mathbb{P}_{\bm{\theta}_0}} 0.
\end{align*}
\end{lemma}

\begin{lemma} \label{lemma:LnConvergence}
Let $\bm{\lambda}_N^\mathrm{[SS]}$ be defined as in \eqref{eq:lambda}. As $h \to 0$, $Nh \to \infty$ and $Nh^2 \to 0$, it holds that
\begin{align*}
        &\bm{\lambda}_N^\mathrm{[SS]}  \xrightarrow[\substack{h \to 0 \\ Nh \to \infty\\ Nh^2 \to 0}]{d} \mathcal{N}(\bm{0}, 4\mathbf{C}^\mathrm{[SS]}(\bm{\theta}_0))
\end{align*}
under $\mathbb{P}_{\bm{\theta}_0}$.
\end{lemma}

\section{Simulation study} \label{sec:Simulations}

In this section, we conduct a simulation study of the multivariate Wright-Fisher diffusion \eqref{eq:SDEWF} and the Student Kramers oscillator \eqref{eq:StudentKramersSDE}, demonstrating the theoretical and computational performance of our proposed estimator and comparing it to the EM, GA, and Ozaki's LL estimators. The comparison with these three methods is motivated by their widespread use and performance: the EM estimator is commonly applied in practice, our proposed method generalizes the GA estimator, and the LL estimator is recognized as one of the state-of-the-art frequentist based methods, as discussed in \cite{pilipovic2024, pilipovic2024SecondOrder}. In Section \ref{sec:TechnicalDetails}, we briefly explain the EM, GA, and LL estimators and provide references. 

The LL estimator is only defined for SDEs with constant diffusion coefficients. For reducible diffusions \citep{AitSahalia2008}, the Lamperti transform is used to obtain an SDE with constant diffusion coefficients. The Student Kramers oscillator is reducible, and the LL estimator can thus also be applied. For non-reducible SDEs such as the Wright-Fisher diffusion, we apply LL by assuming that the diffusion matrix $\bm{\Sigma}(\mathbf{X}_t)$ is constant locally on the intervals $[t_k, t_{k+1})$, i.e., $\bm{\Sigma}(\mathbf{X}_t) = \bm{\Sigma}(\mathbf{X}_{t_k})$, for $t\in [t_k, t_{k+1})$. This idea is also used by \cite{melnykova2020parametric} for hypoelliptic SDEs.

For evaluating and optimizing the objective functions across all simulation studies, we implement the estimators in \texttt{Python} using \texttt{JAX} \citep{jax2018github}  with 64-bit precision. This framework provides exact analytical gradients via automatic differentiation (\texttt{jax.value\_and\_grad}) and highly efficient execution through just-in-time (JIT) compilation. Source code is available at \citep{pilipovic2026code}.

\subsection{Multivariate Wright-Fisher diffusion} \label{sec:WFsimulation}

We use a one-locus model with four alleles, $L = 1$ and $M = 4$, so $\mathbf{S} = \mathbf{0}$ and the model reduces to \eqref{eq:WF_final_vec}, with $\mathbf{q} = (q_1, q_2, q_3, q_4)^\top$, $\bm{P} \in \mathbb{R}^{4 \times 4}_+$ as in \eqref{eq:P}, diffusion matrix $\bm{\Sigma}\bm{\Sigma}^\top(\mathbf{X}_t) = \diag(\mathbf{X}_t) - \mathbf{X}_t\mathbf{X}_t^\top$, and simplex constraint $\sum_{i=1}^4 X_t^{(i)} = 1$. Since one component is determined by the others, we reduce \eqref{eq:WF_final_vec} to a three-dimensional system
\begin{equation} \label{eq:WF_final_reduced_vec}
    \dif \mathbf{X}_t = (\bm{\kappa} + \bm{\kappa}  \mathbf{X}_t - \mathbf{X}_t\mathbf{X}_t^\top \bm{\lambda})\dif t + \bm{\Sigma}(\mathbf{X}_t) \dif \mathbf{W}_t,
\end{equation}
where $\mathbf{X}_t = (X_t^{(1)}, X_t^{(2)}, X_t^{(3)})$ and
\begin{align*}
    \bm{\kappa} \coloneqq \frac{\tau}{2} 
    \begin{bmatrix} 
    p_{41} \\ 
    p_{42} \\ 
    p_{43} 
    \end{bmatrix}, \ \bm{\lambda} \coloneqq 
    \begin{bmatrix} 
    q_1 - q_4 \\ 
    q_2 - q_4 \\ 
    q_3 - q_4 
    \end{bmatrix}, \ \mathbf{K} \coloneqq \frac{\tau}{2} 
    \begin{bmatrix} 
    p_{11} - p_{41} + \frac{2}{\tau} \lambda_1 -1 & p_{21} - p_{41} & p_{31} - p_{41} \\ 
    p_{12} - p_{42} & p_{22} - p_{42} + \frac{2}{\tau} \lambda_2 -1 & p_{23} - p_{42} \\ 
    p_{13} - p_{43} & p_{32} - p_{43} & p_{33} - p_{43} + \frac{2}{\tau} \lambda_3 - 1 \end{bmatrix}.
\end{align*}
Note that only $\lambda_i = q_i - q_4$ and $\tau p_{ij}$ are identifiable, not $q_i$, $q_4$, $\tau$ and $p_{ij}$ separately. We estimate the 15 parameters 
\begin{equation}
    \bm{\theta} = (\kappa_1, \kappa_2, \kappa_3, K_{11}, K_{12}, K_{13}, K_{21}, K_{22}, K_{23}, K_{31}, K_{32}, K_{33}, \lambda_1, \lambda_2, \lambda_3), \label{eq:thetaWF}
\end{equation}
and transform them back to $\tilde{\bm{\theta}} = (q_1, q_2, q_3, p_{11}, p_{12}, p_{21}, \ldots, p_{43})$ using the true parameters $\tau_0$ and $q_4^0$.

We set the true parameters to 
\begin{equation}
\mathbf{q}_0 = 
\begin{bmatrix}
25 \\
40 \\
30 \\
10
\end{bmatrix}, \qquad
\mathbf{P}_0 =
\begin{bmatrix}
0.2 & 0.3 & 0.15 & 0.35 \\
0.2 & 0.05 & 0.35 & 0.40 \\
0.25 & 0.6 & 0.1 & 0.05 \\
0.15 & 0.1 & 0.1 & 0.65
\end{bmatrix}
\end{equation}
Sample paths of length $T = 20$ were simulated using the EM discretization with step size $h^{\mathrm{sim}} = 0.0002$ to ensure high accuracy. We sub-sampled to reduce discretization bias to time steps $h = 0.2$ and sample size $N = 100$. We repeated the simulations to obtain 1000 data sets.

The drift of \eqref{eq:WF_final_reduced_vec} is
\begin{align}
    \mathbf{F}(\mathbf{x}) = \bm{\kappa} + \bm{K}\mathbf{x} - \mathbf{x}\mathbf{x}^\top \bm{\lambda} \label{eq:F_WF}
\end{align}
with derivative
\begin{align}
    &D \mathbf{F}(\mathbf{x}) = \bm{K} - \mathbf{x}^\top \bm{\lambda}\mathbf{I} - \mathbf{x} \bm{\lambda}^\top.
    \label{eq:DF_WF}
\end{align}
We use $\mathbf{F}(\mathbf{x})$ and $D\mathbf{F}(\mathbf{x})$ in \eqref{eq:F_WF} -- \eqref{eq:DF_WF}  and $\bm{\Sigma} \bm{\Sigma}^\top(\mathbf{x})$ in \eqref{eq:SigmaSigmaT_final_WF} to implement the EM and LL estimators. For the GA estimator, we used \eqref{eq:GA} and \texttt{Wolfram Mathematica} to calculate $\bm{\Omega}_h^\mathrm{[GA]}$ up to order $h^3$ (computations are not provided).

The inverse of $\bm{\Sigma}\bm{\Sigma}^\top(\mathbf{x}) = \diag(\mathbf{x}) - \mathbf{x}\mathbf{x}^\top$ is $\bm{\Sigma}\bm{\Sigma}^\top(\mathbf{x})^{-1} = \diag(\mathbf{x})^{-1} + \frac{\mathbf{1}\mathbf{1}^\top}{1 - \mathbf{1}^\top \mathbf{x}}$. Then, we can calculate the information matrix $\mathbf{C}({\bm{\theta}_0})$ from Theorem \ref{thm:AsymtoticNormality} for the parameter vector $\bm{\theta}$ in \eqref{eq:thetaWF} as
\begin{align*}
    \mathbf{C}(\bm{\theta}_0) = \mathbf{C}=
    \mathbb{E}_{\nu_0}
    \begin{bmatrix}
        \diag(\mathbf{X})^{-1} + \dfrac{\mathbf{1}\mathbf{1}^\top}{X^{(4)}}
        & \left(\diag(\mathbf{X})^{-1} + \dfrac{\mathbf{1}\mathbf{1}^\top}{X^{(4)}}\right)\otimes \mathbf{X}^\top
        & -\dfrac{\mathbf{1}\,\mathbf{X}^\top}{(X^{(4)})^2}
        \\[14pt]
        \star
        & \left(\diag(\mathbf{X})^{-1} + \dfrac{\mathbf{1}\mathbf{1}^\top}{X^{(4)}}\right)\otimes \mathbf{X}\mathbf{X}^\top
        & -\dfrac{\mathbf{1} \otimes \mathbf{X}\mathbf{X}^\top}{(X^{(4)})^2}
        \\[14pt]
        \star & \star
        & \dfrac{\mathbf{1}^\top\mathbf{X}}{X^{(4)}}\,\mathbf{X}\mathbf{X}^\top
    \end{bmatrix},    
\end{align*}
where $\mathbb{E}_{\nu_0}$ denotes expectation under the stationary distribution $\nu_0$ of \eqref{eq:WF_final_reduced_vec} and $\star$ denotes the symmetric counterpart. All expectations are approximated by ergodic averages over the simulated trajectory.

Then, Theorem~\ref{thm:AsymtoticNormality} yields
\begin{align*}
    \sqrt{Nh}\,\bigl(\hat{\bm{\theta}}_N - \bm{\theta}_0\bigr) \xrightarrow{d} \mathcal{N}\!\left(\mathbf{0},\, \mathbf{C}^{-1}\right) \qquad \text{as } N \to \infty,\ h \to 0,\ Nh \to \infty,\ Nh^2 \to 0.
\end{align*}
Since $\tilde{\bm{\theta}} =\bm{\phi}(\bm{\theta}) = \mathbf{J}\bm{\theta} + \mathbf{c}$ is affine in $\bm{\theta}$ (given the known constants $\tau_0$ and $q_4^0$), with $\mathbf{J} \in \mathbb{R}^{15 \times 15}$ straightforwardly computed from the back-transformation, the asymptotic distribution of $\tilde{\bm{\theta}}$ follows exactly
\begin{align*}
    \sqrt{Nh}\,\bigl(\tilde{\bm{\theta}}_N - \tilde{\bm{\theta}}_0\bigr)
    = \mathbf{J}\,\sqrt{Nh}\,\bigl(\hat{\bm{\theta}}_N - \bm{\theta}_0\bigr)
    \xrightarrow{d} \mathcal{N}\!\left(\mathbf{0},\,
    \mathbf{J}\,\mathbf{C}^{-1}\mathbf{J}^\top\right).
\end{align*}
The asymptotic standard deviations displayed in Figure~\ref{fig:WF_final_plot} are the square roots of the diagonal entries of $\mathbf{J}\,\widehat{\mathbf{C}}^{-1}\mathbf{J}^\top / (Nh)$, where $\widehat{\mathbf{C}}$ is estimated by ergodic averages over the simulated trajectory.

\subsubsection{Strang splitting estimator} \label{sec:WF_strang}

Usually, the drift is split such that the linear part is the linearization around a fixed point. However, analytic roots of $\mathbf{F}$ in \eqref{eq:F_WF} are not available, so we propose a splitting that solves an equation involving the Jacobian of $\mathbf{F}$. For polynomial drifts, this reduces the order of the equations by one, implying here that we solve a linear system since $\mathbf{F}$ is quadratic.

In the first step, we choose $\mathbf{A}$ as $\mathbb{E}[D \mathbf{F}(\mathbf{X_t})]$, such that $\mathbb{E}[D \mathbf{N}(\mathbf{X_t})] = \mathbb{E}[D \mathbf{F}(\mathbf{X_t})] - \mathbf{A} = \mathbf{0}$. Then, we find $\mathbf{b}$ by
\begin{equation}
    \mathbf{A} = \mathbb{E}[D \mathbf{F}(\mathbf{X}_t)] = D \mathbf{F}(\mathbf{b}). \label{eq:findingb}
\end{equation}
In practice, we use empirical means instead of expectations. 

Since $D \mathbf{F}(\mathbf{x})$ is linear, \eqref{eq:findingb} leads to $\mathbf{b} = \mathbb{E} [\mathbf{X}_t]$, and $\mathbf{A} = D \mathbf{F}(\mathbb{E} [\mathbf{X_t}]) =D \mathbf{F}(\mathbf{b})$, that is,
\begin{equation*}
    \mathbf{A} = D \mathbf{F}(\mathbf{b}) = \bm{K} - \mathbf{b}^{\top}\bm{\lambda}\mathbf{I} - \mathbf{b} \bm{\lambda}^\top.
\end{equation*}
Then, $\mathbf{N}(\mathbf{x})$ is 
\begin{align*}
    \mathbf{N}(\mathbf{x}) &= \mathbf{F}(\mathbf{x}) - \mathbf{A}(\mathbf{x} - \mathbf{b})\\
    &=\bm{\kappa} + \bm{K}\mathbf{x} - \mathbf{x}\mathbf{x}^\top \bm{\lambda}- (\bm{K} - \mathbf{b}^\top \bm{\lambda}\mathbf{I} - \mathbf{b} \bm{\lambda}^\top)(\mathbf{x} - \mathbf{b})\\
    &= \bm{\kappa} + \bm{K}\mathbf{b} - \mathbf{b}^\top \bm{\lambda} \mathbf{b} - \mathbf{b} \bm{\lambda}^\top \mathbf{b} + (\mathbf{b}^\top \bm{\lambda} \mathbf{I} + \mathbf{b}\bm{\lambda}^\top) \mathbf{x} - \mathbf{x}\mathbf{x}^\top \bm{\lambda}\\
    &=\mathbf{F}(\mathbf{b}) - (\mathbf{x} - \mathbf{b})(\mathbf{x} - \mathbf{b})^\top \bm{\lambda}. 
\end{align*}
Thus,
\begin{align*}
    D\mathbf{N}(\mathbf{x}) = -(\mathbf{x} - \mathbf{b})^\top \bm{\lambda} \mathbf{I}  -  (\mathbf{x}-\mathbf{b})\bm{\lambda}^\top,
\end{align*}
which indeed has zero expectation.
To solve the ODE in \eqref{eq:SplittingEq2}, we use a Runge-Kutta fourth order approximation
\begin{align*}
    &\mathbf{k}_1(\mathbf{x}) = \mathbf{N}(\mathbf{x}), &&D\mathbf{k}_1(\mathbf{x}) = D\mathbf{N}(\mathbf{x})\\
    &\mathbf{k}_2(\mathbf{x}, h) = \mathbf{N}(\mathbf{x} + \frac{h}{2} \mathbf{k}_1(\mathbf{x})), && D\mathbf{k}_2(\mathbf{x}, h) = D\mathbf{N}(\mathbf{x} + \frac{h}{2} \mathbf{k}_1(\mathbf{x}))(\mathbf{I} + \frac{h}{2}D\mathbf{k}_1(\mathbf{x})), \\  
    &\mathbf{k}_3(\mathbf{x},h) = \mathbf{N}(\mathbf{x} + \frac{h}{2} \mathbf{k}_2(\mathbf{x}, h)), &&D\mathbf{k}_3(\mathbf{x}, h) = D\mathbf{N}(\mathbf{x} + \frac{h}{2} \mathbf{k}_2(\mathbf{x}, h))(\mathbf{I} + \frac{h}{2}D\mathbf{k}_2(\mathbf{x},h)),\\  
    &\mathbf{k}_4(\mathbf{x}, h) = \mathbf{N}(\mathbf{x} + h \mathbf{k}_3(\mathbf{x}, h)), &&D\mathbf{k}_4(\mathbf{x}, h) = D\mathbf{N}(\mathbf{x} + h\mathbf{k}_3(\mathbf{x}, h))(\mathbf{I} + h D\mathbf{k}_3(\mathbf{x}, h)).
\end{align*}
Then, we approximate $\bm{f}_h(\mathbf{x})$ as
\begin{equation*}
    \bm{f}_h(\mathbf{x}) \approx \bm{f}_h^{[\mathrm{RK4}]}(\mathbf{x}) = \mathbf{x} + \frac{h}{6}(\mathbf{k}_1(\mathbf{x}) + 2 \mathbf{k}_2(\mathbf{x}, h) + 2\mathbf{k}_3(\mathbf{x}, h)+ \mathbf{k}_4(\mathbf{x}, h)).
\end{equation*}
To implement the Strang splitting estimator, we still need to compute $\log|\det (D\bm {f}_{-h/2}(\mathbf{x}))|$. We know that 
\begin{equation}
    D\bm{f}_{h}(\mathbf{x}) \approx D\bm{f}_{h}^{[\mathrm{RK4}]}(\mathbf{x}) = \mathbf{I} + \frac{h}{6}(D\mathbf{k}_1(\mathbf{x}) + 2 D\mathbf{k}_2(\mathbf{x},h) + 2D\mathbf{k}_3(\mathbf{x}, h)+ D\mathbf{k}_4(\mathbf{x}, h)), \label{eq:Dfh}
\end{equation}
then, we find $D\bm {f}_{-h/2}(\mathbf{x}))$ by plugging $-h/2$ in \eqref{eq:Dfh}.
\begin{remark}
    While we use \eqref{eq:Dfh} in the implementation when \eqref{eq:SplittingEq2} does not have a closed-form solution, we also know that $-2\log|\det (D\bm {f}_{-h/2}(\mathbf{x}))| = h\tr (D\mathbf{N}(\mathbf{x})) + \mathbf{R}(h^2, \mathbf{x})$. Since $\mathbb{E}[D\mathbf{N}(\mathbf{x})] = 0$ in our splitting, we expect that the term $-2\log|\det (D\bm {f}_{-h/2}(\mathbf{x}))| $ does not contribute much to the overall value of the objective function \eqref{eq:S_obj}.
\end{remark}

Now, we have all components needed to calculate $\mathbf{Z}_{t_k}(\bm{\theta^{(1)}})$ from \eqref{eq:Ztk}. For details on calculating $\bm{\Omega}_h(\mathbf{x})$ in \eqref{eq:Omegah_int}, see Supplementary Material~\ref{sec:Implementation}.

\begin{figure}
    \centering
    \includegraphics[width=\textwidth]{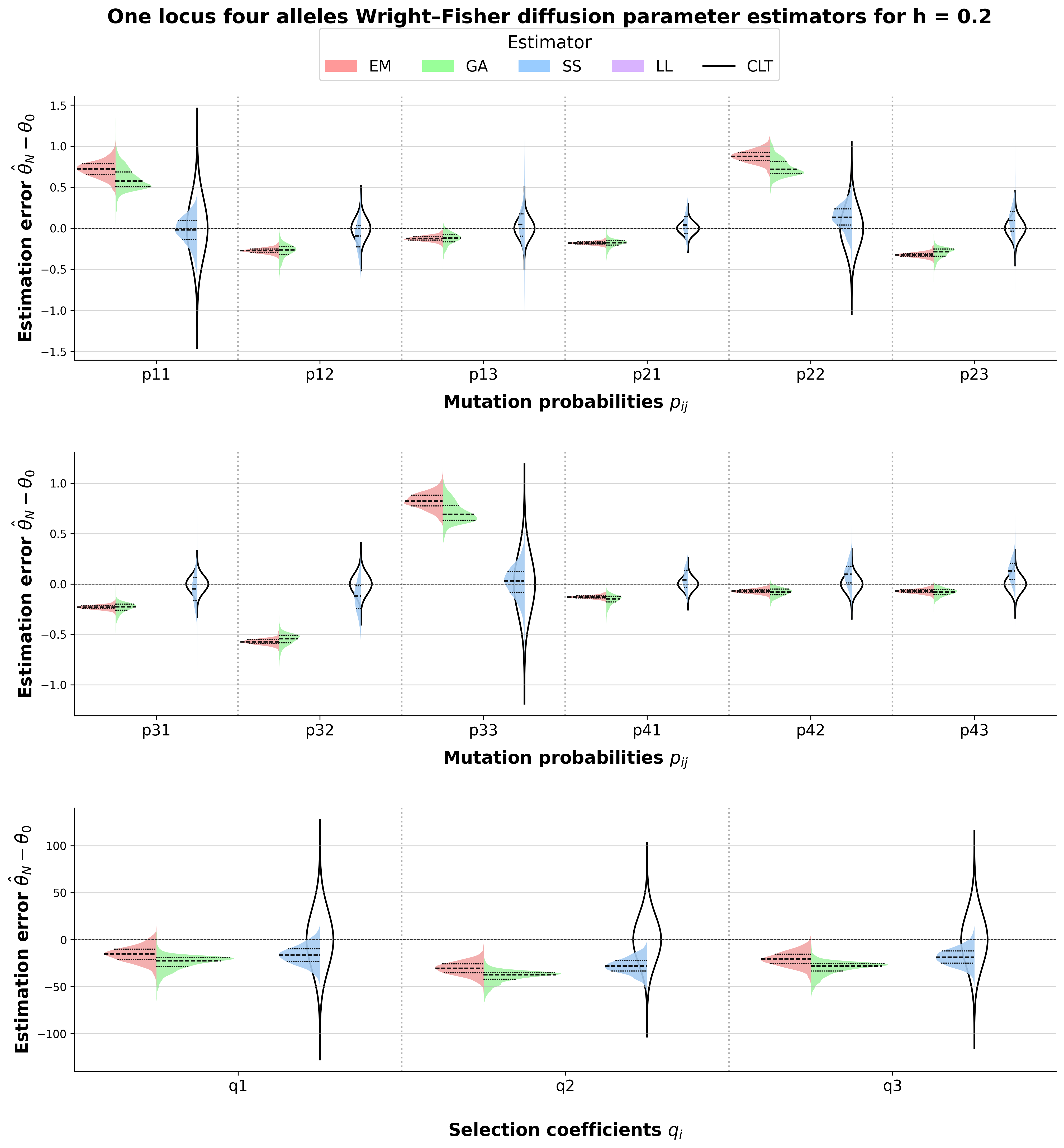}
    \caption{Wright-Fisher diffusion. Violin plots of parameter estimation errors $\hat{\bm{\theta}}_N - \bm{\theta}_0$ based on 1000 simulated datasets ($T = 20$, $N = 100$, $h = 0.2$). Colors indicate estimator. Inside each violin plot, black dashed lines show medians and the 25th and 75th percentiles. Black violin plots represent the asymptotic distributions obtained with the central limit theorem (CLT). The LL estimator returned NaN for each dataset. For the SS and GA estimator, we removed approximatly 5\% outliers.}
    \label{fig:WF_final_plot}
\end{figure}

\subsubsection{Optimization}

The optimization proceeds in two phases. In the first phase, the Adam optimizer \citep{kingma2015adam}, implemented through the \texttt{optax} library \citep{deepmind2020jax}, is run for up to \num{1000} iterations with a learning rate of \texttt{lr = 0.01}. Initial parameters are set to
\begin{equation*}
\mathbf{q}_\text{init} = 
\begin{bmatrix}
11 \\
11 \\
11 \\
10
\end{bmatrix}, \qquad
\mathbf{P}_\text{init} =
\begin{bmatrix}
0.53 & 0.05 & 0.05 & 0.37 \\
0.05 & 0.53 & 0.05 & 0.37 \\
0.05 & 0.05 & 0.53 & 0.37 \\
0.05 & 0.05 & 0.05 & 0.85
\end{bmatrix}.
\end{equation*}
Early stopping is applied with a tolerance of $10^{-6}$ when no sufficient decrease in the objective is observed. In the second phase, the Adam solution is refined using the BFGS algorithm \citep{nocedal2006numerical}, available as \texttt{jax.scipy.optimize.minimize} with \texttt{method = "BFGS"}.

\subsubsection{Results}

Figure~\ref{fig:WF_final_plot} presents the simulation results for the Wright--Fisher model, showing the distributions of the parameter estimation errors $\hat{\bm{\theta}}_N - \bm{\theta}_0$ for the three estimators, EM, GA and SS, based on 1000 simulated datasets with $T = 20$, $N = 100$ observations, and step size $h = 0.2$. The LL estimator did not converge for this step size. In Supplementary Material \ref{sec:ResultsWF}, we show results for $h = 0.02$, where SS and LL perform comparable, while still outperforming EM and GA. Moreover, the asymptotic regime is obtained. 

The distribution of the estimates of the mutation probabilities are shown in the first two rows of Figure~\ref{fig:WF_final_plot}. Each panel additionally displays the outline of the theoretical asymptotic normal distribution of the SS estimator (shown in black), derived via the delta method from the Fisher information matrix. For large $h = 0.2$ and small $N = 100$, none of the estimators have yet reached this asymptotic regime, as evidenced by the discrepancy between the empirical distributions and the theoretical envelope. Estimator SS has smaller bias than both EM and GA across all parameters except $p_{42}$, $p_{43}$, and remains closest to the asymptotic distribution. EM and GA strongly overestimate the probabilities of remaining in the same allelic state, $p_{11}$, $p_{22}$, and $p_{33}$. By the row-sum constraints of the transition matrix, this induces underestimation of all off-diagonal entries. The bias of SS is negligible in comparison. EM and GA have lower variance, reflected in a narrow interquartile range (IQR, the interval between the 25th and 75th percentiles, shown as dashed lines), yet the IQR does not cover the true parameter values, indicating that these estimators are overconfident and biased.

The estimated distributions of the selection coefficients are depicted in the last row of Figure~\ref{fig:WF_final_plot}. All estimators show similar performance. These parameters are harder to identify, since changes in their values produce trajectories that are nearly indistinguishable under the model likelihood, leading to higher estimation uncertainty for all methods.

Outliers were removed based on an interquartile range rule applied to the estimation errors. For each parameter, the central range containing the middle half of the errors was identified across all datasets. Any dataset for which at least one parameter error was more than three interquartile ranges below or above the central range was flagged as an outlier for that estimator. Using this rule removed about 4.5\% of the datasets for the SS estimator and 7\% for the GA estimator.

\subsection{The Student Kramers oscillator} 

Parameters for the simulation were set to $\eta_0 = 30$, $a_0 = -125$, $b_0 = 40$, $c_0 = 150$, $d_0 = -20$, $\alpha_0 = 20$, $\beta_0 = -8$, and $\gamma_0 = 1280.8$. The non-zero squared diffusion coefficient is thus $[\bm{\Sigma}\bm{\Sigma}^\top (x, v)]_{22} = 20(v-0.2)^2+1280$ and strictly positive for all $v \in \mathbb{R}$.

We simulate sample paths using the Milstein discretization scheme \citep{KloedenPlaten} with a step size of $h^{\mathrm{sim}} = 0.0001$ to ensure high accuracy. To reduce discretization errors, we sub-sample to obtain time steps $h = 0.01$, and $h = 0.02$. We fix the total time length to $T = 50$ and adjust the sample sizes for each $h$ accordingly. We repeat the simulations to obtain 1000 data sets.

To compute the asymptotic distribution of the SS estimator from Theorem~\ref{thm:AsymtoticNormality},
we evaluate $\mathbf{C}_{\bm{\theta}^{(1)}}(\bm{\theta}_0)$ and $\mathbf{C}_{\bm{\theta}^{(2)}}(\bm{\theta}_0)$
for the model \eqref{eq:StudentKramersSDE}. Since only the second component of $\mathbf{F}_0$ is stochastic,
straightforward calculation gives
\begin{align*}
    \mathbf{C}_{\bm{\theta}^{(1)}}(\bm{\theta}_0) =
    \mathbb{E}_{\nu_0}\!\left[
    \frac{1}{\alpha_0 V^2 + \beta_0 V + \gamma_0}
    \begin{bmatrix}
        V^2      & -VX^3  & -VX^2  & -VX   & -V    \\
        \star    & X^6    & X^5    & X^4   & X^3   \\
        \star    & \star  & X^4    & X^3   & X^2   \\
        \star    & \star  & \star  & X^2   & X     \\
        \star    & \star  & \star  & \star & 1
    \end{bmatrix}
    \right],
\end{align*}
and
\begin{align*}
    \mathbf{C}_{\bm{\theta}^{(2)}}(\bm{\theta}_0) =
    \frac{1}{2}\,\mathbb{E}_{\nu_0}\!\left[
    \frac{1}{(\alpha_0 V^2 + \beta_0 V + \gamma_0)^2}
    \begin{bmatrix}
        V^4   & V^3   & V^2 \\
        \star & V^2   & V   \\
        \star & \star & 1
    \end{bmatrix}
    \right],
\end{align*}
where $\star$ denotes the symmetric counterpart and all expectations are approximated by ergodic averages over the simulated trajectory. Theorem~\ref{thm:AsymtoticNormality} then yields
\begin{align*}
    \sqrt{Nh}\,\bigl(\hat{\bm{\theta}}^{(1)}_N - \bm{\theta}^{(1)}_0\bigr)
    \xrightarrow{d} \mathcal{N}\!\left(\mathbf{0},\, \mathbf{C}_{\bm{\theta}^{(1)}}(\bm{\theta}_0)^{-1}\right),
    \qquad
    \sqrt{N}\,\bigl(\hat{\bm{\theta}}^{(2)}_N - \bm{\theta}^{(2)}_0\bigr)
    \xrightarrow{d} \mathcal{N}\!\left(\mathbf{0},\, \mathbf{C}_{\bm{\theta}^{(2)}}(\bm{\theta}_0)^{-1}\right),
\end{align*}
as $N \to \infty$, $h \to 0$, $Nh \to \infty$ and $Nh^2 \to 0$. The asymptotic standard deviations displayed in Figure~\ref{fig:final_plot_SK} are the square roots of the diagonal entries of $\mathbf{C}_{\bm{\theta}^{(1)}}(\bm{\theta}_0)^{-1}/(Nh)$ and $\mathbf{C}_{\bm{\theta}^{(2)}}(\bm{\theta}_0)^{-1}/N$ for $\bm{\theta}^{(1)}$ and $\bm{\theta}^{(2)}$ respectively, where both matrices are estimated by ergodic averages.

We first outline the estimators tailored for the Student Kramers oscillator. We then detail the simulation procedure and describe the optimization process. Finally, we present and interpret the results.

\begin{figure}
    \centering
    \includegraphics[width=\textwidth]{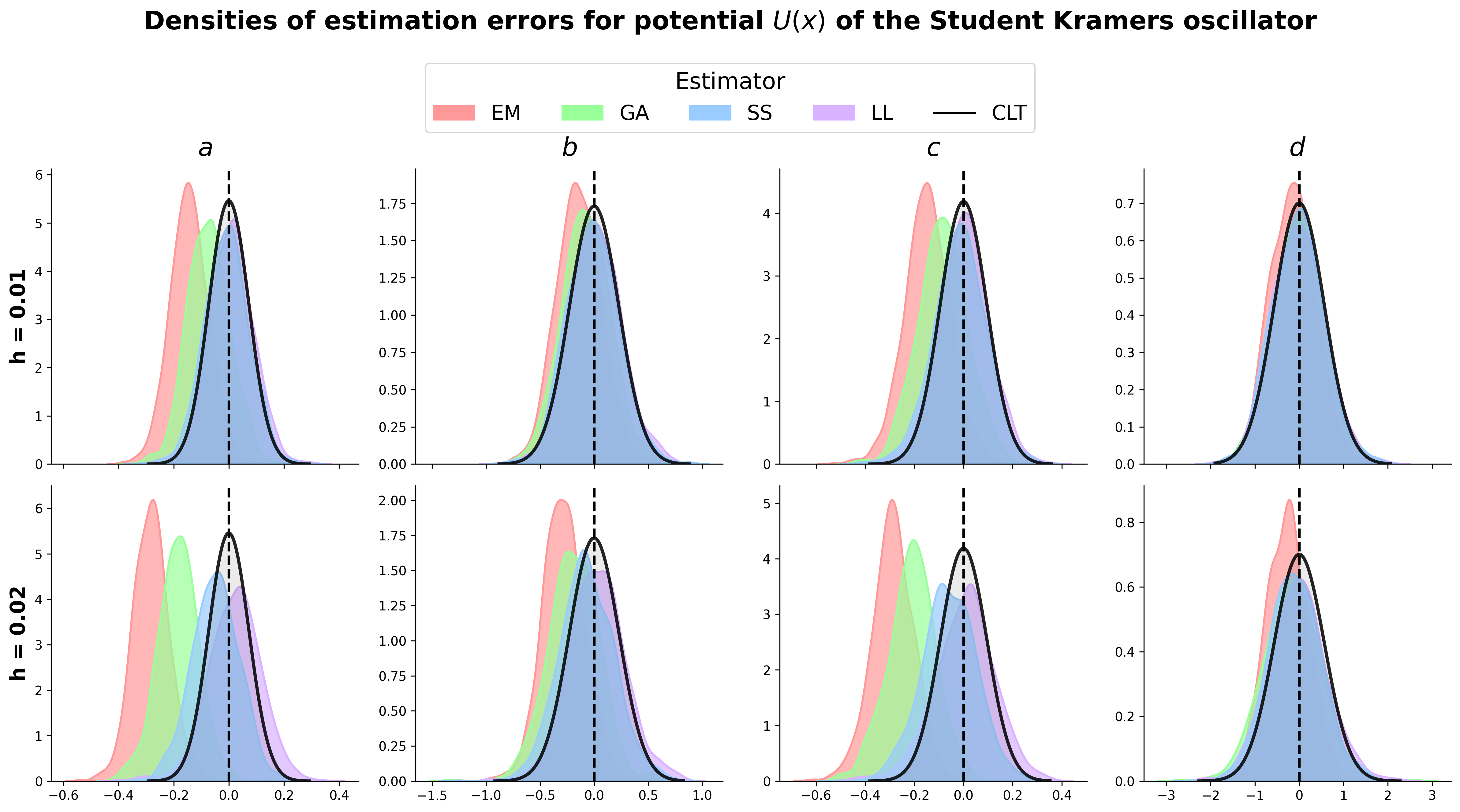}
    \includegraphics[width=\textwidth]{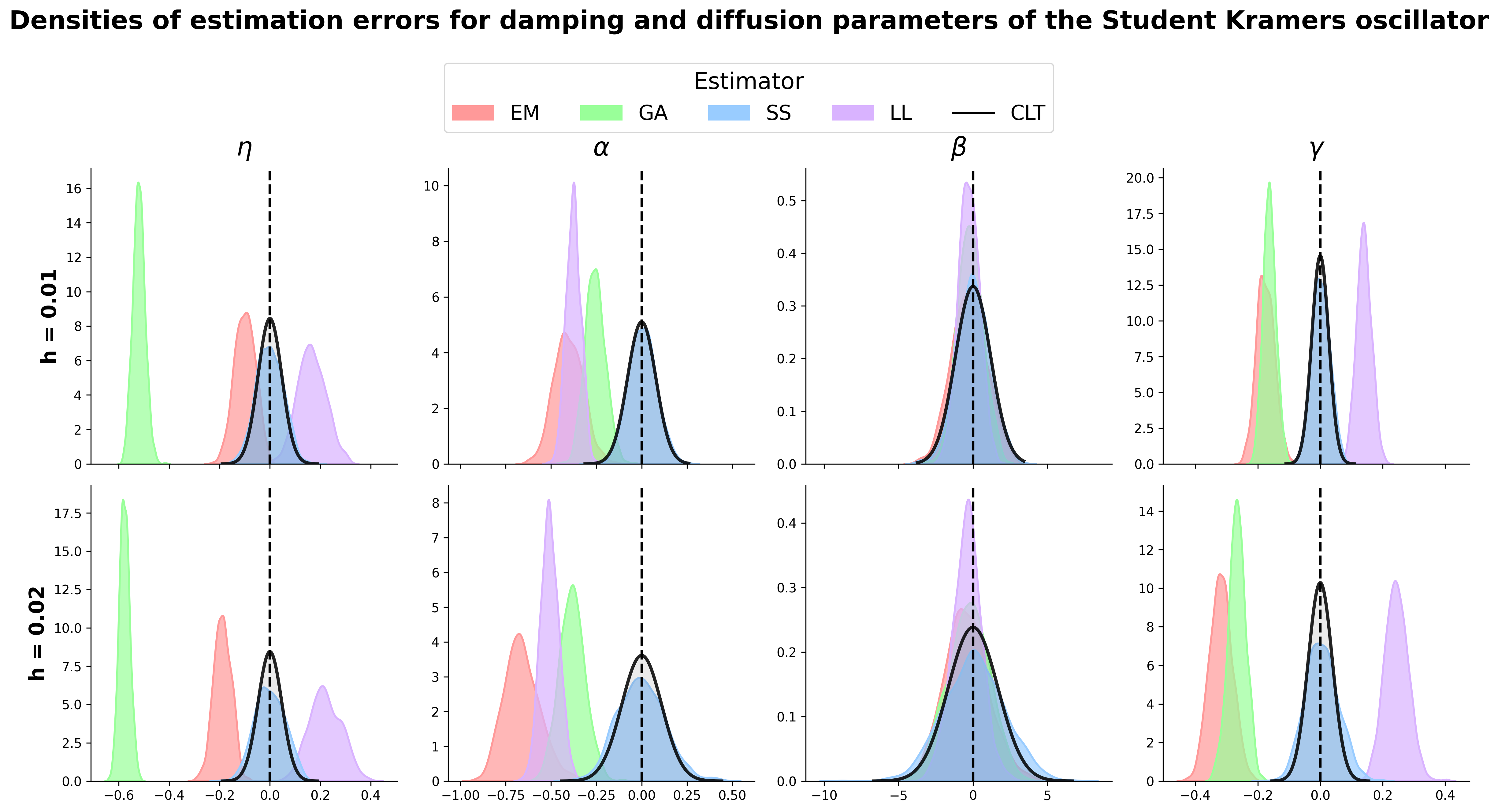}
    \caption{Student Kramers oscillator. Normalized distributions of parameter estimation errors $(\hat{\bm{\theta}}_N - \bm{\theta}_0) \oslash \bm{\theta}_0$ based on 1000 simulated datasets  with a fixed time interval of length $T = 50$. Different colors indicate the type of estimator. Each column corresponds to a different parameter, and each row corresponds to a different value of $h$, and consequently $N$. Black density lines represent the asymptotic distributions. \textbf{Upper panel:} Distributions of the errors for the potential $U(x)$ parameters. \textbf{Lower panel:} Distributions of the errors of the damping parameter $\eta$ and diffusion parameters.}
    \label{fig:final_plot_SK}
\end{figure} 

\begin{figure}
    \centering
    \includegraphics[width=.8\textwidth]{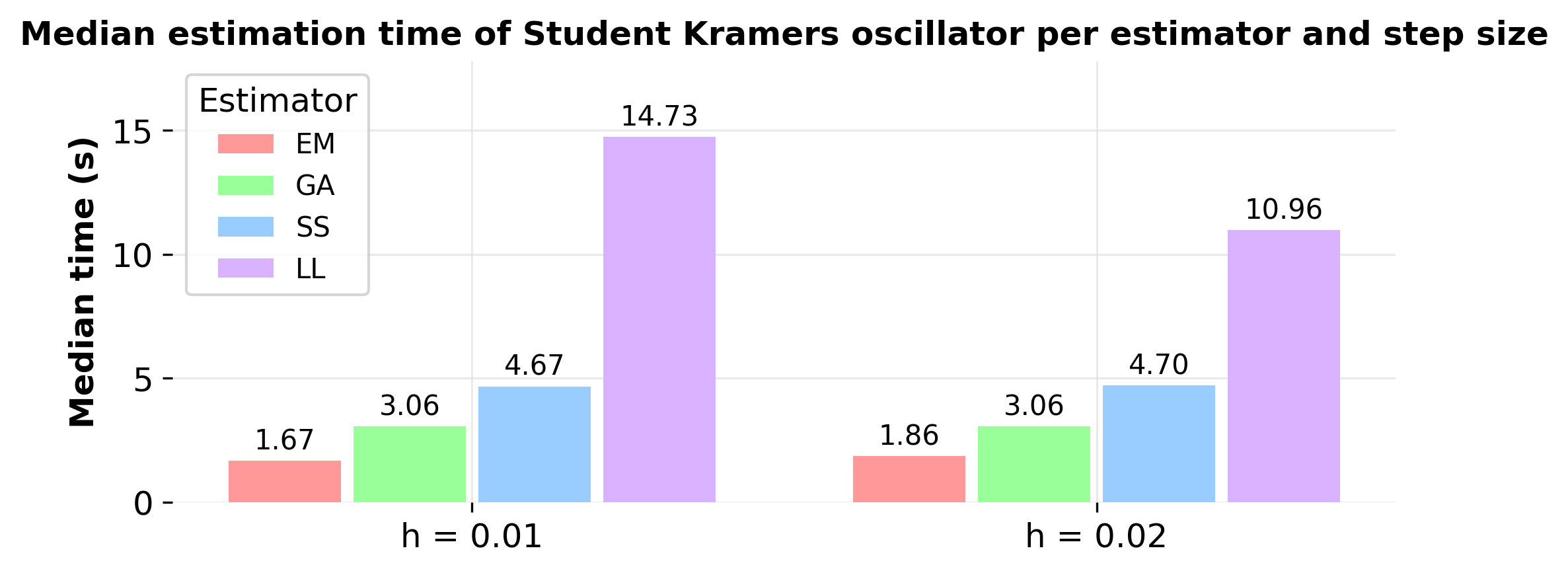}
    \caption{Median wall-clock estimation time (in seconds) to estimate parameters in the Student Kramers oscillator for step sizes $h \in \{0.01, 0.02\}$, based on 1000 simulated datasets ($T = 50$, $N \in \{5000, 2500\}$). Outliers and failed runs are excluded. Colors indicate estimator.}
    \label{fig:SK_estimation_times}
\end{figure} 

\subsubsection{Strang splitting estimator}

To define the SS estimator, we rewrite  \eqref{eq:StudentKramersSDE} as 
\begin{align} \label{eq:StudentKramerSDE2}
    \dif \begin{bmatrix}
        X_{t}\\
        V_{t} \\
        \end{bmatrix} = \underbrace{\begin{bmatrix}
        V_t\\
        -\eta V_t + a X_t^3 + b X_t^2 + cX_t + d\\
    \end{bmatrix}}_{\mathbf{F}(X_t, V_t)}\dif t + \begin{bmatrix}
        0 & 0\\
        0 & \sqrt{\alpha V_t^2 + \beta V_t + \gamma}\\
    \end{bmatrix}\dif W_t. 
\end{align}
As before, we choose $\mathbf{A}=\mathbb{E}[D \mathbf{F}(\mathbf{x})]$. Starting with
\begin{equation*}
    D \mathbf{F}(x, v) = \begin{bmatrix}
        0 & 1\\
        3 a x^2 +2  b x + c & - \eta \\
    \end{bmatrix},
\end{equation*}
we get
\begin{equation*}
    \mathbf{A} = \begin{bmatrix}
        0 & 1\\
        3 a \mathbb{E}[X_t^2] +2  b \mathbb{E}[X_t] + c & - \eta \\
    \end{bmatrix}.
\end{equation*}
For model \eqref{eq:StudentKramerSDE2}, Eq.~\eqref{eq:findingb} is equivalent to solving 
\begin{equation*}
    3 a b_x^2 +2  b b_x - 3 a \mathbb{E}[X_t^2] - 2  b \mathbb{E}[X_t] = 0.
\end{equation*}
The solution is 
\begin{equation*}
    b_x^{\pm} = \frac{-b \pm \sqrt{b^2 + 3 a (3 a \mathbb{E}[X_t^2] +2  b \mathbb{E}[X_t])}}{3 a} = -\frac{b}{3a} \pm \sqrt{\text{var}[X_t] + \left(\mathbb{E}[X_t] + \frac{b}{3a}\right)^2 },
\end{equation*}
which is well defined for all parameter values. This did not provide a solution for $b_v$, but the expectation of $V_t$ is zero, so we set $b_v = 0$. Thus, \eqref{eq:StudentKramerSDE2} can be written as 

\begin{equation} \label{eq:StudentKramerSDE3}
\begin{aligned} 
    \dif \begin{bmatrix}
        X_{t}\\
        V_{t} \\
        \end{bmatrix} &= \begin{bmatrix}
        0 & 1\\
        3 a \mathbb{E}[X_t^2] +2  b \mathbb{E}[X_t] + c & - \eta \\
    \end{bmatrix} \left(\begin{bmatrix}
        X_{t}\\
        V_{t} \\
        \end{bmatrix} - \begin{bmatrix}
         b_x^{\pm}\\
        0 \\
        \end{bmatrix}\right) \dif t + \begin{bmatrix}
        0 & 0\\
        0 & \sqrt{\alpha V_t^2 + \beta V_t + \gamma}\\
    \end{bmatrix}\dif W_t  \\
        &+ \underbrace{\begin{bmatrix}
        0\\
        a X_t^3 + b X_t^2 + cX_t + d - (3 a \mathbb{E}[X_t^2] +2  b \mathbb{E}[X_t] + c)(X_t -  b_x^{\pm})\\
    \end{bmatrix}}_{\mathbf{N}(X_t, V_t)}\dif t. 
\end{aligned}
\end{equation}

The nonlinear ODE driven by $\mathbf{N}(x,v)$ has a trivial solution where $x$ is a constant and $v$ is linear in time. We incorporate these components into the objective function \eqref{eq:S_obj} to obtain the SS estimator.

For the GA and LL estimators for the Student Kramers oscillator, see Supplementary Material \ref{sec:Implementation}.

\subsubsection{Optimization}

For this model, we use the Limited-memory BFGS (L-BFGS) algorithm provided by the \texttt{jaxopt} library \citep{jaxopt_implicit_diff}. To ensure robust convergence and strictly decreasing steps, we use a strong Wolfe line search condition. The optimization is restricted to a maximum of \num{1000} iterations, and optimization is stopped when the Euclidean norm of the parameter difference between consecutive iterations falls below $10^{-5}$. The initial parameter vector is set to $\boldsymbol{\theta}_{\text{init}} = (50, -200, 10, 100, 10, 30, -5, 1000)$

\subsubsection{Results}

Figure~\ref{fig:final_plot_SK} presents the normalized estimation errors $(\hat{\bm{\theta}}_N - \bm{\theta}_0) \oslash \bm{\theta}_0$ of the simulations, where $\oslash$ denotes element-wise division. The upper panel shows the estimators of the parameters of the potential function $U(x)$, the lower panel displays the damping and diffusion parameters.

For the potential parameters, all estimators perform reasonably well for both step sizes. However, EM and GA exhibit visibly larger bias, especially for the coarser step size $h=0.02$, while SS and LL display similar bias and variability across both values of $h$. For the damping and diffusion parameters, all methods except SS show substantial bias for $\eta$, $\alpha$, and $\gamma$ at both step sizes.

The SS estimator comes closest to the theoretical asymptotic regime. For the finer step size $h=0.01$, the empirical distribution of SS aligns well with the theoretical Gaussian density, with the variance only slightly exceeding the asymptotic prediction. For the coarser step size $h=0.02$, the variance of SS is more overestimated relative to the theoretical counterpart, indicating that the asymptotic regime has not yet been fully reached. Nevertheless, SS remains the least biased estimator at both step sizes, and its variance is consistently closer to the theoretical one than that of EM, GA, and LL.

Outliers were removed using a two-pass interquartile range (IQR) rule applied to the relative estimation errors. For each parameter, we identified the central region containing the middle half of the errors and flagged any dataset where at least one estimate fell more than 1.5 times the IQR outside this region. This removed less than 1\% of the datasets for the EM, SS, and LL estimators, compared to approximately 1\% at $h = 0.01$ and 45\% at $h = 0.02$ for the GA estimator. 

Figure~\ref{fig:SK_estimation_times} shows the median computation times of the Student Kramers oscillator per estimator for the two step sizes ($h = 0.02$ and $h = 0.2$), computed after removing NA and outlier runs. EM, GA SS converge considerably faster than LL across both settings.

\section{Application to Greenland Ice Core Data} \label{sec:Greenland}

The Greenland ice core data \citep{RASMUSSEN201414} consist of calcium ion concentration ($\text{Ca}^{2+}$) measurements from the GRIP ice core, uniformly spaced at 20-year intervals via binning and averaging, covering the last glacial period from approximately 115,000 to 12,000 years before the present. The $\text{Ca}^{2+}$ series serves as a proxy for paleotemperature and captures the abrupt Dansgaard–Oeschger (DO) warming events characteristic of this period. Since the data represent averages over time intervals, a second-order SDE is a natural modeling framework \citep{ditlevsen2002fast, DitlevsenSorensen2004}.

In \citep{pilipovic2024SecondOrder}, the Kramers oscillator with additive noise was fitted to this data. The model fit, assessed through the empirical invariant distribution against the fitted theoretical invariant distribution, was poor, suggesting that the model lacks sufficient flexibility to capture the distributional features of the data. In particular, the velocity variable was poorly captured and showed a heavier tail than what the additive noise model could produce. This motivates fitting the Student Kramers oscillator \eqref{eq:StudentKramersSDE}, which extends the model by allowing state-dependent and heavy-tailed noise in the velocity component.

We use the same data pre-processing as in \citep{pilipovic2024SecondOrder}, and refer the reader there for details. Since the proposed estimator assumes complete observations of both $X_t$ and $V_t$, but only $X_t$ is directly observed, we approximate $V_t$ by forward differences. To account for the errors induced by this finite difference approximation, we add correction terms to the pseudo-likelihood (see \cite{pilipovic2024SecondOrder} for a thorough discussion on the imputation of $V_t$).

We consider three nested models of increasing complexity, summarised in Table~\ref{tab:model_comparison}. $\mathrm{M}_1$, with $b = d = \alpha = \beta = 0$, corresponds to a symmetric double-well potential with constant diffusion. The negative log-likelihood for $\mathrm{M}_1$ shows that the estimates in \citep{pilipovic2024SecondOrder} ($\widehat{\eta} = 62.5$, $\widehat{a} = -219$, $\widehat{d} = 297$, $\widehat{\gamma} = 9125$, and $\text{NLL} = 8708.05$) are not optimal. The intermediate model $\mathrm{M}_2$ retains Student-type noise by freeing $\alpha$ and $\beta$, but keeps $b = d = 0$, so the potential $U(x)$ is symmetric. The full model $\mathrm{M}_3$ frees all eight parameters.

\begin{figure}
    \centering
    \includegraphics[width = \textwidth]{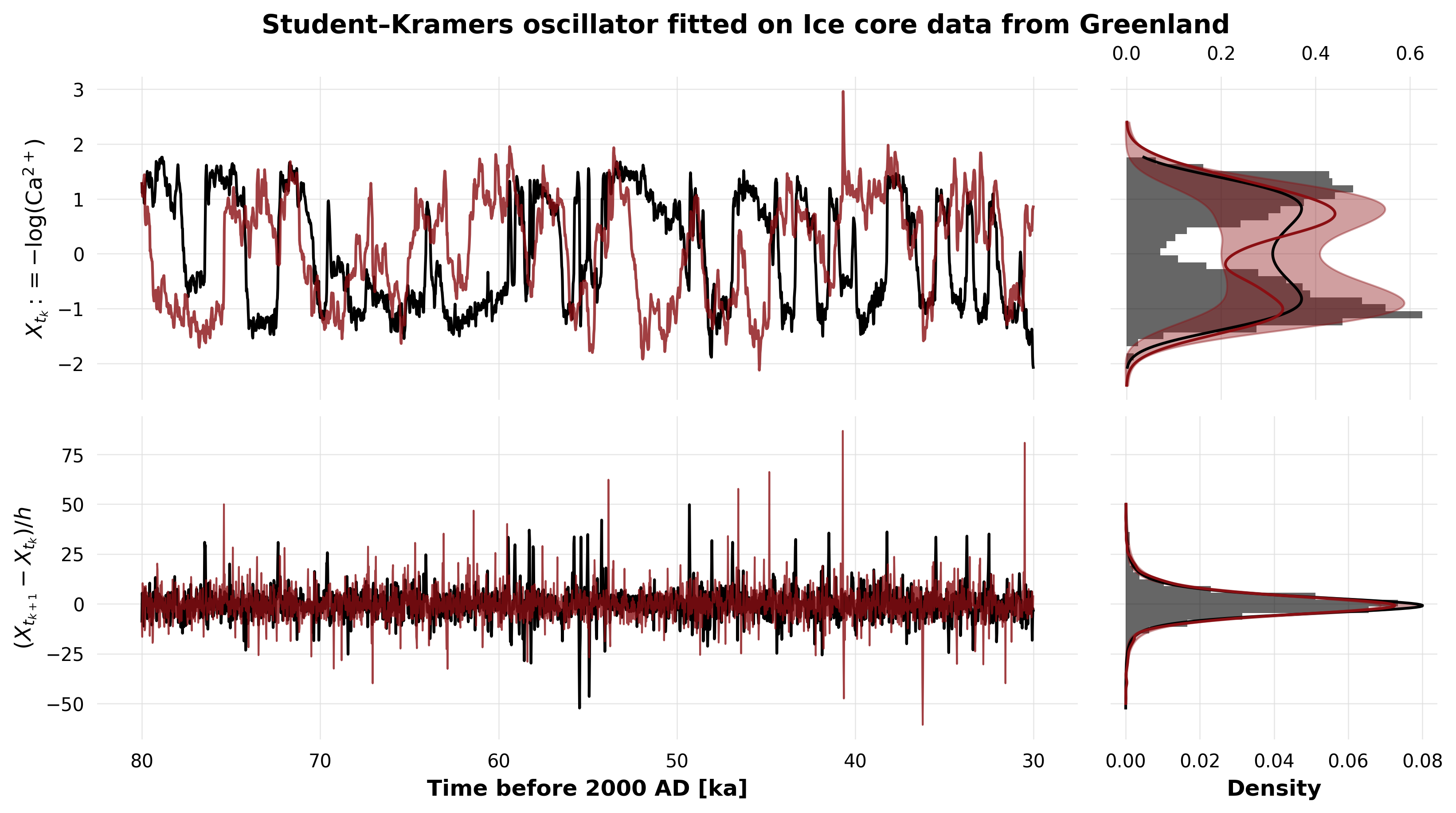}
    \caption{\textbf{Ice core data from Greenland.} \textbf{Left:} Trajectories (black) over time (in 1000 years) of the centered negative logarithm of the $\text{Ca}^{2+}$ measurements (top) and forward difference approximations of its rate of change (bottom). Simulated trajectory from the estimated model (dark red). \textbf{Right:} Histograms of marginal distributions (gray) and estimated approximated invariant densities (black line) with prediction intervals (dark red ribbons), and the empirical density of a simulated sample from the estimated model (dark red lines).}
    \label{fig:ice_core}
\end{figure}

\begin{table}[h]
\centering
\caption{Parameter estimates and negative log-likelihood (NLL) for the three fitted models.}
\label{tab:model_comparison}
\begin{tabular}{lrrrrrrrrr}
\toprule
Model & $\hat{\eta}$ & $\hat{a}$ & $\hat{b}$ & $\hat{c}$ & $\hat{d}$ & $\hat{\alpha}$ & $\hat{\beta}$ & $\hat{\gamma}$ & NLL \\
\midrule
$\mathrm{M}_1$ $(b = d = \alpha = \beta = 0)$  & 57.7 & $-$130.0 & -- & 115.3 & -- & -- & -- & 8322.5 & 8673.11 \\
$\mathrm{M}_2$ $(b = d = 0)$  & 67.9 & $-$96.6  & -- & 65.4  & -- & 64.3 & 233.2 & 4387.8 & 8524.35 \\
$\mathrm{M}_3$ & 67.9 & $-$97.8  & 8.4 & 66.2 & $-$9.1 & 64.4 & 228.0 & 4371.0 & 8524.06 \\
\bottomrule
\end{tabular}
\end{table}

To select among the models, we first performed a bootstrap likelihood ratio test comparing models $\mathrm{M}_1$ and $\mathrm{M}_2$, obtaining a $p$-value of effectively zero, which leads us to reject the hypothesis of constant diffusion in favor of Student-type noise. The negative log-likelihood difference between models $\mathrm{M}_2$ and $\mathrm{M}_3$ is negligible ($8524.35$ vs.\ $8524.06$). This is confirmed by the bootstrap confidence intervals for $b$ and $d$ that both contains zero under $\mathrm{M}_3$, providing no evidence that the asymmetric potential terms improve the fit. We therefore adopt the  $\mathrm{M}_2$ with $b = d = 0$ as our final model.

To interpret the noise structure, we approximate the invariant distribution of $V_t$ by the skew $t$-distribution \eqref{eq:skew_student} with parameters estimated from $\mathrm{M}_2$
\begin{equation*}
    \hat{\nu} = \frac{2\hat{\eta}}{\hat{\alpha}} + 1 \approx 3.11, \quad \hat{\mu} = -\frac{\hat{\beta}}{2\hat{\alpha}} \approx -1.81, \quad \hat{\nu}\hat{\sigma}^2 = \frac{4\hat{\alpha}\hat{\gamma} - \hat{\beta}^2}{4\hat{\alpha}^2} \approx 65, \quad \hat{\omega} = \frac{2\hat{\beta}\hat{\eta}}{\hat{\alpha}\sqrt{4\hat{\alpha}\hat{\gamma} - \hat{\beta}^2}} \approx 0.475.
\end{equation*}
The estimated degrees of freedom $\hat{\nu} \approx 3.11$ indicate heavy tails, consistent with the motivation for this model: the additive Gaussian noise of the standard Kramers oscillator failed to capture the spread in the observed velocity, and the Student-type noise provides a substantially better account of this heavy-tailed behaviour. The location parameter $\hat{\mu} \approx -1.81$ indicates that the velocity distribution is centred slightly below zero, while the scale $\hat{\nu}\hat{\sigma}^2 \approx 65$ reflects the large spread of the velocity process. The skewness parameter $\hat{\omega} \approx 0.475$ indicates a moderate positive skew, meaning the right tail of the approximating velocity distribution is heavier than the left. This implies that the invariant distribution becomes non-symmetric, even if the potential is symmetric, justifying that parameters $b$ and $d$ are set to zero even if the empirical distribution is skewed.

The model fit is assessed in the right panels of Figure~\ref{fig:ice_core}. We present the empirical invariant distributions of both coordinates alongside the fitted theoretical invariant distribution. A $95\%$ prediction interval is constructed by simulating 1000 datasets from the fitted model matching the size of the empirical data, and taking pointwise 2.5th and 97.5th percentiles of the resulting approximated densities. A single example trajectory is shown in red. Note that the theoretical invariant densities are approximated under the assumption that $X_t$ and $V_t$ are decoupled.

Compared to the standard Kramers oscillator $\mathrm{M}_1$, the Student Kramers oscillator $\mathrm{M}_2$ improves the fit for the heavy-tailed velocity variable $V_t$. However, the fit for $X_t$ becomes worse since the estimated stable equilibria are closer together than the observations, leading to the model placing substantial probability mass between the two states. Thus, the implied transitions are smoother and more gradual than the nearly discontinuous jumps seen in the data.

The inability to reproduce abrupt transitions, suggests that the model is still too simple to fully capture the dynamics, and that more complex models may be necessary to account for the abrupt nature of the Dansgaard--Oeschger transitions.

\section{Technical details} \label{sec:TechnicalDetails}

In this Section we prove Theorem~\ref{thm:Covariance_matrix} and recall other estimators used in the simulation study.

\begin{proof}[Proof of Theorem \ref{thm:Covariance_matrix}]

For $\bm{\phi}(t, \mathbf{x}) \in C^{1,2}(\mathbb{R}_+ \times \mathbb{R}^d)$, It\^o's formula yields 
\begin{equation} \label{eq:difExp}
    \frac{\dif }{\dif t} \mathbb{E}[\bm{\phi}(t, \mathbf{X}_t)] = \mathbb{E}\left[\frac{\partial \bm{\phi}(t, \mathbf{X}_t)}{\partial t}\right] + \sum_{i=1}^{d} \mathbb{E}\left[\frac{\partial \bm{\phi}(t, \mathbf{X}_t)}{\partial x^{(i)}} F^{(i)}(\mathbf{X}_t)\right] + \frac{1}{2} \sum_{i, j=1}^d \mathbb{E}\left[\frac{\partial^2 \bm{\phi}(t, \mathbf{X}_t)}{\partial x^{(i)} \partial x^{(j)}} [\bm{\Sigma}\bm{\Sigma}^\top(\mathbf{X}_t)]_{ij}\right].
\end{equation}
To find the mean vector $\mathbf{m}(t) = \mathbb{E}[\mathbf{X}_t]$, we set $\bm{\phi}(t, \mathbf{x}) = \mathbf{x}$ and obtain
\begin{equation}
    \frac{\dif \mathbf{m}(t)}{\dif t} = \mathbb{E}[\mathbf{F}(\mathbf{X}_t)] = \mathbf{A}(\mathbf{m}(t) - \mathbf{b}). \label{eq:MeanODE}
\end{equation}
The solution of the linear ODE \eqref{eq:MeanODE} is \eqref{eq:MeanODESolution}. 

For the covariance matrix $\mathbf{C}(t) = \mathbb{E}[(\mathbf{X}_t - \mathbf{m}(t))(\mathbf{X}_t - \mathbf{m}(t))^\top]$, setting $\bm{\phi}(t, \mathbf{x}) = \mathbf{x}\mathbf{x}^\top - \mathbf{m}(t)\mathbf{m}(t)^\top$, we derive
\begin{equation} \label{eq:CovODE}
\begin{aligned}
    \frac{\dif \mathbf{C}(t)}{\dif t} &= \mathbb{E}[\mathbf{F}(\mathbf{X}_t) (\mathbf{X}_t - \mathbf{m}(t))^\top] + \mathbb{E}[(\mathbf{X}_t - \mathbf{m}(t)) \mathbf{F}(\mathbf{X}_t)^\top] + \mathbb{E}[\bm{\Sigma}\bm{\Sigma}^\top(\mathbf{X}_t)] \\
    &= \mathbf{A} \mathbf{C}(t) + \mathbf{C}(t) \mathbf{A}^{\top} + \mathbb{E}[\bm{\Sigma}\bm{\Sigma}^\top(\mathbf{X}_t)]. 
\end{aligned}
\end{equation}
Equation \eqref{eq:CovODE} is called the differential Sylvester equation with a unique solution given initial condition $\mathbf{C}(0)$ (Theorem 2 in \citep{behr2019solution}),
\begin{equation} \label{eq:Sylvester_sol}
    \mathbf{C}(t) = \exp(\mathbf{A} t) \mathbf{C}(0) \exp(\mathbf{A}^\top t) + \int_0^t \exp(\mathbf{A}(t-s))\mathbb{E}[\bm{\Sigma}\bm{\Sigma}^\top(\mathbf{X}_s)] \exp(\mathbf{A}^\top(t-s)) \dif s.
\end{equation}
For general $\bm{\Sigma}\bm{\Sigma}^\top(\mathbf{x})$ as defined in \eqref{eq:SigmaSigmaT}, the integral in \eqref{eq:Sylvester_sol} does not have a closed form solution. However, we provide an explicit computation of $\vect(\mathbf{C}(t))$.

Vectorizing \eqref{eq:CovODE} and using \eqref{eq:SigmaSigmaT_vec}, we obtain
\begin{align*}
    \frac{\dif \vect(\mathbf{C}(t))}{\dif t} &= \vect(\mathbf{A} \mathbf{C}(t)) + \vect(\mathbf{C}(t) \mathbf{A}^{\top}) + \mathbb{E}[\vect(\bm{\Sigma}\bm{\Sigma}^\top(\mathbf{X}_t))]\\
    &= (\mathbf{I}\otimes \mathbf{A}) \vect(\mathbf{C}(t)) + (\mathbf{A} \otimes \mathbf{I})\vect(\mathbf{C}(t)) + \check{\bm{\alpha}} \vect(\mathbb{E}[\mathbf{X}_t\mathbf{X}_t^\top]) + \check{\bm{\beta}} \vect(\mathbb{E}[\mathbf{X}_t]) + \check{\bm{\gamma}}\\
    &=(\mathbf{A} \oplus \mathbf{A}  + \check{\bm{\alpha}})\vect(\mathbf{C}(t)) + \check{\bm{\alpha}} \vect(\mathbf{m}(t)\mathbf{m}(t)^\top) + \check{\bm{\beta}} \mathbf{m}(t) + \check{\bm{\gamma}}\\
    &=(\mathbf{A} \oplus \mathbf{A}  + \check{\bm{\alpha}})\vect(\mathbf{C}(t)) + \vect(\bm{\Sigma}\bm{\Sigma}^\top(\mathbf{m}(t))).
\end{align*}
The solution to this linear ODE, given initial condition $\mathbf{C}(0)$, is
\begin{align}
    \vect(\mathbf{C}(t)) &= \exp(( \mathbf{A} \oplus \mathbf{A} + \check{\bm{\alpha}})t) \vect(\mathbf{C}(0))+\int_0^t \exp((\mathbf{A} \oplus \mathbf{A} + \check{\bm{\alpha}})(t-s))  \vect(\bm{\Sigma}\bm{\Sigma}^\top(\mathbf{m}(s))) \dif s. \label{eq:CovODESolution}
\end{align}
We could expand $\vect(\bm{\Sigma}\bm{\Sigma}^\top(\mathbf{m}(s)))$ in a Taylor series around $s\to 0$ and apply Theorem 1 from \cite{Carbonell2008}. However, if we substitute $\mathbf{m}(s)$ from \eqref{eq:MeanODESolution} into \eqref{eq:CovODESolution} we could directly use Theorem 1 from \cite{VanLoan1978} without the need for a Taylor approximation,
\begin{align*}
    &\int_0^t \exp((\mathbf{A} \oplus \mathbf{A} + \check{\bm{\alpha}})(t-s)) \vect(\bm{\Sigma}\bm{\Sigma}^\top(\mathbf{m}(s))) \dif s\notag\\
    &=\int_0^t \exp((\mathbf{A} \oplus \mathbf{A} + \check{\bm{\alpha}})(t-s)) (\check{\bm{\alpha}} \vect(\mathbf{m}(s)\mathbf{m}(s)^\top) + \check{\bm{\beta}} \mathbf{m}(s) + \check{\bm{\gamma}}) \dif s\notag\\
    &= \mathbf{I}_1(t, \mathbf{A}, \check{\bm{\alpha}}) \vect((\mathbf{m}(0) - \mathbf{b})(\mathbf{m}(0) - \mathbf{b})^\top) +  
    \mathbf{I}_2(t, \mathbf{A}, \check{\bm{\alpha}}) \vect((\mathbf{m}(0) - \mathbf{b})\mathbf{b}^\top)\notag\\
    &+  \mathbf{I}_3(t, \mathbf{A}, \check{\bm{\alpha}})  \vect(\mathbf{b}(\mathbf{m}(0) - \mathbf{b})^\top)
    +  \mathbf{I}_4(t, \mathbf{A}, \check{\bm{\alpha}}, \check{\bm{\beta}})  (\mathbf{m}(0) - \mathbf{b})
    +  \mathbf{I}_5(t, \mathbf{A}, \check{\bm{\alpha}}) \vect(\bm{\Sigma}\bm{\Sigma}^\top(\mathbf{b})). 
\end{align*}
\end{proof}

For convenience, we briefly explain the other estimators used in the simulation study.

\subsection{Euler-Maruyama estimator}\label{sec:EM}

The EM method uses a first-order Taylor expansion of the SDE \eqref{eq:SDE}
\begin{align*}
    \mathbf{X}_{t_k}^\mathrm{[EM]} \coloneqq \mathbf{X}_{t_{k-1}}^\mathrm{[EM]} + h \mathbf{F}(\mathbf{X}_{t_{k-1}}^\mathrm{[EM]}; \bm{\theta}^{(1)}) + \bm{\xi}_{h,k}^\mathrm{[EM]}(\mathbf{X}_{t_{k-1}}^\mathrm{[EM]}; \bm{\theta}^{(2)}), 
\end{align*}
where $\bm{\xi}_{h,k}^\mathrm{[EM]}(\mathbf{X}_{t_{k-1}}^\mathrm{[EM]}; \bm{\theta}^{(2)})\mid \mathbf{X}_{t_{k-1}}^\mathrm{[EM]} = \mathbf{x}_{k-1} \sim \mathcal{N}_d(\bm{0}, h\bm{\Sigma}\bm{\Sigma}^\top(\mathbf{x}_{k-1}; \bm{\theta}^{(2)}))$  \citep{KloedenPlaten}. The transition density $p^\mathrm{[EM]}(\mathbf{X}_{t_k} \mid \mathbf{X}_{t_{k-1}}; \bm{\theta})$ is Gaussian, so the pseudo-likelihood follows trivially.

\subsection{Gaussian approximation estimator}\label{sec:GA}

The GA estimator assumes Gaussian transition densities $p^\mathrm{[GA]}(\mathbf{X}_{t_k} \mid \mathbf{X}_{t_{k-1}}; \bm{\theta})$ with the true mean and covariance of the solution $\mathbf{X}_t$ \citep{Kessler1997}, that is, the density of SDE \eqref{eq:SDE} is approximated by 
\begin{equation*}
    p^{[\mathrm{GA}]}(\mathbf{x}_{t_k} \mid \mathbf{x}_{t_{k-1}}) = \mathcal{N}(\mathbf{x}_{t_k}; \bm{\mu}_h^{[\mathrm{GA}]}(\mathbf{x}_{t_{k-1}};\bm{\theta}^{(1)}), \bm{\Omega}_h^{[\mathrm{GA}]}(\mathbf{x}_{t_{k-1}}; \bm{\theta})).
\end{equation*}
When the moments are unknown, they are approximated using the infinitesimal generator $\mathbb{L}$ \eqref{eq:L} and formula \eqref{eq:expendingL}. Then, the GA second-order approximation becomes
\begin{equation}
\label{eq:GA}
    \begin{aligned}
    \mathbf{X}_{t_k}^\mathrm{[GA]} &\coloneqq \bm{\mu}_h^\mathrm{[GA]}(\mathbf{X}_{t_{k-1}}^\mathrm{[GA]}; \bm{\theta}) + \bm{\xi}_{h}^\mathrm{[GA]}(\mathbf{X}_{t_{k-1}}^\mathrm{[GA]}; \bm{\theta}),
\end{aligned}
\end{equation}
where $\bm{\xi}_{h}^\mathrm{[GA]}(\mathbf{X}_{t_{k-1}}^\mathrm{[GA]}; \bm{\theta}) \mid \mathbf{X}_{t_{k-1}}^\mathrm{[GA]} = \mathbf{x}_{k-1} \sim \mathcal{N}_d(\bm{0}, \bm{\Omega}_{h}^\mathrm{[GA]}(\mathbf{x}_{{k-1}}; \bm{\theta}))$, and 
\begin{align}
    \bm{\mu}_h^\mathrm{[GA]}(\mathbf{x}; \bm{\theta}) &= \mathbf{x} + h \mathbf{F}(\mathbf{x}; \bm{\theta}^{(1)}) + \frac{h^2}{2} \mathbb{L}\mathbf{F}(\mathbf{x}; \bm{\theta}^{(1)}), \notag\\
    \bm{\Omega}_{h}^\mathrm{[GA]}(\mathbf{x}; \bm{\theta}) &= h \bm{\Sigma} \bm{\Sigma}^\top(\mathbf{x}; \bm{\theta}^{(2)})\notag\\
    &+ \frac{h^2}{2} \left(D\mathbf{F}(\mathbf{x}; \bm{\theta}^{(1)}) \bm{\Sigma} \bm{\Sigma}^\top(\mathbf{x}; \bm{\theta}^{(2)}) + \bm{\Sigma} \bm{\Sigma}^\top(\mathbf{x}; \bm{\theta}^{(2)}) D^\top\mathbf{F}(\mathbf{x}; \bm{\theta}^{(1)}) + \mathbb{L}\bm{\Sigma}\bm{\Sigma}^\top(\mathbf{x}; \bm{\theta}^{(2)})\right). \label{eq:OmegaGa}
\end{align}

While this method is straightforward, a few problems exist in practice.

First, obtaining closed-form formulas for $\bm{\mu}_h^\mathrm{[GA]}$ and $\bm{\Omega}_{h}^\mathrm{[GA]}$ becomes more complex for higher-order approximations. Thus, the second-order approximation is the most commonly used in practice. However, for hypoelliptic systems, we need the third-order approximation of $\bm{\Omega}_{h}^\mathrm{[GA]}$. We do not provide a general formula for the third-order approximation. Still, we calculate $\bm{\Omega}_{h}^\mathrm{[GA]}$ up to order $h^3$ using \texttt{Wolfram Mathematica} for the Wright-Fisher diffusion and the Student Kramers oscillator.

Second, $\bm{\Omega}_{h}^\mathrm{[GA]}$ does not need to be positive definite. To avoid this problem, \cite{Kessler1997} suggested Taylor expanding $\log\det \bm{\Omega}_{h}^\mathrm{[GA]}$ and $(\bm{\Omega}_{h}^\mathrm{[GA]})^{-1}$ in $h$.

\subsection{Ozaki's local linearization estimator}\label{sec:LL}

Ozaki's LL method assumes an SDE with additive noise
\begin{equation}
    \dif \mathbf{X}_t = \mathbf{F}(\mathbf{X}_t; \bm{\theta}^{(1)}) \dif t + \bm{\Sigma} \dif \mathbf{W}_t.\label{eq:SDEreduced}
\end{equation}
If the starting SDE \eqref{eq:SDE} is reducible, that is, if there is a bijective transformation between \eqref{eq:SDE} and \eqref{eq:SDEreduced}, then we can use the LL estimator. Otherwise, we apply LL on \eqref{eq:SDE}, such that locally on the interval $[t_k, t_{k+1})$, the diffusion matrix $\bm{\Sigma}(\mathbf{X}_t)$ is constant, i.e., $\bm{\Sigma}(\mathbf{X}_t) = \bm{\Sigma}(\mathbf{X}_{t_k})$, for $t\in [t_k, t_{k+1})$. 

Here we describe the latter approach. To approximate \eqref{eq:SDE}, we first approximate the drift of \eqref{eq:SDE} between consecutive observations by a linear function, and then we find the closed-form solution of the resulting linear SDE (see \citep{JimenezLL99}). The approximation becomes
\begin{align}
    \mathbf{X}_{t_k}^\mathrm{[LL]} &\coloneqq \bm{\mu}_h^\mathrm{[LL]}(\mathbf{X}_{t_{k-1}}^\mathrm{[LL]}; \bm{\theta}) + \bm{\xi}_{h}^\mathrm{[LL]}(\mathbf{X}_{t_{k-1}}^\mathrm{[LL]}; \bm{\theta}), \label{eq:LL}
\end{align}
where $\bm{\xi}_{h}^\mathrm{[LL]}(\mathbf{X}_{t_{k-1}}^\mathrm{[LL]}; \bm{\theta}) \mid \mathbf{X}_{t_{k-1}} = \mathbf{x}_{k-1} \sim \mathcal{N}_d(\bm{0}, \bm{\Omega}_{h}^\mathrm{[LL]}(\mathbf{x}_{k-1}; \bm{\theta}))$, and
\begin{align}
    \bm{\mu}_h^\mathrm{[LL]}(\mathbf{x}; \bm{\theta}) &\coloneqq \mathbf{x} + \mathbf{R}_{h,0}( D\mathbf{F}(\mathbf{x}; \bm{\theta}^{(1)}))\mathbf{F}(\mathbf{x}; \bm{\theta}^{(1)}) + (h \mathbf{R}_{h, 0}( D\mathbf{F}(\mathbf{x}; \bm{\theta}^{(1)})) - \mathbf{R}_{h, 1}( D\mathbf{F}(\mathbf{x}; \bm{\theta}^{(1)}))) \mathbf{M}(\mathbf{x}; \bm{\theta}), \notag\\
    \bm{\Omega}_{h}^\mathrm{[LL]}(\mathbf{x}; \bm{\theta}) &\coloneqq \int_0^h e^{D\mathbf{F}(\mathbf{x}; \bm{\theta}^{(1)})(h-u)} \bm{\Sigma} \bm{\Sigma}^\top(\mathbf{x}) e^{ D\mathbf{F}(\mathbf{x}; \bm{\theta}^{(1)})^\top (h-u)} \dif u. \label{eq:OmegaLL}
\end{align}
Previously we defined
\begin{align*}
    \mathbf{R}_{h, r}(D\mathbf{F}(\mathbf{x}; \bm{\theta}^{(1)})) \coloneqq \int_0^h \exp(D\mathbf{F}(\mathbf{x}; \bm{\theta}^{(1)}) u) u^r \dif u, \qquad
    \mathbf{M}(\mathbf{x}; \bm{\theta}) \coloneqq \frac{1}{2}\sum_{i,j = 1}^d \partial^2_{x^{(i)}x^{(j)}} \mathbf{F}(\mathbf{x}) [\bm{\Sigma}\bm{\Sigma}^\top(\mathbf{x})]_{ij},
\end{align*}
for $r = 0,1$. As before, we compute integrals of matrix exponentials $\mathbf{R}_{h, r}$ and $\bm{\Omega}_{h, k}^\mathrm{[LL]}(\bm{\theta})$ using formulas from \citep{VanLoan1978}. Thus, $p^\mathrm{[LL]}(\mathbf{x}_{t_k} \mid \mathbf{x}_{t_{k-1}}; \bm{\theta})$ is Gaussian and standard likelihood inference applies. 

While this method usually performs among the best in the additive case due to having strong order 1.5 convergence, it is the slowest due to the calculation of $e^{D\mathbf{F}(\mathbf{x}; \bm{\theta}^{(1)})h}$ in $\bm{\Omega}_{h}^\mathrm{[LL]}$. However, when the diffusion matrix is not constant, the LL approximation loses half an order of convergence due to poor approximation of $\bm{\Omega}_{h}^\mathrm{[LL]}(\mathbf{x}; \bm{\theta})$. Namely,
\begin{align*}
    \bm{\Omega}_{h}^\mathrm{[LL]}(\mathbf{x}; \bm{\theta}) &= h \bm{\Sigma} \bm{\Sigma}^\top(\mathbf{x}; \bm{\theta}^{(2)})+ \frac{h^2}{2} \left(D\mathbf{F}(\mathbf{x}; \bm{\theta}^{(1)}) \bm{\Sigma} \bm{\Sigma}^\top(\mathbf{x}; \bm{\theta}^{(2)}) + \bm{\Sigma} \bm{\Sigma}^\top(\mathbf{x}; \bm{\theta}^{(2)}) D^\top\mathbf{F}(\mathbf{x}; \bm{\theta}^{(1)}) \right) + \mathbf{R}(h^3, \mathbf{x}),
\end{align*}
which differs from the $\bm{\Omega}_{h}^\mathrm{[GA]}$ \eqref{eq:OmegaGa} by $\mathbb{L}\bm{\Sigma} \bm{\Sigma}^\top(\mathbf{x})$.

\section{Conclusion} \label{sec:Conclusion}

This paper extends the Pearson diffusion framework to multivariate models with quadratic diffusion coefficients and nonlinear drift. We formalized a class of multivariate Pearson diffusions and derived closed-form expressions for their first two moments. As concrete examples, we treated a multivariate Wright–Fisher diffusion and introduced the Student Kramers oscillator, showing that the latter is a hypoelliptic two-dimensional model with nonlinear drift and Student-type noise, and proving existence, uniqueness, and the existence of an invariant measure for its solution.

We proposed a Strang splitting estimator for nonlinear multivariate Pearson diffusions, based on splitting the dynamics into a multivariate Pearson diffusion and a nonlinear ODE and approximating the Pearson component by a Gaussian transition density with exact first and second moments. In comprehensive simulation studies for both the multivariate Wright–Fisher diffusion and the Student Kramers oscillator, the Strang estimator outperformed the EM, GA, and LL methods in terms of estimation error.

Finally, we illustrated the applicability of both the estimator and the Student Kramers oscillator by fitting them to high-resolution Greenland ice-core data, thereby demonstrating how the multivariate Pearson framework can be used for real climate proxy series while also revealing where more complex models may be needed. By broadening the scope of Pearson diffusions to multivariate and nonlinear-drift settings and providing an efficient estimation procedure, this work offers a flexible tool for modeling and inference in complex stochastic dynamical systems.

\section*{Code availability}
All simulation studies, real-data application results, and figures in this paper are fully reproducible. The supplementary code is publicly available on GitHub at \url{https://github.com/PredragPilipovic/nonlinear_multivariate_pearson_diffusions} \citep{pilipovic2026code}.

\section*{Acknowledgement}

This work has received funding from the European Union's Horizon 2020 research and innovation program under the Marie Skłodowska-Curie grant agreement No 956107, "Economic Policy in Complex Environments (EPOC)", and from Novo Nordisk Foundation NNF20OC0062958.

\bibliographystyle{abbrvnat}
\bibliography{references}

\clearpage

\setcounter{section}{0}
\setcounter{equation}{0}

\renewcommand{\thesection}{S\arabic{section}}
\renewcommand{\theequation}{S\arabic{equation}}

\section{Supplementary Material: One-dimensional Pearson diffusions} 
\label{sec:1DPearson}

The ergodic one-dimensional Pearson diffusions \eqref{eq:Pearson1D} can be classified into six cases based on the form of the squared diffusion coefficient $\sigma^2(x) = \alpha x^2 + \beta x + \gamma$. Each case presents specific conditions for the existence and uniqueness of ergodic solutions and corresponding invariant distributions. This classification is up to equivalence classes because the Pearson diffusions are closed under translations and scale transformations. Up to translation and scale transformations, the six classes are as follows \citep{FormanSorensen2008}.
\begin{itemize}
    \item For $\sigma^2(x) = \gamma$, it is an OU process defined on the entire real line. The unique ergodic solution exists for all $m\in \mathbb{R}$, with the invariant distribution being normal with mean $m$ and variance $\gamma/(2\lambda)$.
    \item For $\sigma^2(x) = \beta x$, it is a SQ process, defined on the positive half-line $(0, \infty)$. A unique ergodic solution exists for $2\lambda m \geq \beta$. The invariant distribution is the gamma distribution with scale parameter $\beta/(2\lambda)$ and shape parameter $2\lambda m/\beta$. For $0 < 2\lambda m < \beta$, the boundary at zero can be reached with positive probability, but the process remains stationary with the same gamma invariant distribution. 
    \item For $\sigma^2(x) = \alpha x^2$, it is a geometric Brownian motion type process, also called a GARCH diffusion. It is defined on the positive half-line $(0, \infty)$. A unique ergodic solution exists for all $\alpha > 0$ and $m> 0$. The invariant distribution is an inverse gamma distribution with shape parameter $1 - 2\lambda /\alpha$ and scale parameter $\alpha/(2\lambda m)$. The variance only exists for $\alpha < 2\lambda$.
    \item For $\sigma^2(x) = \alpha (x^2 + 1)$, the diffusion is defined on the entire real line. A unique ergodic solution exists for all $\alpha > 0$ and $m\in \mathbb{R}$. If $m= 0$, the invariant distribution is a scaled Student's $t$-distribution with $1 - 2 a/\alpha$ degrees of freedom. For $m\neq 0$, the invariant distribution is a skew $t$-distribution. 
    \item For $\sigma^2(x) = \alpha x (x + 1)$, the process is defined on the positive half-line $(0, \infty)$. A unique ergodic solution exists for all $\alpha > 0$ and $m\geq \alpha/(2\lambda)$. The invariant distribution is a scaled F-distribution with $-4 a m/\alpha$ and $2(1 - 2\lambda /\alpha)$ degrees of freedom. If $0 < m< \alpha/(2\lambda)$, the boundary at zero can be reached, but with an instantaneous reflecting boundary, the process remains stationary with the same F-distribution.
    \item For $\sigma^2(x) = \alpha x (x - 1)$, the process is a Jacobi diffusion, defined on the interval $(0, 1)$. A unique ergodic solution exists for $\alpha < 0$ and $m$ such that $\min(m, 1 - m) \geq \alpha/(2\lambda )$. The invariant distribution is a Beta distribution with shape parameters $2\lambda m/\alpha$ and $2\lambda (1 - m)/\alpha$. If $0 < m< \alpha/(2\lambda )$, the boundary at zero can be reached, and similarly for the boundary at one for $0 < 1 - m< \alpha/(2\lambda )$.
\end{itemize}
This classification illustrates the diverse applications of Pearson diffusions, allowing for explicit computation of invariant distributions and parameter conditions across different state spaces and diffusion forms. An important feature of Pearson diffusions is the ability to find explicit expressions for the marginal and conditional moments. The $k$-th absolute moment, $\mathbb{E}|X_t|^k$, is finite if and only if $\alpha/(2\lambda) < 1/(k - 1)$. This implies that all moments exist if $\alpha \leq 0$. However, for $\alpha > 0$, only moments satisfying $k < 2\lambda/\alpha + 1$ exist \citep{FormanSorensen2008}.

\section{Supplementary Material: Simulating Multivariate Wright-Fisher diffusion} \label{sec:SimulationWF}

To simulate Wright-Fisher diffusion, we follow \cite{He2020} and avoid Cholesky decomposition of $\bm{\Sigma}$ by using six independent Brownian motions $W_t^{(1)}, \ldots, W_t^{(6)}$.

The resulting three-dimensional SDE is
\begin{equation} \label{eq:WF_2_fin}
    \begin{aligned}
        \dif X_t^{(1)} &= \bigl(\kappa_1 + K_{11}X_t^{(1)} + K_{12}X_t^{(2)} + K_{13}X_t^{(3)} - X_t^{(1)}(\lambda_1 X_t^{(1)} + \lambda_2 X_t^{(2)} + \lambda_3 X_t^{(3)})\bigr)\dif t \\
        &\quad + \sqrt{X_t^{(1)}X_t^{(2)}}\dif W_t^{(1)} + \sqrt{X_t^{(1)}X_t^{(3)}}\dif W_t^{(2)} + \sqrt{X_t^{(1)}(1-X_t^{(1)}-X_t^{(2)}-X_t^{(3)})}\dif W_t^{(3)}, \\
        \dif X_t^{(2)} &= \bigl(\kappa_2 + K_{21}X_t^{(1)} + K_{22}X_t^{(2)} + K_{23}X_t^{(3)} - X_t^{(2)}(\lambda_1 X_t^{(1)} + \lambda_2 X_t^{(2)} + \lambda_3 X_t^{(3)})\bigr)\dif t \\
        &\quad - \sqrt{X_t^{(1)}X_t^{(2)}}\dif W_t^{(1)} + \sqrt{X_t^{(2)}X_t^{(3)}}\dif W_t^{(4)} + \sqrt{X_t^{(2)}(1-X_t^{(1)}-X_t^{(2)}-X_t^{(3)})}\dif W_t^{(5)}, \\
        \dif X_t^{(3)} &= \bigl(\kappa_3 + K_{31}X_t^{(1)} + K_{32}X_t^{(2)} + K_{33}X_t^{(3)} - X_t^{(3)}(\lambda_1 X_t^{(1)} + \lambda_2 X_t^{(2)} + \lambda_3 X_t^{(3)})\bigr)\dif t \\
        &\quad - \sqrt{X_t^{(1)}X_t^{(3)}}\dif W_t^{(2)} - \sqrt{X_t^{(2)}X_t^{(3)}}\dif W_t^{(4)} + \sqrt{X_t^{(3)}(1-X_t^{(1)}-X_t^{(2)}-X_t^{(3)})}\dif W_t^{(6)},
    \end{aligned}
\end{equation}
where $0 < X_t^{(i)} < 1$, $\kappa_i > 0$, for $i = 1,2,3$.

\begin{figure}
    \centering
    \includegraphics[width=\textwidth]{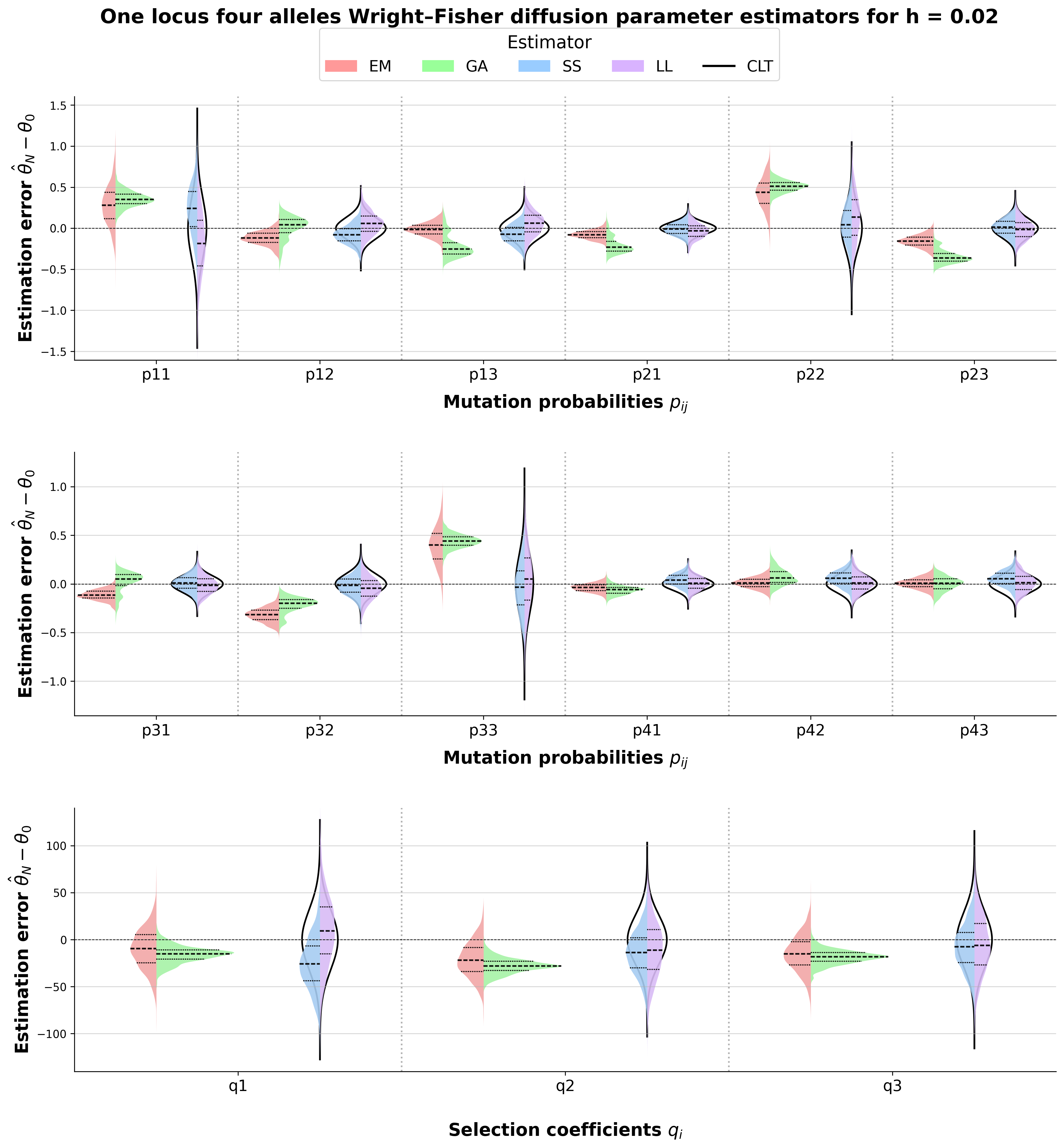}
    \caption{Violin plots of parameter estimation errors $\hat{\bm{\theta}}_N - \bm{\theta}_0$ based on 1000 simulated datasets ($T = 20$, $N = 1000$, $h = 0.02$). Colors indicate estimator. Inside each violin plot, black dashed lines show medians and the 25th and 75th percentiles. Black violin plots represent the asymptotic distributions.}
    \label{fig:WF_final_plot2}
\end{figure}

\begin{figure}
    \centering
    \includegraphics[width=.8\textwidth]{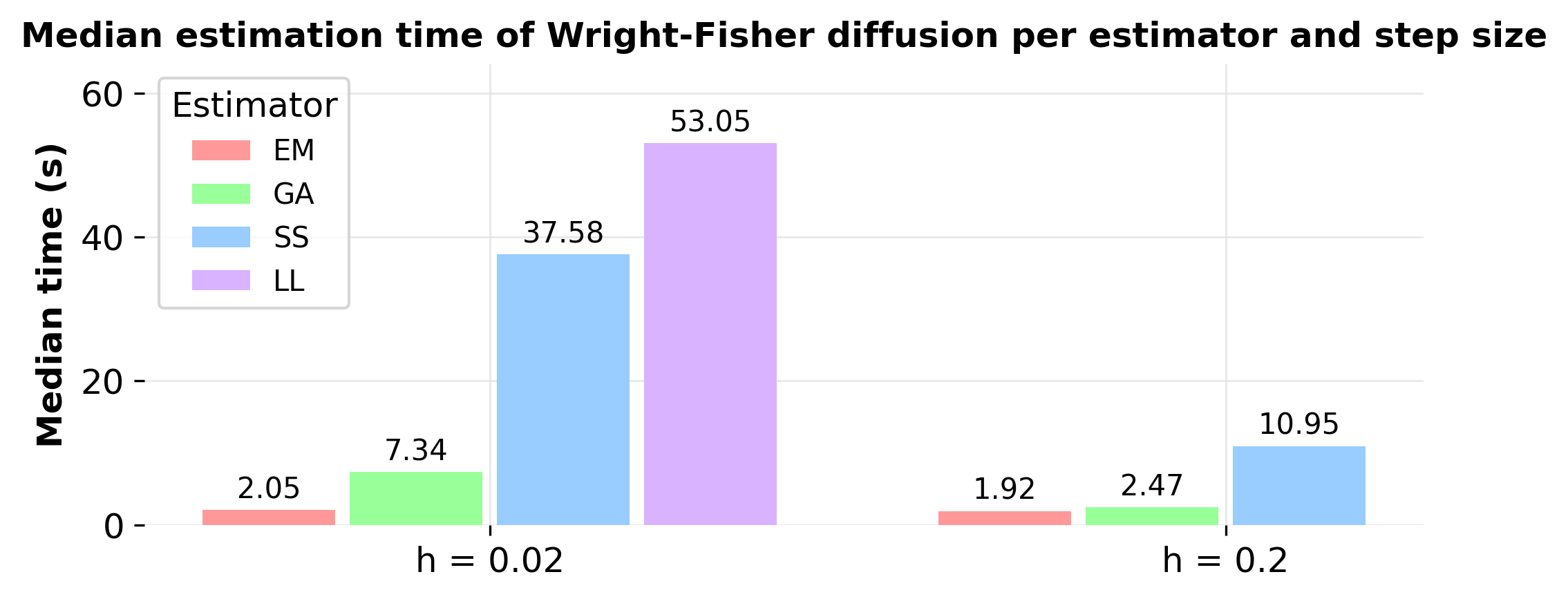}
    \caption{Median wall-clock estimation time (in seconds) across step sizes $h \in \{0.02, 0.2\}$, based on 1000 simulated datasets ($T = 20$, $N \in \{1000, 100\}$). Outliers and failed runs are excluded. Colors indicate estimator.}
    \label{fig:WF_times}
\end{figure}

 \section{Supplementary Material: Results from the Multivariate Wright-Fisher diffusion} \label{sec:ResultsWF}

For the finer time step $h=0.02$ (Figure~\ref{fig:WF_final_plot2}), all four estimators perform well and their differences become much less pronounced, as expected when the discretization error is small. In this regime, LL is less numerically unstable (20\% trajectories returned NaNs) and achieves accuracy comparable to SS, while EM and GA still tend to show larger biases for several transition probabilities and selection coefficients. This highlights that the main strength of the SS estimator is not at very small $h$, where essentially all methods work, but at coarser time steps where LL frequently fails and EM and GA underperform, whereas SS continues to be reasonably accurate. Note also that, although the GA estimator is in principle theoretically preferable to EM, GA approximation of its covariance matrix does not enforce positive definiteness, which makes this estimator numerically less stable in practice even when its nominal errors are smaller.

Figure~\ref{fig:WF_times} shows the median computation time per estimator across two step sizes ($h = 0.02$ and $h = 0.2$), computed after removing NA and outlier runs. EM and GA converge considerably faster across both settings, with LL being the most computationally expensive.

\section{Supplementary Material: Implementation} \label{sec:Implementation}

\subsection{Implementing covariance matrix for SS estimator for the Wright-Fisher diffusion}

In this Section we illustrate how to implement $\bm{\Omega}_h(\mathbf{x})$ for the Wright--Fisher diffusion \eqref{eq:WF_final_vec}. 

We start by expanding $\bm{\alpha}$ from \eqref{eq:alpha_l_ij} and $\bm{\beta}$ from Equation \eqref{eq:beta_l_ij} to the form Equation \eqref{eq:check_abc}
\begin{equation}
    \check{\bm{\alpha}} = \begin{bmatrix}
-1 & 0 & 0 & 0 & 0 & 0 & 0 & 0 & 0 \\
0 & -\frac{1}{2} & 0 & -\frac{1}{2} & 0 & 0 & 0 & 0 & 0 \\
0 & 0 & -\frac{1}{2} & 0 & 0 & 0 & -\frac{1}{2} & 0 & 0 \\
0 & -\frac{1}{2} & 0 & -\frac{1}{2} & 0 & 0 & 0 & 0 & 0 \\
0 & 0 & 0 & 0 & -1 & 0 & 0 & 0 & 0 \\
0 & 0 & 0 & 0 & 0 & -\frac{1}{2} & 0 & -\frac{1}{2} & 0 \\
0 & 0 & -\frac{1}{2} & 0 & 0 & 0 & -\frac{1}{2} & 0 & 0 \\
0 & 0 & 0 & 0 & 0 & -\frac{1}{2} & 0 & -\frac{1}{2} & 0 \\
0 & 0 & 0 & 0 & 0 & 0 & 0 & 0 & -1
\end{bmatrix}, \qquad\check{\bm{\beta}} = \begin{bmatrix}
1 & 0 & 0 \\
0 & 0 & 0 \\
0 & 0 & 0 \\
0 & 0 & 0 \\
0 & 1 & 0 \\
0 & 0 & 0 \\
0 & 0 & 0 \\
0 & 0 & 0 \\
0 & 0 & 1
\end{bmatrix}. \label{eq:alpha_beta}
\end{equation}
The choice of $\bm{\alpha}^{ij}$, or equivalently $\check{\bm{\alpha}}$, does not influence the theoretical results. However, some choices lead to faster and more numerically stable implementations. It is not computationally optimal to choose $\check{\bm{\alpha}}$ in \eqref{eq:check_abc}, as mentioned in Remark \ref{rmrk:sym}. We see that the choice of $\bm{\alpha}^{(l)ij}$ as in \eqref{eq:alpha_l_ij} leads to non-diagonal $\check{\bm{\alpha}}$ \eqref{eq:check_abc}. In the following we propose diagonal $\check{\bm{\alpha}}$ which leads to faster and more numerically stable algorithms.

First, we write the full squared diffusion matrix as a function of $\mathbf{X}$ as follows
\begin{equation}\label{eq:SigmaSigmaTWF}
    \bm{\Sigma}\bm{\Sigma}^\top(\mathbf{X}) = \diag(\mathbf{X}) - \sum_{l=1}^L \bm{\Pi}_l \mathbf{X}\mathbf{X}^\top \bm{\Pi}_l,
\end{equation}
where $\bm{\Pi}_l$ is the $M \times M$ block matrix with the $(l,l)$-th block equal to $\mathbf{I}_{M_l}$ and all other blocks equal to zero. Then, we vectorize $\bm{\Sigma}\bm{\Sigma}^\top_l(\bm{X}^{(l)})$ in \eqref{eq:SigmaSigmaT_final_WF},
\begin{equation}
\label{eq:vecSigmaFischerWright}
    \vect(\bm{\Sigma}\bm{\Sigma}^\top_l(\bm{X}^{(l)})) = \vect(\diag(\mathbf{X})) - \sum_{l=1}^L \vect(\bm{\Pi}_l \mathbf{X}\mathbf{X}^\top \bm{\Pi}_l).
\end{equation}
Instead of directly finding $\bm{\alpha}^{(l)ij}$, we will solve for $\check{\bm{\alpha}}$ as defined in \eqref{eq:check_abc}. To do so, we focus on the second term in \eqref{eq:vecSigmaFischerWright} and use vectorization of the outer product formula 
\begin{equation*}
    -\sum_{l=1}^L \vect(\bm{\Pi}_l \mathbf{X}\mathbf{X}^\top \bm{\Pi}_l) = -\left(\sum_{l=1}^L \bm{\Pi}_l \otimes\bm{\Pi}_l\right)\vect(\mathbf{X} \mathbf{X}^\top ).
\end{equation*}
Thus, $\check{\bm{\alpha}}= - \sum_{l=1}^L \bm{\Pi}_l \otimes\bm{\Pi}_l$, which is a diagonal matrix. Moreover, in our case $L = 1$, making $\check{\bm{\alpha}} = - \mathbf{I}_{d^2\times d^2}$. Now, we can use Theorem~\ref{thm:Covariance_matrix} to implement $\bm{\Omega}_h(\mathbf{x})$ for the Wright--Fisher diffusion \eqref{eq:WF_final_vec} with $\check{\bm{\alpha}} = - \mathbf{I}_{d^2\times d^2}$.

Alternately, we show how the choice of $\check{\bm{\alpha}} = - \mathbf{I}_{d^2\times d^2}$ corresponds to directly plugging the definition $\bm{\Sigma}\bm{\Sigma}^\top(\mathbf{x})$ for $L = 1$ from Equation \eqref{eq:SigmaSigmaT_vec} to Equation \eqref{eq:CovODE}. This derivation will also lead to simpler implementation formulas that will not involve matrix exponentials of $d^2 \times d^2$ matrices. 

First, we notice that 
\begin{equation} \label{eq:Ct_WF}
\begin{aligned}
    \frac{\dif \mathbf{C}(t)}{\dif t} &= \mathbf{A} \mathbf{C}(t) + \mathbf{C}(t) \mathbf{A}^{\top} + \mathbb{E}[\bm{\Sigma}\bm{\Sigma}^\top(\mathbf{X}_t)]\\
    &= \mathbf{A} \mathbf{C}(t) + \mathbf{C}(t) \mathbf{A}^{\top} + \diag(\mathbf{m}(t)) - \mathbf{C}(t) - \mathbf{m}(t)\mathbf{m}(t)^\top\\
    &= (\mathbf{A} - \frac{1}{2}\mathbf{I})  \mathbf{C}(t) + \mathbf{C}(t)(\mathbf{A} - \frac{1}{2}\mathbf{I})^\top + \diag(\mathbf{m}(t))  - \mathbf{m}(t)\mathbf{m}(t)^\top.
\end{aligned}
\end{equation}

The differential Sylvester equation \eqref{eq:Ct_WF} given initial condition $\mathbf{C}(0) =\mathbf{0}$ has a unique solution 
\begin{align*} 
    \mathbf{C}(t) &= -\int_0^t \exp((\mathbf{A}-\frac{1}{2}\mathbf{I})(t-s))\mathbf{m}(s)\mathbf{m}(s)^\top\exp((\mathbf{A}-\frac{1}{2}\mathbf{I})^\top(t-s))\dif s\\
    &+ \int_0^t \exp((\mathbf{A}-\frac{1}{2}\mathbf{I})(t-s))\diag(\mathbf{m}(s)) \exp((\mathbf{A}-\frac{1}{2}\mathbf{I})^\top(t-s)) \dif s.
\end{align*}
Using $\mathbf{m}(t)$ from \eqref{eq:MeanODESolution}, we get
\begin{equation}
    \mathbf{C}(t) = \mathbf{J}_1(t, \mathbf{A}) + \mathbf{J}_2(t, \mathbf{A}) + \mathbf{J}_3(t, \mathbf{A}) +\mathbf{J}_4(t, \mathbf{A}) +\mathbf{J}_5(t, \mathbf{A}),
\end{equation}
where we defined
\begin{align*}
    \mathbf{J}_1(t, \mathbf{A})&\coloneqq  -\int_0^t \exp((\mathbf{A}-\frac{1}{2}\mathbf{I})(t-s))\exp(\mathbf{A} s)(\mathbf{m}(0) - \mathbf{b})(\mathbf{m}(0) - \mathbf{b})^\top \exp(\mathbf{A}^\top s)\exp((\mathbf{A}-\frac{1}{2}\mathbf{I})^\top(t-s)) \dif s,\\
    \mathbf{J}_2(t, \mathbf{A})&\coloneqq - \int_0^t \exp((\mathbf{A}-\frac{1}{2}\mathbf{I})(t-s))\exp(\mathbf{A} s)(\mathbf{m}(0) - \mathbf{b})\mathbf{b}^\top \exp((\mathbf{A}-\frac{1}{2}\mathbf{I})^\top(t-s)) \dif s,\\
    \mathbf{J}_3(t, \mathbf{A})&\coloneqq -\int_0^t \exp((\mathbf{A}-\frac{1}{2}\mathbf{I})(t-s))\mathbf{b}(\mathbf{m}(0) - \mathbf{b})^\top \exp(\mathbf{A}^\top s)\exp((\mathbf{A}-\frac{1}{2}\mathbf{I})^\top(t-s)) \dif s,\\
    \mathbf{J}_4(t, \mathbf{A}) &\coloneqq \int_0^t \exp((\mathbf{A}-\frac{1}{2}\mathbf{I})(t-s))\diag(\exp(\mathbf{A} s)(\mathbf{m}(0) - \mathbf{b})) \exp((\mathbf{A}-\frac{1}{2}\mathbf{I})^\top(t-s)) \dif s,\\
    \mathbf{J}_5(t, \mathbf{A}) &\coloneqq \int_0^t \exp((\mathbf{A}-\frac{1}{2}\mathbf{I})(t-s))\bm{\Sigma}\mathbf{\Sigma}^\top(\mathbf{b})\exp((\mathbf{A}-\frac{1}{2}\mathbf{I})^\top(t-s))\dif s.
\end{align*}
Now, comparing the previous integrals with integrals \eqref{eq:cov_int_1}-\eqref{eq:cov_int_5} from Theorem~\ref{thm:Covariance_matrix}, it is clear that 
\begin{align*}
    \vect(\mathbf{J}_1(t, \mathbf{A})) &= \mathbf{I}_1(h, \mathbf{A}, -\mathbf{I}_{d^2\times d^2}) \vect((\mathbf{x} - \mathbf{b})(\mathbf{x} - \mathbf{b})^\top);\\
    \vect(\mathbf{J}_2(t, \mathbf{A})) &= \mathbf{I}_2(h, \mathbf{A}, -\mathbf{I}_{d^2\times d^2}) \vect((\mathbf{x} - \mathbf{b})\mathbf{b}^\top)\notag\\
    \vect(\mathbf{J}_3(t, \mathbf{A})) &=  \mathbf{I}_3(h, \mathbf{A}, -\mathbf{I}_{d^2\times d^2})  \vect(\mathbf{b}(\mathbf{x} - \mathbf{b})^\top);\\
    \vect(\mathbf{J}_4(t, \mathbf{A})) &= \mathbf{I}_4(h, \mathbf{A}, -\mathbf{I}_{d^2\times d^2}, \check{\bm{\beta}})  (\mathbf{x} - \mathbf{b});\\
    \vect(\mathbf{J}_5(t, \mathbf{A})) &= \mathbf{I}_5(h, \mathbf{A}, -\mathbf{I}_{d^2\times d^2}) \vect(\bm{\Sigma}\bm{\Sigma}^\top(\mathbf{b})),
\end{align*}
where $\check{\bm{\beta}}$ is defined in \eqref{eq:alpha_beta}. The four integrals $\mathbf{J}_1, \mathbf{J}_2, \mathbf{J}_3$ and $\mathbf{J}_5$ can be calculated explicitly without vectorization as follows
\begin{equation} \label{eq:Js}
\begin{aligned}
    \mathbf{J}_1(t, \mathbf{A})& =  \frac{1 - e^t}{e^t}\exp(\mathbf{A} t)(\mathbf{m}(0) - \mathbf{b})(\mathbf{m}(0) - \mathbf{b})^\top \exp(\mathbf{A}^\top t),\\
    \mathbf{J}_2(t, \mathbf{A})&=  - \exp(\mathbf{A} t)(\mathbf{m}(0) - \mathbf{b}) \mathbf{b}^\top \int_0^t \exp((\mathbf{A}-\mathbf{I})^\top s) \dif s,\\
    \mathbf{J}_3(t, \mathbf{A})& = -\int_0^t \exp((\mathbf{A}-\mathbf{I})s)\dif s\ \mathbf{b}(\mathbf{m}(0) - \mathbf{b})^\top \exp(\mathbf{A}^\top t) = \mathbf{J}_2(t, \mathbf{A})^\top,\\
    \mathbf{J}_5(t, \mathbf{A}) &= \int_0^t \exp((\mathbf{A}-\frac{1}{2}\mathbf{I})(t-s))\bm{\Sigma}\mathbf{\Sigma}^\top(\mathbf{b})\exp(-(\mathbf{A}-\frac{1}{2}\mathbf{I})^\top s)\dif s\exp((\mathbf{A}-\frac{1}{2}\mathbf{I})^\top t).
\end{aligned}
\end{equation}
The integral $\mathbf{J}_4$ can only be calculated as $\vect(\mathbf{J}_4(t, \mathbf{A}))$. All integrals in \eqref{eq:Js} can be computed using Theorem 1 from \citep{VanLoan1978}. Thus, for $L = 1$ and $\check{\bm{\alpha}} = - \mathbf{I}_{d^2\times d^2}$, we do not need to vectorize four out of five integrals, which makes computations much cheaper, faster, and more stable. Conversely, if we worked with $\check{\bm{\alpha}} $ as defined in \eqref{eq:alpha_beta}, we would need to vectorize all five integrals.

To reduce computational cost, all quantities in $\bm{\Omega}_h$ that depend only on $\mathbf{A}$ and $h$ (the matrix exponentials, the integral of matrix exponentials, the vectorized flow matrix $\mathbf{J}_4$, and the constant diffusion term $\bm{\Sigma}\mathbf{\Sigma}^\top(\mathbf{b})$) are computed once prior to the scan over observations and cached. Inside the scan, only the state-dependent terms (quadratic and cross terms in $\mathbf{x} - \mathbf{b}$) are evaluated per step, reducing each likelihood contribution to a handful of matrix-vector products.

\subsection{Gaussian approximation estimator for the Student Kramers oscillator}

From \eqref{eq:GA}, the transition distribution of the GA approximation is
\begin{align*}
    \begin{bmatrix}
        X_{t_k}\\
        V_{t_k} \\
    \end{bmatrix} \mid
    \begin{bmatrix}
        X_{t_{k-1}}\\
        V_{t_{k-1}}\\
    \end{bmatrix} = \begin{bmatrix}
        x\\
        v\\
    \end{bmatrix} &\sim \mathcal{N}\left(
    \bm{\mu}_h^{[\mathrm{GA}]}(x, v), 
    \bm{\Omega}_h^{\mathrm{[GA]}}(x,v)\right),
\end{align*}
where
\begin{equation}
   \bm{\mu}_h^{[\mathrm{GA}]}(x, v) =  \begin{bmatrix}
        x + h v - \frac{h^2}{2}(\eta v + U'(x))\\
        v - h (\eta v + U'(x)) + \frac{h^2}{2}(\eta^2 v + \eta U'(x) - U''(x) v)
    \end{bmatrix}.
\end{equation}
Here, we cannot use $\bm{\Omega}_h^{[\mathrm{GA}]}$ from \eqref{eq:GA} directly because it is singular. Instead, we use formula \eqref{eq:expendingL} to add one more order of $h$ to $\bm{\Omega}_h^{[\mathrm{GA}]}$ and obtain
\begin{align}
    [\bm{\Omega}_h^{[\mathrm{GA}]}(x,v)]_{11} &= \frac{h^3}{3} \left( v \left( v \alpha + \beta \right) + \gamma \right), \label{eq:OmegahK11} \\
    [\bm{\Omega}_h^{[\mathrm{GA}]}(x,v)]_{12} &= \frac{h^2}{2} \left( v \left( v \alpha + \beta \right) + \gamma \right) + \frac{h^3}{6} \left( 2bv x^2 \alpha + 2av x^3 \alpha + v^2 \alpha^2 + bx^2 \beta + ax^3 \beta + v \alpha \beta + d \left( 2v \alpha + \beta \right) \right) \notag \\
    & +\frac{h^3}{6} \left( cx \left( 2v \alpha + \beta \right) + \alpha \gamma - \left( 5v^2 \alpha + 4v \beta + 3 \gamma \right) \eta \right), \label{eq:OmegahK12} \\
     [\bm{\Omega}_h^{[\mathrm{GA}]}(x,v)]_{22} &= h \left( v \left( v \alpha + \beta \right) + \gamma \right) + \frac{h^2}{2} \left( 2bv x^2 \alpha + 2av x^3 \alpha + v^2 \alpha^2 + bx^2 \beta + ax^3 \beta + v \alpha \beta + d \left( 2v \alpha + \beta \right) \right) \notag \\
    & + \frac{h^2}{2} \left( cx \left( 2v \alpha + \beta \right) + \alpha \gamma - \left( 4v^2 \alpha + 3v \beta + 2 \gamma \right) \eta \right) \notag \\
    & + \frac{h^3}{6} \left( 2d^2 \alpha + 8bv^2 x \alpha + 12av^2 x^2 \alpha \right) \notag \\
    & + \frac{h^3}{6} \left( 2b^2 x^4 \alpha + 4ab x^5 \alpha + 2a^2 x^6 \alpha + 2bv x^2 \alpha^2 + 2av x^3 \alpha^2 + v^2 \alpha^3 + 6bv x \beta + 9av x^2 \beta + bx^2 \alpha \beta \right) \notag \\
    & + \frac{h^3}{6} \left( ax^3 \alpha \beta + v \alpha^2 \beta + d \alpha \left( 4x (c + x (b + a x)) + 2v \alpha + \beta \right) + 4bx \gamma + 6ax^2 \gamma + \alpha^2 \gamma + 2c^2 x^2 \alpha \right) \notag \\
    & + \frac{h^3}{6} \left( -d \left( 10v \alpha + 3 \beta \right) \eta - \left( bx^2 \left( 10v \alpha + 3 \beta \right) + ax^3 \left( 10v \alpha + 3 \beta \right) + \alpha \left( 6v^2 \alpha + 5v \beta + 4 \gamma \right) \right) \eta \right) \notag \\
    & + \frac{h^3}{6} c \left( 4v^2 \alpha + 4x^3 (b + a x) \alpha + 3v \beta + x \alpha \beta + 2 \gamma + 2v x \alpha ( \alpha - 5 \eta ) - 3x \beta \eta \right) \notag \\
    & + \frac{h^3}{6} \left( 12v^2 \alpha + 7v \beta + 4 \gamma \right) \eta^2. \label{eq:OmegahK22}
\end{align}
Although we could trim $\bm{\Omega}_h^{[\mathrm{GA}]}$ to include only the lowest order terms of $h$, leading to an estimator as in \cite{DitlevsenSamson2018}, we use the full approximation given by formulas \eqref{eq:OmegahK11}-\eqref{eq:OmegahK22}. The GA estimator is defined as
\begin{align*}
\widehat{\bm{\theta}}_N^\mathrm{[GA]} = \argmin_{\bm{\theta}} \sum_{k=1}^N \left( \log \det \bm{\Omega}_h^{[\mathrm{GA}]}(\mathbf{Y}_{t_{k-1}}) + (\mathbf{Y}_{t_k} - \bm{\mu}_h^{[\mathrm{GA}]}(\mathbf{Y}_{t_{k-1}}))^\top \bm{\Omega}_h^{[\mathrm{GA}]}(\mathbf{Y}_{t_{k-1}})^{-1} (\mathbf{Y}_{t_k} - \bm{\mu}_h^{[\mathrm{GA}]}(\mathbf{Y}_{t_{k-1}})) \right),
\end{align*}
where $Y_{t_k} = (X_{t_k}, V_{t_k})$ for $k = 0, 1, \ldots, N$.

\subsection{Ozaki's local linearization estimator for the Student Kramers Oscillator}

To derive the LL estimator, we first transform  \eqref{eq:StudentKramersSDE} to an SDE with constant diffusion coefficient. We achieve this by applying the Lamperti transform $\psi$, similar to the approach in \cite{nagahara1996non} for a one-dimensional nonlinear student-type Pearson diffusion. We have
\begin{align*}
    U_t = \psi(V_t) = \int^{V_t} \frac{d v}{\sqrt{\alpha v^2 + \beta v + \gamma}} = \frac{1}{\sqrt{\alpha}}\arcsinh \left(\frac{2 \alpha V_t + \beta}{\sqrt{ 4 \alpha \gamma - \beta^2}}\right).
\end{align*}
Since we assume that $\alpha > 0$ and $4\alpha \gamma - \beta^2 \leq 0$, then
\begin{align*}
    V_t = \psi^{-1}(U_t) = \frac{\sqrt{4 \alpha \gamma - \beta^2}\sinh(\sqrt{\alpha} U_t) - \beta}{2 \alpha}.
\end{align*}
We transform \eqref{eq:StudentKramersSDE} by applying It\^o's lemma to $\widetilde{\bm{\psi}}(X_t, V_t) = (X_t, \psi(V_t))$
\begin{align*}
    \dif X_t &= F^{(1)}(\widetilde{\bm{\psi}}^{-1}(X_t, U_t)) \dif t,\\
    \dif U_t &= \left(\frac{\partial \psi}{\partial v}(\widetilde{\bm{\psi}}^{-1}(X_t, U_t))F^{(2)}(\widetilde{\bm{\psi}}^{-1}(X_t, U_t))+\frac{1}{2}\frac{\partial^2 \psi}{\partial v^2}(\widetilde{\bm{\psi}}^{-1}(X_t, U_t)\sigma^2(\psi^{-1}(U_t)))\right)\dif t + \dif W_t.
\end{align*}
The transformed SDE becomes
\begin{align*}
    \dif X_t &= \left(\sqrt{ \gamma - \frac{\beta^2}{4 \alpha}}\frac{\sinh(\sqrt{\alpha} U_t)}{\sqrt{\alpha}} - \frac{\beta}{2 \alpha}\right) \dif t,\\
    \dif U_t &= \left(-\left(\eta + \frac{\alpha}{2}\right) \frac{\tanh(\sqrt{\alpha} U_t)}{\sqrt{\alpha}} + \frac{\frac{\beta}{2 \alpha} \eta - U'(X_t)}{\sqrt{ \gamma - \frac{\beta^2}{4\alpha}}\cosh(\sqrt{\alpha}U_t)}\right)\dif t + \dif W_t.
\end{align*}
To implement the LL estimator, we need to find $D\mathbf{F}(x,u)$ for the corresponding drift function $\mathbf{F}$, which is
\begin{align*}
    D\mathbf{F}(x, u) = \begin{bmatrix}
        0 & \sqrt{ \gamma - \frac{\beta^2}{4\alpha}}\cosh(\sqrt{\alpha}U_t)\\
        -\frac{U''(x)}{\sqrt{ \gamma - \frac{\beta^2}{4\alpha}}\cosh(\sqrt{\alpha}U_t)} & -\frac{\eta + \frac{\alpha}{2}}{\cosh^2(\sqrt{\alpha}U_t)} - \frac{\frac{\beta}{2 \alpha} \eta - U'(X_t)}{\sqrt{ \gamma - \frac{\beta^2}{4\alpha}}\cosh(\sqrt{\alpha}U_t)}\sqrt{\alpha}\tanh(\sqrt{\alpha}U_t)
    \end{bmatrix}.
\end{align*}
Then, the LL estimator for SDE \eqref{eq:StudentKramersSDE} is given by
\begin{align}
\widehat{\bm{\theta}}_N^\mathrm{[LL]} = \argmin_{\bm{\theta}} &\left\{\sum_{k=1}^N (\widetilde{\bm{\psi}}(\mathbf{Y}_{t_k}) - \bm{\mu}_h^{[\mathrm{LL}]}(\widetilde{\bm{\psi}}(\mathbf{Y}_{t_{k-1}})))^\top \bm{\Omega}_h^{[\mathrm{LL}]}(\widetilde{\bm{\psi}}(\mathbf{Y}_{t_{k-1}}))^{-1} (\widetilde{\bm{\psi}}(\mathbf{Y}_{t_k}) - \bm{\mu}_h^{[\mathrm{LL}]}(\widetilde{\bm{\psi}}(\mathbf{Y}_{t_{k-1}}))) \right. \notag \\
&\left. + \sum_{k=1}^N \log \det \bm{\Omega}_h^{[\mathrm{LL}]}(\widetilde{\bm{\psi}}(\mathbf{Y}_{t_{k-1}})) + \sum_{k=1}^N \log (\psi'(V_{t_k})^2) \right\}, \label{eq:LLstudent}
\end{align}
where $\bm{\mu}_h^{[\mathrm{LL}]}$ and $\bm{\Omega}_h^{[\mathrm{LL}]}$ are defined in \eqref{eq:LL}. The last term in \eqref{eq:LLstudent} arises due to the change of variable formula.

\section{Supplementary Material: Proofs}  \label{sec:Proofs}

Denote the infinitesimal generator of SDE \eqref{eq:SDE} by $\mathbb{L}$ and the infinitesimal generator of SDE \eqref{eq:SplittingEq1} by $\mathbb{L}^{[1]}$. For sufficiently smooth function $\bm{\phi}$, it holds
\begin{align}
    \mathbb{L}\bm{\phi}(\mathbf{x}) &\coloneqq \sum_{i=1}^d \partial_{x^{(i)}} \bm{\phi}(\mathbf{x}) F^{(i)}(\mathbf{x})+ \frac{1}{2}\sum_{i,j = 1}^d \partial^2_{x^{(i)} x^{(j)}} \bm{\phi}(\mathbf{x}) [\bm{\Sigma}\bm{\Sigma}^\top(\mathbf{x})]_{ij}, \label{eq:L}\\
    \mathbb{L}^{[1]}\bm{\phi}(\mathbf{x}) &\coloneqq \sum_{i=1}^d \partial_{x^{(i)}}\bm{\phi}(\mathbf{x}) \mathbf{A}_{i\cdot}(\mathbf{x} - \mathbf{b})+ \frac{1}{2}\sum_{i,j = 1}^d \partial^2_{x^{(i)} x^{(j)}} \bm{\phi}(\mathbf{x}) [\bm{\Sigma}\bm{\Sigma}^\top(\mathbf{x})]_{ij}. \label{eq:L1}
\end{align} 
To emphasize parameters, we sometimes write
\begin{equation*}
    \mathbb{L}_{\bm{\theta}}\bm{\phi}(\mathbf{x}; \bm{\theta}) = \mathbb{L}_{\bm{\theta}^{(1)},\bm{\theta}^{(2)}}\bm{\phi}(\mathbf{x}; \bm{\theta})
\end{equation*}
and under the true parameters
\begin{equation*}
    \mathbb{L}_{\bm{\theta}_0}\bm{\phi}(\mathbf{x}; \bm{\theta}) = \mathbb{L}_{\bm{\theta}^{(1)}_0,\bm{\theta}^{(2)}_0}\bm{\phi}(\mathbf{x}; \bm{\theta}) = \mathbb{L}_0\bm{\phi}(\mathbf{x}; \bm{\theta}).
\end{equation*}

To prove the asymptotic behavior of the estimator, we approximate $\bm{\Omega}_h$. However, we use the exact expressions in the implementation, with the approximations serving only as a tool for the proofs.
\begin{proposition}
    \label{prop:OmegahProp}
    For the covariance matrix $\bm{\Omega}_h^\mathrm{[SS]}(\mathbf{x})$ defined in \eqref{eq:OmegahSS}, it holds
    \begin{align}
        \bm{\Omega}_h^\mathrm{[SS]}(\mathbf{x}) &= h\bm{\Sigma}\bm{\Sigma}^\top(\mathbf{x}) + \frac{h^2}{2}\left(\mathbf{A} \bm{\Sigma}\bm{\Sigma}^\top(\mathbf{x}) + \bm{\Sigma}\bm{\Sigma}^\top(\mathbf{x}) \mathbf{A}^\top  + \mathbb{L}\bm{\Sigma}\bm{\Sigma}^\top(\mathbf{x})\right)+ \mathbf{R}(h^3, \mathbf{x}); \label{eq:OmegahSS_approx}\\
        \bm{\Omega}_h^{\mathrm{[SS]}}(\mathbf{x})^{-1} &=\frac{1}{h}\bm{\Sigma}\bm{\Sigma}^\top(\mathbf{x})^{-1} - \frac{1}{2}(\bm{\Sigma}\bm{\Sigma}^\top(\mathbf{x})^{-1}\mathbf{A} + \mathbf{A}^\top \bm{\Sigma}\bm{\Sigma}^\top(\mathbf{x})^{-1})\notag\\
        &-\frac{1}{2}\bm{\Sigma}\bm{\Sigma}^\top(\mathbf{x})^{-1}\mathbb{L}\bm{\Sigma}\bm{\Sigma}^\top(\mathbf{x})\bm{\Sigma}\bm{\Sigma}^\top(\mathbf{x})^{-1}  + \mathbf{R}(h, \mathbf{x});\\
        \log \det \bm{\Omega}_h^{\mathrm{[SS]}}(\mathbf{x}) &=d \log h +  \log \det \bm{\Sigma}\bm{\Sigma}^\top(\mathbf{x}) + h \tr \mathbf{A} + \frac{h}{2} \tr \bm{\Sigma}\bm{\Sigma}^\top(\mathbf{x})^{-1}\mathbb{L}\bm{\Sigma}\bm{\Sigma}^\top(\mathbf{x}) + {R}(h^2, \mathbf{x}).
    \end{align}
\end{proposition}

\subsection{Proof of Theorem \ref{thm:StudentKramers}}
To prove the existence and uniqueness of the solution to \eqref{eq:StudentKramersSDE}, we show that the diffusion function
\begin{equation*}
    \bm{\Sigma}(x,v) = \begin{bmatrix}
        0 & 0 \\
        0 & \sqrt{\alpha v^2 + \beta v + \gamma}
    \end{bmatrix}
\end{equation*}
is Lipschitz and that the drift function
\begin{equation*}
    \mathbf{F}(x, v) = \begin{bmatrix}
        v\\
        -\eta v + a x^3 + b x^2 + c x + d
    \end{bmatrix}
\end{equation*}
is one-sided Lipschitz and of at most polynomial growth (see Assumptions \ref{as:monoton_condition} and \ref{as:F_polynomial_growth} in Section \ref{sec:Assumptions}).

To see that $\bm{\Sigma}(x,v)$ is Lipschitz, we show that the derivative of $\sigma(v) = \sqrt{\alpha v^2 + \beta v + \gamma}$ is bounded for all $v \in \mathbb{R}$. We have
\begin{align*}
    |\sigma'(v)| = \frac{|2 \alpha v + \beta|}{2 \sqrt{\alpha v^2 + \beta v + \gamma}}.
\end{align*}
If $v = 0$, then $|\sigma'(v)| = \frac{|\beta|}{2\sqrt{\gamma}}$. Otherwise, for $|v| > c_1 > 0$ and $\sqrt{1 + \beta/(\alpha v) + \gamma / (\alpha v^2)} > c_2 > 0$, we have
\begin{align*}
    |\sigma'(v)| &\leq \frac{|\alpha v|}{\sqrt{\alpha v^2 + \beta v + \gamma}} + \frac{|\beta|}{2 \sqrt{\alpha v^2 + \beta v + \gamma}} < \frac{\sqrt{\alpha}}{c_2} + \frac{|\beta|}{2 \sqrt{\alpha} c_1 c_2}.
\end{align*}

To show that $\mathbf{F}$ is one-sided Lipschitz, we start by computing the Jacobian matrix of $\mathbf{F}$ as
\begin{equation}
    D \mathbf{F}(x, v) = \begin{bmatrix}
        0 & 1\\
        3 a x^2 + 2 b x + c & -\eta
    \end{bmatrix}.
\end{equation}
Since $a < 0$ and $\eta > 0$, all components of $D \mathbf{F}(x, v)$ are upper bounded, so $\mathbf{F}$ is one-sided Lipschitz. By construction, $\mathbf{F}$ has polynomial growth. Thus, a unique, strong solution exists to \eqref{eq:StudentKramersSDE}.

Next, we prove that the solution to \eqref{eq:StudentKramersSDE} is hypoelliptic using H\"ormander's condition \citep{Hormander1967}. To apply this condition, we first write  \eqref{eq:StudentKramersSDE} in Stratonovich form:
\begin{equation} \label{eq:StudentKramersSDEStratonovich}
    \begin{aligned}
        \dif X_t &= V_t \dif t, \\
        \dif V_t &= \left(-\eta V_t + a X_t^3 + b X_t^2 + c X_t + d - \frac{1}{2}\left(\alpha V_t + \frac{1}{2}\beta\right) \right) \dif t + \sqrt{\alpha V_t^2 + \beta V_t + \gamma} \circ \dif W_t.
    \end{aligned}
\end{equation}
Now, the associated drift and diffusion vector fields are
\begin{equation}
    \mathbf{V}_0(x, v) = \begin{bmatrix}
        v\\
        \left(-\eta - \frac{1}{2}\alpha\right) v + a x^3 + b x^2 + c x + d - \frac{1}{4}\beta
    \end{bmatrix}, \quad
    \mathbf{V}_1(x, v) = \begin{bmatrix}
        0\\
        \sqrt{\alpha v^2 + \beta v + \gamma}
    \end{bmatrix}.
\end{equation}
We recall that the Lie bracket of smooth vector fields $\bm{f}, \bm{g}: \mathbb{R}^{d} \to \mathbb{R}^{d}$, defined as
\begin{equation*}
    [\bm{f}, \bm{g}] \coloneqq D \bm{g}(\mathbf{x}) \bm{f}(\mathbf{x}) - D \bm{f}(\mathbf{x}) \bm{g}(\mathbf{x}),
\end{equation*}
is used to verify H\"ormander's condition.

We define the set $\mathcal{H}$ of vector fields iteratively as
\begin{align*}
    \mathbf{V}_1 \in \mathcal{H},\\
    \mathbf{H} \in \mathcal{H} \Rightarrow [\mathbf{V}_0, \mathbf{H}], [\mathbf{V}_1, \mathbf{H}] \in \mathcal{H}.
\end{align*}
The weak H\"ormander condition is met if the vectors in $\mathcal{H}$ span $\mathbb{R}^{d}$ at every point $\mathbf{x} \in \mathbb{R}^{d}$. Initially, vector $\mathbf{V}_1$ spans $\{(0, y) \in \mathbb{R}^{2} \mid y \in \mathbb{R}\}$, a one-dimensional subspace. Therefore, we need to verify the existence of some $\mathbf{H} \in \mathcal{H}$ with a non-zero first element. It is easy to see that
\begin{align*}
    [\mathbf{V}_0, \mathbf{V}_1]^{(1)} &= -\sigma(v) < 0, \quad \forall v \in \mathbb{R},\\
    [\mathbf{V}_1, \mathbf{V}_0]^{(1)} &= \sigma(v) \geq 0, \quad \forall v \in \mathbb{R}.
\end{align*}
Thus, the Student Kramers SDE defined in \eqref{eq:StudentKramersSDE} is hypoelliptic, meaning it admits a smooth transition density.

Since SDE \eqref{eq:StudentKramersSDE} is hypoelliptic and the drift and diffusion functions with their corresponding derivatives grow at most polynomially, SDE \eqref{eq:StudentKramersSDE} has the strong Feller property (see, for example, Theorem 1.2 in \cite{Hairer2011}). Then, we use Theorem 4.5 from \cite{Meyn1993} to prove that SDE \eqref{eq:StudentKramersSDE} has an invariant measure. We choose $V(x,v) = \frac{1}{2} v^2 + U(x)$ as a Lyapunov function. Since $a < 0$, $\lim\limits_{\|\mathbf{x}\| \to \infty} V(\mathbf{x}) = +\infty$.

For SDE \eqref{eq:StudentKramersSDE}, the infinitesimal generator is
\begin{equation*}
   \mathbb{L} \phi(x, v) = v \frac{\partial \phi}{\partial x} + (-\eta v - U'(x)) \frac{\partial \phi}{\partial v} + \frac{1}{2}(\alpha v^2 + \beta v + \gamma) \frac{\partial^2 \phi}{\partial v^2}.
\end{equation*}
Then,
\begin{equation}
   \mathbb{L} V(x, v) = \left(\frac{\alpha}{2} - \eta\right) v^2 + \frac{\beta}{2} v + \frac{\gamma}{2}.
\end{equation}
Since $\alpha < 2 \eta$, we can find a compact set $K \subset \mathbb{R}^2$ and constants $c_1 > 0, c_2 \in \mathbb{R}$, such that
\begin{equation}
   \mathbb{L} V(\mathbf{x}) \leq -c_1 \|\mathbf{x}\|^2 + c_2 \mathds{1}\{\mathbf{x} \in K\}.
\end{equation}
According to Theorem 4.5 in \cite{Meyn1993}, an invariant measure $\pi$ for SDE \eqref{eq:StudentKramersSDE} exists.

We can justify decoupling of the invariant density $ \pi(x,v)\approx \pi_X(x)\pi_V(v)$ by analyzing the system across two limits: the marginalization of the momentum variable, and the overdamped Smoluchowski-Kramers limit of the spatial variable.

Assuming strong damping $\eta \gg 1$, the velocity relaxation occurs on a much faster timescale than the macroscopic evolution of $X_t$. To establish the macroscopic behavior of $X_t$, we apply the Smoluchowski-Kramers approximation—also known as the overdamped limit. Mathematically, this is formalized via the averaging principle for fast-slow stochastic systems \citep{Pavliotis2008, Pavliotis2014}, where $V_t$ acts as a rapidly mixing fast variable and $X_t$ as the slow variable.

In the strong friction regime, the instantaneous dynamics of $V_t$ are dominated by the large $O(\eta)$ friction term and the noise, making the bounded $O(1)$ gradient $-U'(X_t)$ negligible from the perspective of the fast timescale. Consequently, for a frozen macroscopic position $X_t \approx x$, the velocity evolves as a one-dimensional Pearson diffusion:
\begin{equation}
\mathrm{d} V_t \approx -\eta V_t \mathrm{d}t + \sqrt{\alpha V_t^2 + \beta V_t + \gamma} \mathrm{d} W_t.
\end{equation}
The exact stationary density $\pi_V(v)$ for this fast process satisfies the stationary Fokker-Planck equation
\begin{equation}
\frac{\partial}{\partial v} (\eta v \pi_V) + \frac{1}{2} \frac{\partial^2}{\partial v^2} \Big[ (\alpha v^2 + \beta v + \gamma) \pi_V \Big] = 0.
\end{equation}
Integrating once with vanishing boundary conditions at infinity yields
\begin{equation}
\frac{\mathrm{d}}{\mathrm{d}v} \ln \pi_V(v) = \frac{-2\eta v - \frac{\mathrm{d}}{\mathrm{d}v}\sigma^2(v)}{\sigma^2(v)} = \frac{-(2\eta + 2\alpha)v - \beta}{\alpha v^2 + \beta v + \gamma}.
\end{equation}
Integrating this provides the exact skew $t$-distribution invariant density
\begin{equation} \label{eq:pi_V}
\pi_V(v) \propto (\alpha v^2 + \beta v + \gamma)^{-\left(\frac{\eta}{\alpha} + 1\right)} \exp\left( \frac{2\beta \eta}{\alpha \sqrt{4\alpha\gamma - \beta^2}} \arctan\left( \frac{2\alpha v + \beta}{\sqrt{4\alpha\gamma - \beta^2}} \right) \right).
\end{equation}

To establish the homogenized marginal equation for $X_t$, we algebraically rearrange the $V_t$ equation in \eqref{eq:StudentKramersSDE} to isolate the displacement $V_t \mathrm{d}t$
\begin{equation}
\mathrm{d}X_t = V_t \mathrm{d}t = -\frac{1}{\eta} U'(X_t) \mathrm{d}t + \frac{1}{\eta}\sigma(V_t)\mathrm{d}W_t - \frac{1}{\eta} \mathrm{d}V_t.
\end{equation}
In the high-friction limit $\eta \to \infty$, we analyze the macroscopic dynamics over timescales much larger than the velocity relaxation time. The exact differential term $\frac{1}{\eta}\mathrm{d}V_t$ integrates to $(V_t - V_0)/\eta$, which vanishes in probability because the stationary velocity process $V_t$ is bounded. 

The remaining stochastic integral $\frac{1}{\eta} \int \sigma(V_t) \mathrm{d}W_t$ is a continuous martingale. Because the isolated fast process $V_t$ is ergodic, its time-averaged quadratic variation converges almost surely to its stationary expectation. By the functional central limit theorem for martingales \citep{EthierKurtz1986}, this highly fluctuating noise term converges in distribution to a standard Itô integral driven by the constant effective variance
\begin{equation}
\sigma_{\text{eq}}^2 = \frac{1}{\eta^2}\mathbb{E}_{\pi_V} [\sigma^2(V)] = \frac{1}{\eta^2}\mathbb{E}_{\pi_V} [\alpha V^2 + \beta V + \gamma] = \frac{2}{\eta}\mathbb{E}_{\pi_V}[V^2] = \frac{2\gamma}{\eta(2 \eta - \alpha)},
\end{equation}
where the latter equalities follow from the stationarity of the fast process. Specifically, enforcing the condition that the infinitesimal generator expectations vanish, $\mathbb{E}_{\pi_V}[\mathbb{L}f(V)] = 0$ for any $f$, immediately yields $\mathbb{E}_{\pi_V}[V] = 0$ and $\mathbb{E}_{\pi_V}[\sigma^2(V)] = 2\eta\mathbb{E}_{\pi_V}[V^2]$. 

The homogenization process generates no spurious noise-induced drift because the diffusion coefficient $\sigma(V_t)$ depends strictly on the fast variable and is independent of the spatial coordinate $X_t$. Consequently, the effective macroscopic SDE for the spatial coordinate is the overdamped Itô Langevin equation
\begin{equation}
\mathrm{d} X_t \approx -\frac{1}{\eta} U'(X_t) \mathrm{d} t + \sqrt{\sigma_{\text{eq}}^2} \mathrm{d} \tilde{W}_t,
\end{equation}
where $\tilde{W}_t$ is a standard macroscopic Brownian motion. The stationary solution to this purely gradient-driven system is the Boltzmann-Gibbs distribution
\begin{equation} \label{eq:pi_X}
\pi_X(x) \propto \exp\left( -\frac{2 U(x)/\eta}{\sigma_{\text{eq}}^2} \right).
\end{equation}

\subsection{Proof of Proposition \ref{prop:OmegahProp}}

To approximate $\bm{\Omega}_h$, we approximate the term in the integrated in \eqref{eq:Omegah_int}. We first recall the well-known formula 
\begin{equation} \label{eq:expendingL}
    \mathbb{E}[\bm{\phi}(\mathbf{X}_{t_k}) \mid \mathbf{X}_{t_{k-1}} = \mathbf{x}] = \sum_{n = 0}^\infty \frac{h^n}{n!} \mathbb{L}^n\bm{\phi}(\mathbf{x}),
\end{equation}
where $\mathbb{L}^0\bm{\phi}(\mathbf{x}) = \bm{\phi}(\mathbf{x})$. 
We then apply it to the middle term  in \eqref{eq:Omegah_int} and obtain 
\begin{align*}
    \mathbb{E}^{[1]}[\bm{\Sigma}\bm{\Sigma}^\top(\mathbf{X}_s^{[1]}) \mid \mathbf{X}_{t_{k-1}}^{[1]} = \mathbf{x} ] = \bm{\Sigma}\bm{\Sigma}^\top(\mathbf{x}) + (s-t_{k-1}) \mathbb{L}^{[1]}\bm{\Sigma}\bm{\Sigma}^\top(\mathbf{x}) + \mathbf{R}((s-t_{k-1})^2, \mathbf{x}).
\end{align*}
Then, \eqref{eq:Omegah_int} becomes 
\begin{align}
    \bm{\Omega}_h(\mathbf{x}) 
    &= h\bm{\Sigma}\bm{\Sigma}^\top(\mathbf{x}) + \frac{h^2}{2}\left(\mathbf{A} \bm{\Sigma}\bm{\Sigma}^\top(\mathbf{x}) + \bm{\Sigma}\bm{\Sigma}^\top(\mathbf{x}) \mathbf{A}^\top  + \mathbb{L}^{[1]}\bm{\Sigma}\bm{\Sigma}^\top(\mathbf{x})\right)+ \mathbf{R}(h^3, \mathbf{x}). \label{eq:Omega_approx}
\end{align}
As Strang splitting approximation \eqref{eq:StrangSplitting} uses $\bm{\Omega}_h(\bm{f}_{h/2}(\mathbf{x}))$,   we   approximate it using  Lemma \ref{lemma:fh} and the fact that $\bm{f}_{h/2}(\mathbf{x}) = \mathbf{x} + h/2 \mathbf{N}(\mathbf{x}) + \mathbf{R}(h^2, \mathbf{x})$. We get
\begin{align}
    [\bm{\Sigma}\bm{\Sigma}^\top(\bm{f}_{h/2}(\mathbf{x}))]_{ij} 
    &=[\bm{\Sigma}\bm{\Sigma}^\top(\mathbf{x}))]_{ij} + \frac{h}{2} (D [\bm{\Sigma}\bm{\Sigma}^\top(\mathbf{x}))]_{ij})^\top \mathbf{N}(\mathbf{x}) + \mathbf{R}(h^2, \mathbf{x}). \label{eq:SigmaSigmaTf}
\end{align}

Now, we derive $\mathbb{L}^{[1]}\bm{\Sigma}\bm{\Sigma}^\top(\mathbf{x})$ element-wise
\begin{align}
    [\mathbb{L}^{[1]}\bm{\Sigma}\bm{\Sigma}^\top(\mathbf{x})]_{ij} &= \mathbb{L}^{[1]}[\bm{\Sigma}\bm{\Sigma}^\top(\mathbf{x})]_{ij} = (D [\bm{\Sigma}\bm{\Sigma}^\top(\mathbf{x}))]_{ij})^\top \mathbf{A}(\mathbf{x}-\mathbf{b}) + \sum_{k, l = 1}^d \alpha^{ij}_{kl} [\bm{\Sigma}\bm{\Sigma}^\top(\mathbf{x})]_{kl}. \label{eq:L1SigmaSigmaT}
\end{align}
Finally, combining \eqref{eq:L1}, \eqref{eq:Omega_approx}, \eqref{eq:SigmaSigmaTf} and \eqref{eq:L1SigmaSigmaT}, we get \eqref{eq:OmegahSS_approx}. The rest of the proof is the same as in Lemma 6.2 in \citep{pilipovic2024SecondOrder}

\subsection{Proof of Theorem \ref{thm:Consistency}}
We start by approximating the nonlinear solution $\bm{f}_h$ (Section 1.8 in \citep{SolvingODEI}), $\log |\det D {\bm{f}}_{h/2}(\mathbf{x}; \bm{\theta}^{(1)})|$ which appears in the objective function and $\bm{\mu}_h(\bm{f}_{h/2}(\mathbf{x}))$ up to the lowest necessary order of $h$.
\begin{lemma} \label{lemma:fh} Let Assumptions \ref{as:monoton_condition}, \ref{as:F_polynomial_growth} and \ref{as:fhInv} hold. When $h \to 0$, the $h$-flow of  \eqref{eq:SplittingEq2} approximates as
\begin{align*}
    \bm{f}_h(\mathbf{x}) &= \mathbf{x} + h \mathbf{N}(\mathbf{x}) + \frac{h^2}{2} (D \mathbf{N}(\mathbf{x}))\mathbf{N}(\mathbf{x}) + \mathbf{R}(h^3, \mathbf{x}),\\
    \bm{f}^{-1}_h(\mathbf{x}) &= \mathbf{x} - h \mathbf{N}(\mathbf{x}) + \frac{h^2}{2} (D\mathbf{N}(\mathbf{x}))\mathbf{N}(\mathbf{x}) + \mathbf{R}(h^3, \mathbf{x}).
\end{align*}
For the $h$-flow of \eqref{eq:SplittingEq2}, $\bm{f}_h$, holds
\begin{align*}
        2 \log |\det D {\bm{f}}_{h/2}(\mathbf{X}_{t_k}; \bm{\theta}^{(1)})| &=  h \tr D{\mathbf{N}}(\mathbf{X}_{t_{k-1}}; \bm{\theta}^{(1)})  + {R}(h^2, \mathbf{X}_{t_{k-1}})
    \end{align*}
For the functions $\bm{f}_{h}$ in \eqref{eq:fhflow} and $\bm{\mu}_h$ in \eqref{eq:mu_h}, it holds
\begin{align*}
    \bm{\mu}_h(\bm{f}_{h/2}(\mathbf{x})) &= \mathbf{x} + h(\mathbf{F}(\mathbf{x}) - \frac{1}{2} \mathbf{N}(\mathbf{x})) + \mathbf{R}(h^2,\mathbf{x}). 
    \end{align*}
\end{lemma}

To approximate the random variable $\mathbf{Z}_{t_k}(\bm{\theta}^{(1)})$ in \eqref{eq:Ztk} around $\mathbf{X}_{t_{k-1}}$, we start by defining the following the random sequences  
\begin{align}    
    \bm{\eta}(\mathbf{X}_{t_{k-1}};\bm{\theta}_0)\coloneqq \frac{1}{h^{1/2}}\int_{t_{k-1}}^{t_k} \bm{\Sigma}(\mathbf{X}_t;\bm{\theta}^{(2)}_0) \dif \mathbf{W}_t, && \bm{\xi}(\mathbf{X}_{t_{k-1}};\bm{\theta}_0)  \coloneqq  \frac{1}{h^{3/2}} \int_{t_{k-1}}^{t_k} (t_k - t)\bm{\Sigma}(\mathbf{X}_t;\bm{\theta}^{(2)}_0) \dif\mathbf{W}_t. \label{eq:eta_xi}
\end{align}
The random variables $\bm{\eta}(\mathbf{X}_{t_{k-1}};\bm{\theta}_0)$ and $\bm{\xi}(\mathbf{X}_{t_{k-1}};\bm{\theta}_0)$ in \eqref{eq:eta_xi} have zero mean. Moreover, at time $t_{k-1}$ they are $\mathcal{F}_{t_k}$ measurable and independent of $\mathcal{F}_{t_{k-1}}$. Using It\^{o} formula it is easy to see that $\bm{\eta}(\mathbf{X}_{t_{k-1}};\bm{\theta}_0)$ and $\bm{\xi}(\mathbf{X}_{t_{k-1}};\bm{\theta}_0)$ depend on the drift parameter $\bm{\theta}^{(1)}_0$. For short, we write $\bm{\eta}_0(\mathbf{X}_{t_{k-1}}) \coloneqq \bm{\eta}(\mathbf{X}_{t_{k-1}};\bm{\theta}_0)$, $\bm{\xi}_0(\mathbf{X}_{t_{k-1}}) \coloneqq \bm{\xi}(\mathbf{X}_{t_{k-1}};\bm{\theta}_0)$,  $\bm{\eta}\bm{\eta}_0^\top(\mathbf{X}_{t_{k-1}}) \coloneqq \bm{\eta}(\mathbf{X}_{t_{k-1}};\bm{\theta}_0)\bm{\eta}^\top(\mathbf{X}_{t_{k-1}};\bm{\theta}_0)$ and $\mathbf{F}_0(\mathbf{X}_{t_{k-1}}) \coloneqq \mathbf{F}(\mathbf{X}_{t_{k-1}}; \bm{\theta}^{(1)}_0)$.

The It\^{o} isometry yields
\begin{align}
    \mathbb{E}_{\bm{\theta}_0}[\bm{\eta}_0(\mathbf{X}_{t_{k-1}})\bm{\eta}_0(\mathbf{X}_{t_{k-1}})^\top \mid \mathcal{F}_{t_{k-1}}] &=  \bm{\Sigma}\bm{\Sigma}_0^\top(\mathbf{X}_{t_{k-1}}) + \frac{h}{2}\mathbb{L}_0\bm{\Sigma}\bm{\Sigma}_0^\top(\mathbf{X}_{t_{k-1}})+ \mathbf{R}(h^2, \mathbf{X}_{t_{k-1}}), \label{eq:etaeta}\\
    \mathbb{E}_{\bm{\theta}_0}[\bm{\eta}_0(\mathbf{X}_{t_{k-1}})\bm{\xi}_0(\mathbf{X}_{t_{k-1}})^\top \mid \mathcal{F}_{t_{k-1}}] &= \mathbb{E}_{\bm{\theta}_0}[\bm{\xi}_0(\mathbf{X}_{t_{k-1}})\bm{\eta}_0(\mathbf{X}_{t_{k-1}})^\top \mid \mathcal{F}_{t_{k-1}}] =  \frac{1}{2}\bm{\Sigma}\bm{\Sigma}_0^\top(\mathbf{X}_{t_{k-1}}) + \mathbf{R}(h, \mathbf{X}_{t_{k-1}}), \notag \\
    \mathbb{E}_{\bm{\theta}_0}[\bm{\xi}_0(\mathbf{X}_{t_{k-1}})\bm{\xi}_0(\mathbf{X}_{t_{k-1}})^\top \mid \mathcal{F}_{t_{k-1}}] &= \frac{1}{3}\bm{\Sigma}\bm{\Sigma}_0^\top(\mathbf{X}_{t_{k-1}}) + \mathbf{R}(h, \mathbf{X}_{t_{k-1}}). \notag
\end{align}
The following proposition is the last building block for approximating the objective function $\mathcal{L}^{\mathrm{[SS]}}$ in \eqref{eq:S_obj}. It derives from the previous Lemma. 

\begin{proposition} \label{prop:ZtkandZtkbar} The random variable $\mathbf{Z}_{t_k}(\bm{\theta}^{(1)})$ in \eqref{eq:Ztk} is approximated as
\begin{align*}
    \mathbf{Z}_{t_k}(\bm{\theta}^{(1)}) &= h^{1/2}\bm{\eta}(\mathbf{X}_{t_{k-1}}) + h (\mathbf{F}_0(\mathbf{X}_{t_{k-1}}) - \mathbf{F}(\mathbf{X}_{t_{k-1}}))\\
    &- \frac{h^{3/2}}{2}D\mathbf{N}(\mathbf{X}_{t_{k-1}}) \bm{\eta}(\mathbf{X}_{t_{k-1}})+h^{3/2} D\mathbf{F}_0(\mathbf{X}_{t_{k-1}}) \bm{\xi}(\mathbf{X}_{t_{k-1}}) + \mathbf{R}(h^2,\mathbf{X}_{t_{k-1}}).
\end{align*}
\end{proposition}

We approximate the objective function $\mathcal{L}^{\mathrm{[SS]}}$ in \eqref{eq:S_obj} up to order $R(h^{3/2}, \mathbf{X}_{t_{k-1}})$ by 
\begin{align}
     \mathcal{L}_N^{\mathrm{[SS]}}(\mathbf{X}_{0:t_N}; \bm{\theta}) &\coloneqq\sum_{k=1}^{N}\left( \log \det \bm{\Sigma}\bm{\Sigma}^\top(\mathbf{X}_{t_{k-1}}; \bm{\theta}^{(2)})  + \bm{\eta}_0^\top(\mathbf{X}_{t_{k-1}}) \bm{\Sigma}\bm{\Sigma}^\top(\mathbf{X}_{t_{k-1}}; \bm{\theta}_0^{(2)})^{-1}\bm{\eta}(\mathbf{X}_{t_{k-1}})\right)\notag\\
     &\hspace{-5ex}+ 2\sqrt{h}\sum_{k=1}^{N}  \bm{\eta}_0^\top(\mathbf{X}_{t_{k-1}}) \bm{\Sigma}\bm{\Sigma}^\top(\mathbf{X}_{t_{k-1}}; \bm{\theta}^{(2)})^{-1}(\mathbf{F}(\mathbf{X}_{t_{k-1}}; \bm{\theta}^{(1)}_0) - \mathbf{F}(\mathbf{X}_{t_{k-1}}; \bm{\theta}^{(1)}))\notag\\ 
     &\hspace{-5ex}+ h\sum_{k=1}^{N} (\mathbf{F}(\mathbf{X}_{t_{k-1}}; \bm{\theta}^{(1)}_0) - \mathbf{F}(\mathbf{X}_{t_{k-1}}; \bm{\theta}^{(1)}))^\top \bm{\Sigma}\bm{\Sigma}^\top(\mathbf{X}_{t_{k-1}}; \bm{\theta}^{(2)})^{-1}(\mathbf{F}(\mathbf{X}_{t_{k-1}}; \bm{\theta}^{(1)}_0) - \mathbf{F}(\mathbf{X}_{t_{k-1}}; \bm{\theta}^{(1)}))\notag\\
     &\hspace{-5ex}- h\sum_{k=1}^{N} \left(\bm{\eta}_0^\top(\mathbf{X}_{t_{k-1}})D \mathbf{F}(\mathbf{X}_{t_{k-1}}; \bm{\theta}^{(1)})^\top\bm{\Sigma}\bm{\Sigma}^\top(\mathbf{X}_{t_{k-1}}; \bm{\theta}^{(2)})^{-1}\bm{\eta}_0(\mathbf{X}_{t_{k-1}}) - \tr D\mathbf{F}(\mathbf{X}_{t_k}; \bm{\theta}^{(1)})\right)\notag\\
     &\hspace{-5ex}-\frac{h}{2}\sum_{k=1}^{N} \bm{\eta}_0^\top(\mathbf{X}_{t_{k-1}}) \bm{\Sigma}\bm{\Sigma}^\top(\mathbf{X}_{t_{k-1}}; \bm{\theta}^{(2)})^{-1}\mathbb{L}_{\bm{\theta}}\bm{\Sigma}\bm{\Sigma}^\top(\mathbf{X}_{t_{k-1}}; \bm{\theta}^{(2)})\bm{\Sigma}\bm{\Sigma}^\top(\mathbf{X}_{t_{k-1}}; \bm{\theta}_0^{(2)})^{-1}\bm{\eta}(\mathbf{X}_{t_{k-1}})\notag\\
     &\hspace{-5ex}+\frac{h}{2}\sum_{k=1}^{N}\tr \mathbb{L}_{\bm{\theta}}\bm{\Sigma}\bm{\Sigma}^\top(\mathbf{X}_{t_{k-1}}; \bm{\theta}^{(2)})\bm{\Sigma}\bm{\Sigma}^\top(\mathbf{X}_{t_{k-1}}; \bm{\theta}^{(2)})^{-1}. \label{eq:asymptotic_S_obj}
\end{align}
This approximation is enough to prove the asymptotic properties of the estimator $\hat{\bm{\theta}}_N^{\mathrm{[SS]}}$. It omits the terms of order $R(h, \mathbf{X}_{t_{k-1}})$ that do not depend on $\bm{\theta}^{(1)}$. It holds
\begin{align*}
    \mathcal{L}^{\mathrm{[SS]}}(\mathbf{X}_{0:t_N}; \bm{\theta}) &= \mathcal{L}_N^{\mathrm{[SS]}}(\mathbf{X}_{0:t_N}; \bm{\theta}) \\
    &+2h\sum_{k=1}^{N} \bm{\eta}^\top(\mathbf{X}_{t_{k-1}}; \bm{\theta}_0) \bm{\Sigma}\bm{\Sigma}^\top(\mathbf{X}_{t_{k-1}}; \bm{\theta}^{(2)})^{-1}D\mathbf{F}(\mathbf{X}_{t_{k-1}}; \bm{\theta}^{(1)}_0) \bm{\xi}^\top(\mathbf{X}_{t_{k-1}}; \bm{\theta}_0^{(2)})  + R(h^{3/2}, \mathbf{X}_{t_{k-1}}).
\end{align*}
To establish consistency, we follow the proof of Theorem 1 in \citet{Kessler1997} and study the limit of $\mathcal{L}_N^{\mathrm{[SS]}}(\bm{\theta})$ from \eqref{eq:asymptotic_S_obj} rescaled by the correct rate of convergence. More precisely, the consistency of the diffusion parameter is proved by studying the limit of $\frac{1}{N}\mathcal{L}_N^{\mathrm{[SS]}}(\bm{\theta})$. In contrast, the consistency of the drift parameter is proved by studying the limit of $\frac{1}{Nh}(\mathcal{L}_N^{\mathrm{[SS]}}(\bm{\theta}^{(1)}, \bm{\theta}^{(2)}) - \mathcal{L}_N^{\mathrm{[SS]}}(\bm{\theta}^{(1)}_0, \bm{\theta}^{(2)})$.

To prove the consistency of the diffusion parameter $\bm{\theta}^{(2)}$, we need to prove that
\begin{equation}
    \frac{1}{N} \mathcal{L}_N^{\mathrm{[SS]}} (\bm{\theta}^{(1)},\bm{\theta}^{(2)}) \xrightarrow[\substack{h \to 0 \\ Nh \to \infty}]{\mathbb{P}_{\bm{\theta}_0}} \int \left(\log \det \bm{\Sigma}\bm{\Sigma}^\top(\mathbf{x}; \bm{\theta}^{(2)}) + \tr(\bm{\Sigma}\bm{\Sigma}^\top(\mathbf{x}; \bm{\theta}^{(2)})^{-1}\bm{\Sigma}\bm{\Sigma}^\top(\mathbf{x}; \bm{\theta}^{(2)}_0)\right) \dif \nu_0(\mathbf{x}), \label{eq:LikConvSigma}
\end{equation}
uniformly in $\bm{\theta}$. This can be done following  the proof of consistency for $\widehat{\bm{\theta}}_N^{(2)\mathrm{[SS]}}$ in \citep{Kessler1997}.  

To prove the consistency of the drift estimators $\widehat{\bm{\theta}}_N^{(1)\mathrm{[SS]}}$, we start by finding the limit in $\mathbb{P}_{\bm{\theta}_0}$ of  
\begin{align}
     \frac{1}{N h} ( \mathcal{L}_N^{\mathrm{[SS]}}(\bm{\theta}^{(1)}, \bm{\theta}^{(2)})  - \mathcal{L}_N^{\mathrm{[SS]}}(\bm{\theta}^{(1)}_0, \bm{\theta}^{(2)})), \label{eq:LikConsBeta}
\end{align}
for $N h \to \infty$, $h \to 0$, uniformly in $\bm{\theta}$. Starting with \eqref{eq:asymptotic_S_obj}, we get 
\begin{align*}
    & \frac{1}{N h} ( \mathcal{L}_N^{\mathrm{[SS]}}(\bm{\theta}^{(1)}, \bm{\theta}^{(2)})  - \mathcal{L}_N^{\mathrm{[SS]}}(\bm{\theta}^{(1)}_0, \bm{\theta}^{(2)}) = \\
    &\frac{2}{N\sqrt{h}}\sum_{k=1}^{N}  \bm{\eta}_0^\top(\mathbf{X}_{t_{k-1}}) \bm{\Sigma}\bm{\Sigma}^\top(\mathbf{X}_{t_{k-1}}; \bm{\theta}^{(2)})^{-1}(\mathbf{F}(\mathbf{X}_{t_{k-1}}; \bm{\theta}^{(1)}_0) - \mathbf{F}(\mathbf{X}_{t_{k-1}}; \bm{\theta}^{(1)}))\\
    &+ \frac{1}{N}\sum_{k=1}^{N} (\mathbf{F}(\mathbf{X}_{t_{k-1}}; \bm{\theta}^{(1)}_0) - \mathbf{F}(\mathbf{X}_{t_{k-1}}; \bm{\theta}^{(1)}))^\top \bm{\Sigma}\bm{\Sigma}^\top(\mathbf{X}_{t_{k-1}}; \bm{\theta}^{(2)})^{-1}(\mathbf{F}(\mathbf{X}_{t_{k-1}}; \bm{\theta}^{(1)}_0) - \mathbf{F}(\mathbf{X}_{t_{k-1}}; \bm{\theta}^{(1)}))\notag\\
    &- \frac{1}{N}\sum_{k=1}^{N} \bm{\eta}_0^\top(\mathbf{X}_{t_{k-1}})(D \mathbf{F}(\mathbf{X}_{t_{k-1}}; \bm{\theta}^{(1)}) - D \mathbf{F}(\mathbf{X}_{t_{k-1}}; \bm{\theta}^{(1)}_0))^\top\bm{\Sigma}\bm{\Sigma}^\top(\mathbf{X}_{t_{k-1}}; \bm{\theta}^{(2)})^{-1}\bm{\eta}_0(\mathbf{X}_{t_{k-1}})\\
    &+\frac{1}{N}\sum_{k=1}^{N} \tr(D\mathbf{F}(\mathbf{X}_{t_k}; \bm{\theta}^{(1)}) - D\mathbf{F}(\mathbf{X}_{t_k}; \bm{\theta}^{(1)}_0))\\
    &-\frac{h}{2}\sum_{k=1}^{N} \bm{\eta}_0^\top(\mathbf{X}_{t_{k-1}}) \bm{\Sigma}\bm{\Sigma}^\top(\mathbf{X}_{t_{k-1}}; \bm{\theta}^{(2)})^{-1}(\mathbb{L}_{\bm{\theta}^{(1)},\bm{\theta}^{(2)}} - \mathbb{L}_{\bm{\theta}^{(1)}_0, \bm{\theta}^{(2)}})\{\bm{\Sigma}\bm{\Sigma}^\top(\mathbf{X}_{t_{k-1}}; \bm{\theta}^{(2)})\}\\
    &\hspace{70ex}\bm{\Sigma}\bm{\Sigma}^\top(\mathbf{X}_{t_{k-1}}; \bm{\theta}^{(2)})^{-1}\bm{\eta}_0(\mathbf{X}_{t_{k-1}})\\
    &+\frac{h}{2}\sum_{k=1}^{N}\tr ((\mathbb{L}_{\bm{\theta}^{(1)},\bm{\theta}^{(2)}} - \mathbb{L}_{\bm{\theta}^{(1)}_0, \bm{\theta}^{(2)}})\{\bm{\Sigma}\bm{\Sigma}^\top(\mathbf{X}_{t_{k-1}}; \bm{\theta}^{(2)})\}\bm{\Sigma}\bm{\Sigma}^\top(\mathbf{X}_{t_{k-1}}; \bm{\theta}^{(2)})^{-1}).
\end{align*}
To prove the convergence in probability of the previous sequence, we use Lemma 8 in \citep{Kessler1997} and Lemma 9 in \citep{GenonCatalot&Jacod}. To apply Lemma 9 from \citep{GenonCatalot&Jacod}, we need to show that the sum of expectations converges to a certain value while the sum of covariances converges to zero. Here, we only show the former. Moreover, standard tools like Proposition A1 in \citep{Gloter2006} or Lemma 3.1 in \citep{Yoshida1990} can be used to prove uniform convergence. Thus, we look at the expectation to find the limit of \eqref{eq:LikConsBeta}. We use  the known covariances \eqref{eq:etaeta} to get
\begin{align}
    \frac{1}{N h} ( \mathcal{L}_N^{\mathrm{[SS]}}(\bm{\theta}^{(1)}, \bm{\theta}^{(2)})  - \mathcal{L}_N^{\mathrm{[SS]}}(\bm{\theta}^{(1)}_0, \bm{\theta}^{(2)}) \xrightarrow[\substack{h \to 0 \\ Nh \to \infty}]{\mathbb{P}_{\bm{\theta}_0}} &\int(\mathbf{F}_0(\mathbf{x}) - \mathbf{F}(\mathbf{x}))^\top \bm{\Sigma}\bm{\Sigma}^\top(\mathbf{x})^{-1} (\mathbf{F}_0(\mathbf{x}) - \mathbf{F}(\mathbf{x})) \dif \nu_0(\mathbf{x})\label{eq:drift_cons}\\
    +&\int \tr (D\left(\mathbf{F}_0\left(\mathbf{x}\right)-\mathbf{F}\left(\mathbf{x}\right)\right)(\bm{\Sigma}\bm{\Sigma}_0^\top(\mathbf{x})\bm{\Sigma}\bm{\Sigma}^\top(\mathbf{x})^{-1} - \mathbf{I}))\dif \nu_0(\mathbf{x}) \notag\\
    &\hspace{-29ex}+\frac{1}{2}\int \tr ((\mathbb{L}_{\bm{\theta}^{(1)}_0,\bm{\theta}^{(2)}} - \mathbb{L}_{\bm{\theta}^{(1)}, \bm{\theta}^{(2)}})\{\bm{\Sigma}\bm{\Sigma}^\top(\mathbf{x})\}\bm{\Sigma}\bm{\Sigma}^\top(\mathbf{x})^{-1}(\bm{\Sigma}\bm{\Sigma}_0^\top(\mathbf{x})\bm{\Sigma}\bm{\Sigma}^\top(\mathbf{x})^{-1} - \mathbf{I}))\dif \nu_0(\mathbf{x}). \notag
\end{align}
Thus, the drift estimator's consistency follows the diffusion estimator's consistency. This is because the right-hand side of \eqref{eq:drift_cons} is non-negative when $\bm{\theta}^{(2)} = \bm{\theta}^{(2)}_0$, and the left-hand side is non-positive, following the definition of the likelihood.

\subsection{Proof of the asymptotic normality}

To prove Theorem \ref{thm:AsymtoticNormality}, it is enough to prove Lemmas \ref{lemma:AsymptoticNormality1} and \ref{lemma:LnConvergence}. 

In the following proofs, we use the fact that $\partial_{\theta^{(i)}} R(h^p, \mathbf{Y}_{t_{k-1}}) = R(h^p, \mathbf{Y}_{t_{k-1}})$, for any $p$ and $i$.

\begin{proof}[Proof of Lemma \ref{lemma:AsymptoticNormality1}]
We start by proving the first part of the lemma. First, we find their second derivatives with respect to $\bm{\theta}^{(1)}$
    \begin{align*}
        \frac{1}{N h} \partial_{\theta^{(1)}_{i_1}\theta^{(1)}_{i_2}}^2\mathcal{L}_N^{\mathrm{[SS]}}\left(\mathbf{X}_{0:t_N}; \bm{\theta}\right) &=-\frac{2}{N\sqrt{h}} \sum_{k=1}^{N}   \bm{\eta}_0^\top(\mathbf{X}_{t_{k-1}})     \bm{\Sigma}\bm{\Sigma}^\top(\mathbf{X}_{t_{k-1}}; \bm{\theta}^{(2)})^{-1}\partial_{\theta^{(1)}_{i_1}\theta^{(1)}_{i_2}}^2 \mathbf{F}(\mathbf{X}_{t_{k-1}}; \bm{\theta}^{(1)}) \notag\\
        &\hspace{-25ex}+ \frac{2}{N}\sum_{k=1}^{N} \partial_{\theta^{(1)}_{i_1}}\mathbf{F}(\mathbf{X}_{t_{k-1}}; \bm{\theta}^{(1)})^\top  \bm{\Sigma}\bm{\Sigma}^\top(\mathbf{X}_{t_{k-1}}; \bm{\theta}^{(2)})^{-1}\partial_{\theta^{(1)}_{i_2}}\mathbf{F}(\mathbf{X}_{t_{k-1}}; \bm{\theta}^{(1)})\notag\\
        &\hspace{-25ex}- \frac{2}{N}\sum_{k=1}^{N} \partial_{\theta^{(1)}_{i_1}\theta^{(1)}_{i_2}}^2\mathbf{F}(\mathbf{X}_{t_{k-1}}; \bm{\theta}^{(1)})^\top  \bm{\Sigma}\bm{\Sigma}^\top(\mathbf{X}_{t_{k-1}}; \bm{\theta}^{(2)})^{-1}(\mathbf{F}(\mathbf{X}_{t_{k-1}}; \bm{\theta}^{(1)}_0) - \mathbf{F}(\mathbf{X}_{t_{k-1}}; \bm{\theta}^{(1)}))\notag\\
        &\hspace{-25ex}- \frac{1}{N}\sum_{k=1}^{N} \bm{\eta}_0^\top(\mathbf{X}_{t_{k-1}}) D \partial_{\theta^{(1)}_{i_1}\theta^{(1)}_{i_2}}^2\mathbf{F}(\mathbf{X}_{t_{k-1}}; \bm{\theta}^{(1)})^\top 
        \bm{\Sigma}\bm{\Sigma}^\top(\mathbf{X}_{t_{k-1}}; \bm{\theta}_0^{(2)})^{-1}\bm{\eta}(\mathbf{X}_{t_{k-1}}) + \frac{1}{N}\sum_{k=1}^{N} \tr D  \partial_{\theta^{(1)}_{i_1}\theta^{(1)}_{i_2}}^2 \mathbf{F}(\mathbf{X}_{t_k}; \bm{\theta}^{(1)})\\
        &\hspace{-25ex}-\frac{1}{2N}\sum_{k=1}^{N} \bm{\eta}_0^\top(\mathbf{X}_{t_{k-1}}) \bm{\Sigma}\bm{\Sigma}^\top(\mathbf{X}_{t_{k-1}}; \bm{\theta}^{(2)})^{-1}  \partial_{\theta^{(1)}_{i_1}\theta^{(1)}_{i_2}}^2\mathbb{L}_{\bm{\theta}}\bm{\Sigma}\bm{\Sigma}^\top(\mathbf{X}_{t_{k-1}}; \bm{\theta}^{(2)})\bm{\Sigma}\bm{\Sigma}^\top(\mathbf{X}_{t_{k-1}}; \bm{\theta}^{(2)})^{-1}\bm{\eta}_0(\mathbf{X}_{t_{k-1}})\\
        &\hspace{-25ex}+\frac{1}{2N}\sum_{k=1}^{N}\tr   \partial_{\theta^{(1)}_{i_1}\theta^{(1)}_{i_2}}^2\mathbb{L}_{\bm{\theta}}\bm{\Sigma}\bm{\Sigma}^\top(\mathbf{X}_{t_{k-1}}; \bm{\theta}^{(2)})\bm{\Sigma}\bm{\Sigma}^\top(\mathbf{X}_{t_{k-1}}; \bm{\theta}^{(2)})^{-1}.
    \end{align*}
As in the proof of consistency, it holds
\begin{align*}
    \frac{1}{N h} \partial_{\theta^{(1)}_{i_1}\theta^{(1)}_{i_2}}^2\mathcal{L}_N^{\mathrm{[SS]}} \left(\mathbf{X}_{0:t_N}; \bm{\theta}\right) \Big|_{\bm{\theta} = \bm{\theta}_0} &\xrightarrow[\substack{h \to 0 \\ Nh \to \infty}]{\mathbb{P}_{\bm{\theta}_0}} 2 \int  \partial_{\theta^{(1)}_{i_1}}\mathbf{F}_0(\mathbf{x})^\top  \bm{\Sigma}\bm{\Sigma}_0^\top(\mathbf{x})^{-1}  \partial_{\theta^{(1)}_{i_2}}\mathbf{F}_0(\mathbf{x}) \dif \nu_0(\mathbf{x}),
\end{align*}
uniformly in $\bm{\theta}$. Now, we investigate the limit of $\frac{1}{N \sqrt{h}} \partial_{\theta^{(1)}_{i_1}\theta^{(2)}_{j_2}}^2\mathcal{L}_N^{\mathrm{[SS]}}(\bm{\theta})$ 
\begin{align*}
        \frac{1}{N \sqrt{h}} \partial_{\theta^{(1)}_{i_1}\theta^{(2)}_{j_2}}^2\mathcal{L}_N^{\mathrm{[SS]}}\left(\mathbf{X}_{0:t_N}; \bm{\theta}\right) &= -\frac{2}{N} \sum_{k=1}^{N}  \bm{\eta}^\top(\mathbf{X}_{t_{k-1}}; \bm{\theta}_0)  \partial_{\theta^{(2)}_{j_2}} \bm{\Sigma}\bm{\Sigma}^\top(\mathbf{X}_{t_{k-1}}; \bm{\theta}^{(2)})^{-1}\partial_{\theta^{(1)}_{i_1}} \mathbf{F}(\mathbf{X}_{t_{k-1}}; \bm{\theta}^{(1)})\notag\\
        &+ \frac{1}{N}\sum_{k=1}^N R(\sqrt{h}, \mathbf{X}_{t_{k-1}})\notag.
\end{align*}
The previous sequence converges to zero due to Lemma 9 in \citep{GenonCatalot&Jacod}. Next, we look at the limits of $\frac{1}{N} \partial_{\theta^{(2)}_{j_1}\theta^{(2)}_{j_2}}^2\mathcal{L}_N^{\mathrm{[SS]}}(\bm{\theta})$ 
   \begin{align*}
        \frac{1}{N} \partial_{\theta^{(2)}_{j_1}\theta^{(2)}_{j_2}}^2\mathcal{L}_N^{\mathrm{[SS]}}\left(\mathbf{X}_{0:t_N}; \bm{\theta}\right)&=  \frac{1}{N} \sum_{k=1}^{N} \partial_{\theta^{(2)}_{j_1}\theta^{(2)}_{j_2}}^2 \log \det  \bm{\Sigma}\bm{\Sigma}^\top(\mathbf{X}_{t_{k-1}}; \bm{\theta}^{(2)})  \\
        &\hspace{-22ex}+ \frac{1}{N}\sum_{k=1}^{N}\partial_{\theta^{(2)}_{j_1}\theta^{(2)}_{j_2}}^2\tr(\bm{\eta}\bm{\eta}_0^\top(\mathbf{X}_{t_{k-1}})   \bm{\Sigma}\bm{\Sigma}^\top(\mathbf{X}_{t_{k-1}}; \bm{\theta}^{(2)})^{-1})+ \frac{1}{N}\sum_{k=1}^N R(\sqrt{h}, \mathbf{X}_{t_{k-1}})\\
        &\hspace{-22ex}=  \frac{1}{N} \sum_{k=1}^{N}\tr( \bm{\Sigma}\bm{\Sigma}^\top(\mathbf{X}_{t_{k-1}}; \bm{\theta}^{(2)})^{-1} \partial_{\theta^{(2)}_{j_1}\theta^{(2)}_{j_2}}^2 \bm{\Sigma}\bm{\Sigma}^\top(\mathbf{X}_{t_{k-1}}; \bm{\theta}^{(2)})) \\
        &\hspace{-22ex}- \frac{1}{N} \sum_{k=1}^{N}\tr ( \bm{\Sigma}\bm{\Sigma}^\top(\mathbf{X}_{t_{k-1}}; \bm{\theta}^{(2)})^{-1} ( \partial_{\theta^{(2)}_{j_1}} \bm{\Sigma}\bm{\Sigma}^\top(\mathbf{X}_{t_{k-1}}; \bm{\theta}^{(2)})) \bm{\Sigma}\bm{\Sigma}^\top(\mathbf{X}_{t_{k-1}}; \bm{\theta}^{(2)})^{-1} \partial_{\theta^{(2)}_{j_2}} \bm{\Sigma}\bm{\Sigma}^\top(\mathbf{X}_{t_{k-1}}; \bm{\theta}^{(2)}))\notag\\
        &\hspace{-22ex}-\frac{1}{N}\sum_{k=1}^{N}\tr(\bm{\eta}\bm{\eta}_0^\top(\mathbf{X}_{t_{k-1}})  \bm{\Sigma}\bm{\Sigma}^\top(\mathbf{X}_{t_{k-1}}; \bm{\theta}^{(2)})^{-1}(\partial_{\theta^{(2)}_{j_1}\theta^{(2)}_{j_2}}^2 \bm{\Sigma}\bm{\Sigma}^\top)  \bm{\Sigma}\bm{\Sigma}^\top(\mathbf{X}_{t_{k-1}}; \bm{\theta}^{(2)})^{-1})\notag\\
        &\hspace{-22ex} +\frac{1}{N}\sum_{k=1}^{N}\tr(\bm{\eta}\bm{\eta}_0^\top(\mathbf{X}_{t_{k-1}})  \bm{\Sigma}\bm{\Sigma}^\top(\mathbf{X}_{t_{k-1}}; \bm{\theta}^{(2)})^{-1}(\partial_{\theta^{(2)}_{j_1}} \bm{\Sigma}\bm{\Sigma}^\top)  \bm{\Sigma}\bm{\Sigma}^\top(\mathbf{X}_{t_{k-1}}; \bm{\theta}^{(2)})^{-1}
        (\partial_{\theta^{(2)}_{j_2}} \bm{\Sigma}\bm{\Sigma}^\top)  \bm{\Sigma}\bm{\Sigma}^\top(\mathbf{X}_{t_{k-1}}; \bm{\theta}^{(2)})^{-1})\notag\\
        &\hspace{-22ex} +\frac{1}{N}\sum_{k=1}^{N}\tr(\bm{\eta}\bm{\eta}_0^\top(\mathbf{X}_{t_{k-1}})  \bm{\Sigma}\bm{\Sigma}^\top(\mathbf{X}_{t_{k-1}}; \bm{\theta}^{(2)})^{-1}(\partial_{\theta^{(2)}_{j_2}} \bm{\Sigma}\bm{\Sigma}^\top)  \bm{\Sigma}\bm{\Sigma}^\top(\mathbf{X}_{t_{k-1}}; \bm{\theta}^{(2)})^{-1}
        (\partial_{\theta^{(2)}_{j_1}} \bm{\Sigma}\bm{\Sigma}^\top)  \bm{\Sigma}\bm{\Sigma}^\top(\mathbf{X}_{t_{k-1}}; \bm{\theta}^{(2)})^{-1})\\
        &\hspace{-22ex}+ \frac{1}{N}\sum_{k=1}^N R(\sqrt{h}, \mathbf{X}_{t_{k-1}}).
\end{align*}
Using the conditional covariance matrix of of $\bm{\eta}(\mathbf{X}_{t_{k-1}}; \bm{\theta}_0)$ from \eqref{eq:eta_xi} with additional calculations, we can conclude that
\begin{align*}
    \frac{1}{N} \partial_{\theta^{(2)}_{j_1}\theta^{(2)}_{j_2}}^2\mathcal{L}_N^{\mathrm{[SS]}}\left(\mathbf{X}_{0:t_N}; \bm{\theta}\right) \Big|_{\bm{\theta} = \bm{\theta}_0} &\xrightarrow[\substack{h \to 0 \\ Nh \to \infty}]{\mathbb{P}_{\bm{\theta}_0}}  \int\tr ( \bm{\Sigma}\bm{\Sigma}_0^\top(\mathbf{x})^{-1} ( \partial_{\theta^{(2)}_{j_1}} \bm{\Sigma}\bm{\Sigma}_0^\top(\mathbf{x})) \bm{\Sigma}\bm{\Sigma}_0^\top(\mathbf{x})^{-1} \partial_{\theta^{(2)}_{j_2}} \bm{\Sigma}\bm{\Sigma}_0^\top(\mathbf{x}))\dif \nu_0(\mathbf{x}),
\end{align*}
uniformly in $\bm{\theta}$. 

This concludes the first part of the lemma. The second part follows from the fact that all limits are continuous in $\bm{\theta}$.
\end{proof}

\begin{proof}[Proof of Lemma \ref{lemma:LnConvergence}]
To prove the lemma, we need to compute $\bm{\lambda}_N^\mathrm{[SS]}$ from \eqref{eq:lambda}. We start with $-\frac{1}{\sqrt{N h}}\partial_{\theta^{(1)}_{i}} \mathcal{L}_N^\mathrm{[SS]}$
\begin{align*}
        &-\frac{1}{\sqrt{N h}}\partial_{\theta^{(1)}_{i}}\mathcal{L}_N^{\mathrm{[SS]}}\left(\mathbf{X}_{0:t_N}; \bm{\theta}\right) = \frac{2}{\sqrt{N }} \sum_{k=1}^{N}  \bm{\eta}_0^\top(\mathbf{X}_{t_{k-1}})  \bm{\Sigma}\bm{\Sigma}^\top(\mathbf{X}_{t_{k-1}}; \bm{\theta}^{(2)})^{-1}\partial_{\theta^{(1)}_{i}} \mathbf{F}(\mathbf{X}_{t_{k-1}}; \bm{\theta}^{(1)}) \notag\\
        &\hspace{2ex}+ 2\sqrt{\frac{h}{N }}\sum_{k=1}^{N} \partial_{\theta^{(1)}_{i}} \mathbf{F}(\mathbf{X}_{t_{k-1}}; \bm{\theta}^{(1)})^\top \bm{\Sigma}\bm{\Sigma}^\top(\mathbf{X}_{t_{k-1}}; \bm{\theta}^{(2)})^{-1}(\mathbf{F}(\mathbf{X}_{t_{k-1}}; \bm{\theta}^{(1)}_0) - \mathbf{F}(\mathbf{X}_{t_{k-1}}; \bm{\theta}^{(1)}))\notag\\
        &\hspace{2ex}+ \sqrt{\frac{h}{N }}\sum_{k=1}^{N} \bm{\eta}_0^\top(\mathbf{X}_{t_{k-1}})   D \partial_{\theta^{(1)}_{i}} \mathbf{F}(\mathbf{X}_{t_{k-1}}; \bm{\theta}^{(1)})^\top\bm{\Sigma}\bm{\Sigma}^\top(\mathbf{X}_{t_{k-1}}; \bm{\theta}^{(2)})^{-1}\bm{\eta}_0(\mathbf{X}_{t_{k-1}})  -  \sqrt{\frac{h}{N }}\sum_{k=1}^{N} \tr D \partial_{\theta^{(1)}_{i}} \mathbf{F}(\mathbf{X}_{t_k}; \bm{\theta}^{(1)})\\
        &\hspace{2ex}+\frac{1}{2}\sqrt{\frac{h}{N}}\sum_{k=1}^{N} \bm{\eta}_0^\top(\mathbf{X}_{t_{k-1}}) \bm{\Sigma}\bm{\Sigma}^\top(\mathbf{X}_{t_{k-1}}; \bm{\theta}^{(2)})^{-1}  \partial_{\theta^{(1)}_{i}}\mathbb{L}_{\bm{\theta}}\bm{\Sigma}\bm{\Sigma}^\top(\mathbf{X}_{t_{k-1}}; \bm{\theta}^{(2)})\bm{\Sigma}\bm{\Sigma}^\top(\mathbf{X}_{t_{k-1}}; \bm{\theta}^{(2)})^{-1}\bm{\eta}_0(\mathbf{X}_{t_{k-1}})\\
        &\hspace{2ex}-\frac{1}{2}\sqrt{\frac{h}{N}}\sum_{k=1}^{N}\tr   (\partial_{\theta^{(1)}_{i}}\mathbb{L}_{\bm{\theta}}\bm{\Sigma}\bm{\Sigma}^\top(\mathbf{X}_{t_{k-1}}; \bm{\theta}^{(2)})\bm{\Sigma}\bm{\Sigma}^\top(\mathbf{X}_{t_{k-1}}; \bm{\theta}^{(2)})^{-1}).
\end{align*}
Similarly, for $-\frac{1}{\sqrt{N}}\partial_{\theta^{(2)}_{j}} \mathcal{L}_N^\mathrm{[SS]}$, we get
\begin{align*}
        -\frac{1}{\sqrt{N }}\partial_{\theta^{(2)}_{j}}\mathcal{L}_N^{\mathrm{[SS]}}\left(\mathbf{X}_{0:t_N}; \bm{\theta}\right)&= -\frac{1}{\sqrt{N }} \sum_{k=1}^{N} \tr(\bm{\Sigma}\bm{\Sigma}^\top(\mathbf{X}_{t_{k-1}}; \bm{\theta}^{(2)})^{-1}\partial_{\theta^{(2)}_{j}}\bm{\Sigma}\bm{\Sigma}^\top)\\
        &\hspace{-10ex}+\frac{1}{\sqrt{N }}\sum_{k=1}^N \bm{\eta}_0^\top(\mathbf{X}_{t_{k-1}}) \bm{\Sigma}\bm{\Sigma}^\top(\mathbf{X}_{t_{k-1}}; \bm{\theta}^{(2)})^{-1}(\partial_{\theta^{(2)}_{j}}\bm{\Sigma}\bm{\Sigma}^\top) \bm{\Sigma}\bm{\Sigma}^\top(\mathbf{X}_{t_{k-1}}; \bm{\theta}^{(2)})^{-1} \bm{\eta}_0(\mathbf{X}_{t_{k-1}})\notag\\
        &\hspace{-10ex}+ 2\sqrt{\frac{h}{N }}\sum_{k=1}^{N}  \bm{\eta}_0^\top(\mathbf{X}_{t_{k-1}})  \bm{\Sigma}\bm{\Sigma}^\top(\mathbf{X}_{t_{k-1}}; \bm{\theta}^{(2)})^{-1}(\partial_{\theta^{(2)}_{j}}\bm{\Sigma}\bm{\Sigma}^\top) \bm{\Sigma}\bm{\Sigma}^\top(\mathbf{X}_{t_{k-1}}; \bm{\theta}^{(2)})^{-1}\\
        &\hspace{46ex}(\mathbf{F}(\mathbf{X}_{t_{k-1}}; \bm{\theta}^{(1)}_0) - \mathbf{F}(\mathbf{X}_{t_{k-1}}; \bm{\theta}^{(1)}))\\
        &\hspace{-10ex}+ \sum_{k=1}^{N}R(\frac{h}{\sqrt{N}}, \mathbf{X}_{t_{k-1}}).
\end{align*}
To prove the convergence in distribution of $\bm{\lambda}_N^\mathrm{[SS]}$, we introduce the following triangular arrays that arise from the previous calculations
\begin{align}
    \bm{\phi}_{N,k-1}^{\mathrm{[SS]}(i)}(\bm{\theta}_0) &\coloneqq \frac{2}{\sqrt{N }}  \bm{\eta}_0^\top(\mathbf{X}_{t_{k-1}}) \bm{\Sigma}\bm{\Sigma}^\top_0(\mathbf{X}_{t_{k-1}})^{-1}\partial_{\theta^{(1)}_{i}} \mathbf{F}_0(\mathbf{X}_{t_{k-1}})\label{eq:phiC}\\
    &+  \sqrt{\frac{h}{N }} (\tr(\bm{\eta}\bm{\eta}_0^\top(\mathbf{X}_{t_{k-1}})  D \partial_{\theta^{(1)}_{i}} \mathbf{F}_0(\mathbf{X}_{t_{k-1}})^\top\bm{\Sigma}\bm{\Sigma}^\top_0(\mathbf{X}_{t_{k-1}})^{-1}) - \tr D \partial_{\theta^{(1)}_{i}} \mathbf{F}_0(\mathbf{X}_{t_{k-1}})) \notag\\
    &+\frac{1}{2}\sqrt{\frac{h}{N}}\tr(\bm{\eta}\bm{\eta}_0^\top(\mathbf{X}_{t_{k-1}}) \bm{\Sigma}\bm{\Sigma}^\top(\mathbf{X}_{t_{k-1}}; \bm{\theta}^{(2)})^{-1}  \partial_{\theta^{(1)}_{i}}\mathbb{L}_{\bm{\theta}}\bm{\Sigma}\bm{\Sigma}^\top(\mathbf{X}_{t_{k-1}}; \bm{\theta}^{(2)})\bm{\Sigma}\bm{\Sigma}^\top(\mathbf{X}_{t_{k-1}}; \bm{\theta}^{(2)})^{-1})\notag\\
    &-\frac{1}{2}\sqrt{\frac{h}{N}}\tr   (\partial_{\theta^{(1)}_{i}}\mathbb{L}_{\bm{\theta}}\bm{\Sigma}\bm{\Sigma}^\top(\mathbf{X}_{t_{k-1}}; \bm{\theta}^{(2)})\bm{\Sigma}\bm{\Sigma}^\top(\mathbf{X}_{t_{k-1}}; \bm{\theta}^{(2)})^{-1}),\notag\\
    \bm{\rho}_{N,k-1}^{\mathrm{[SS]}(j)}(\bm{\theta}_0) &\coloneqq \frac{1}{\sqrt{N }}\bm{\eta}_0^\top(\mathbf{X}_{t_{k-1}}) \bm{\Sigma}\bm{\Sigma}^\top_0(\mathbf{X}_{t_{k-1}})^{-1}(\partial_{\theta^{(2)}_{j}}\bm{\Sigma}\bm{\Sigma}_0^\top(\mathbf{X}_{t_{k-1}})) \bm{\Sigma}\bm{\Sigma}^\top_0(\mathbf{X}_{t_{k-1}})^{-1} \bm{\eta}_0(\mathbf{X}_{t_{k-1}})\label{eq:rhoC}\\
    &- \frac{1}{\sqrt{N}}\tr(\bm{\Sigma}\bm{\Sigma}^\top_0(\mathbf{X}_{t_{k-1}})^{-1}\partial_{\theta^{(2)}_{j}}\bm{\Sigma}\bm{\Sigma}_0^\top(\mathbf{X}_{t_{k-1}})).\notag
\end{align}
Then, $\bm{\lambda}_N^\mathrm{[SS]}$ rewrites as
\begin{align}
    \bm{\lambda}_N^\mathrm{[SS]} = \sum_{k=1}^N \begin{bmatrix}
        \bm{\phi}_{N,k-1}^{\mathrm{[SS]}(1)}(\bm{\theta}_0)\\
        \vdots\\
        \bm{\phi}_{N,k-1}^{\mathrm{[SS]}(r)}(\bm{\theta}_0)\\
        \bm{\rho}_{N,k-1}^{\mathrm{[SS]}(1)}(\bm{\theta}_0)\\
        \vdots\\        
        \bm{\rho}_{N,k-1}^{\mathrm{[SS]}(s)}(\bm{\theta}_0)
    \end{bmatrix} + \frac{1}{N}\sum_{k=1}^{N}R(\sqrt{N h^2}, \mathbf{X}_{t_{k-1}}). \label{eq:LN_rewritten}
\end{align}
Thus, to establish estimators' asymptotic normality, we need an extra convergence condition $N h^2 \to 0$. This is common in literature, and it is necessary for most estimators. 

We notice that  $\bm{\phi}_{N,k-1}^{\mathrm{[SS]}(i)}(\bm{\theta}_0)$ and $ \bm{\rho}_{N,k-1}^{\mathrm{[SS]}(j)}(\bm{\theta}_0)$ are centered conditionally to $\mathcal{F}_{t_{k-1}}$ and adapted to the filtration $\mathcal{F}_{t_k}$. Thus, the proof follows directly by applying Proposition 3.1 in \citep{CRIMALDI2005571}. 
\end{proof}

\end{document}